\documentclass[11pt]{article}

\newcommand{\jpa}{J. Phys. A~}

\newcommand{\njp}{New. J. Phys.~}

\newcommand{\prl}{Phys. Rev. Lett.~}
\newcommand{\pra}{Phys. Rev. A~}
\newcommand{\pla}{Phys. Lett. A~}
\newcommand{\rmp}{Rev. Math. Phys.~}

\usepackage{young}
\usepackage{color}
\usepackage[latin1]{inputenc}
\usepackage{appendix}
\usepackage{epstopdf}
\usepackage{color}
\usepackage{subfigure}
\usepackage[T1]{fontenc}
\usepackage[sc]{mathpazo}
\usepackage{amsmath}
\usepackage{amssymb}
\usepackage{graphicx,color}
\usepackage{framed}
\usepackage{multirow}
\usepackage{enumerate}
\usepackage{amsthm}
\usepackage{amsfonts,mathrsfs}
\usepackage{geometry} 
\usepackage{eepic}
\usepackage{ifthen}
\usepackage[vcentermath]{youngtab}
\usepackage[unicode=true,pdfusetitle, bookmarks=true,bookmarksnumbered=false,bookmarksopen=false, breaklinks=false,pdfborder={0 0 0},backref=false,colorlinks=false] {hyperref}
\hypersetup{
colorlinks,linkcolor=myurlcolor,citecolor=myurlcolor,urlcolor=myurlcolor}
\definecolor{myurlcolor}{rgb}{0,0,0.7}

\newcommand{\blue}{\textcolor{blue}}

\newcommand{\proj}[1]{| #1\rangle\!\langle #1 |}

\usepackage{hyperref}
\hypersetup{pdfpagemode=UseNone}

\newcommand{\tinyspace}{\mspace{1mu}}

\newcommand{\op}[1]{\operatorname{#1}}

\newcommand{\abs}[1]{\left\lvert\tinyspace #1 \tinyspace\right\rvert}

\newcommand{\norm}[1]{\left\lVert\tinyspace #1 \tinyspace\right\rVert}

\renewcommand{\det}{\operatorname{det}}
\renewcommand{\t}{{\scriptscriptstyle\mathsf{T}}}

\newcommand{\setft}[1]{\mathrm{#1}}

\newcommand{\density}[1]{\setft{D}\left(#1\right)}

\newcommand{\spn}{\op{span}}
\newcommand{\sign}{\op{sign}}

\def\GL{\mathsf{GL}}

\def\SO{\mathsf{SO}}
\def\SU{\mathsf{SU}}
\def\LU{\mathsf{LU}}

\def \diag {\mathrm{diag}}

\def\complex{\mathbb{C}}
\def\real{\mathbb{R}}
\def\natural{\mathbb{N}}

\def\I{\mathbb{1}}

\def\zero{\mathbf{0}}

\def\bsigma{\boldsymbol{\sigma}}

\newenvironment{mylist}[1]{\begin{list}{}{
    \setlength{\leftmargin}{#1}
    \setlength{\rightmargin}{0mm}
    \setlength{\labelsep}{2mm}
    \setlength{\labelwidth}{8mm}
    \setlength{\itemsep}{0mm}}}
    {\end{list}}

\def\ot{\otimes}

\newcommand{\inner}[2]{\langle #1 , #2\rangle}
\newcommand{\iinner}[2]{\langle #1 | #2\rangle}
\newcommand{\out}[2]{| #1\rangle\langle #2 |}
\newcommand{\Inner}[2]{\big\langle #1 , #2\big\rangle}
\newcommand{\Innerm}[3]{\big\langle #1 \big| #2 \big| #3 \big\rangle}
\newcommand{\defeq}{\stackrel{\smash{\textnormal{\tiny def}}}{=}}

\newcommand{\Herm}{\mathrm{Herm}}

\newcommand{\End}{\mathrm{End}}

\newcommand{\Pa}[1]{\left(#1\right)}

\newcommand{\Br}[1]{\left[#1\right]}
\newcommand{\set}[1]{\{#1\}}
\newcommand{\Set}[1]{\left\{#1\right\}}

\newcommand{\bra}[1]{\langle#1|}

\newcommand{\ket}[1]{|#1\rangle}

\DeclareMathOperator{\trace}{Tr}

\newcommand{\Ptr}[2]{\trace_{#1}\Pa{#2}}

\newcommand{\Tr}[1]{\Ptr{}{#1}}

\newcommand{\Abs}[1]{\left|\tinyspace#1\tinyspace\right|}

\def\cA{\mathcal{A}}\def\cB{\mathcal{B}}
\def\cF{\mathcal{F}}

\def\bP{\mathbf{P}}\def\bQ{\mathbf{Q}}

\def\bsA{\boldsymbol{A}}\def\bsB{\boldsymbol{B}}\def\bsC{\boldsymbol{C}}
\def\bsF{\boldsymbol{F}}
\def\bsL{\boldsymbol{L}}\def\bsM{\boldsymbol{M}}\def\bsN{\boldsymbol{N}}
\def\bsP{\boldsymbol{P}}\def\bsQ{\boldsymbol{Q}}\def\bsR{\boldsymbol{R}}\def\bsT{\boldsymbol{T}}
\def\bsU{\boldsymbol{U}}\def\bsX{\boldsymbol{X}}\def\bsY{\boldsymbol{Y}}
\def\bsZ{\boldsymbol{Z}}

\def\bsa{\boldsymbol{a}}\def\bsb{\boldsymbol{b}}\def\bsd{\boldsymbol{d}}\def\bse{\boldsymbol{e}}

\def\bsn{\boldsymbol{n}}
\def\bsp{\boldsymbol{p}}\def\bsr{\boldsymbol{r}}\def\bss{\boldsymbol{s}}\def\bst{\boldsymbol{t}}
\def\bsu{\boldsymbol{u}}\def\bsv{\boldsymbol{v}}\def\bsw{\boldsymbol{w}}\def\bsx{\boldsymbol{x}}\def\bsy{\boldsymbol{y}}

\def\L{\textsf{L}}\def\O{\textsf{O}}

\def\U{\textsf{U}}

\newtheorem{thrm}{Theorem}
\newtheorem{lem}{Lemma}
\newtheorem{prop}{Proposition}
\newtheorem{cor}{Corollary}
\theoremstyle{definition}
\newtheorem{definition}{Definition}

\newtheorem{exam}{Example}

\newcounter{questionnumber}

\begin{document}

\title{Bargmann-invariant framework for local unitary equivalence and entanglement}

\author{\blue{Lin Zhang}$^1$\footnote{E-mail:
godyalin@163.com},\quad \blue{Bing Xie}$^1$,\quad \blue{Yuanhong
Tao}$^2$
\\
  {\it\small $^1$School of Science, Hangzhou Dianzi University, Hangzhou 310018, People's Republic of China}\\
  {\it\small $^2$College of Science, Zhejiang University of Science and Technology, Hangzhou, Zhejiang 310023, PR China}
  }
\date{\today}
\maketitle

\begin{abstract}
Research on quantum states often focuses on the correlation between
nonlocal effects and local unitary invariants, among which local
unitary equivalence plays a significant role in quantum state
classification and resource theories. This paper focuses on the
local unitary equivalence of multipartite quantum states in quantum
information theory, aiming to determine a complete set of invariants
that identify their local unitary orbits; these invariants are
crucial for deriving polynomial invariants and describing the
physical properties preserved under local unitary transformations.
The study deeply explores the characterization of local unitary
equivalence and the method of detecting entanglement using local
unitary Bargmann invariants. Taking two-qubit systems as an example,
it verifies the measurability of invariants that determine
equivalence and establishes a connection between Makhlin fundamental
invariants (a complete set of 18 local unitary invariants for
two-qubit states) and local unitary Bargmann invariants. These
Bargmann invariants, related to the traces of products of density
operators and marginal states, can be measured through cycle tests
(an extended form of SWAP tests).
\end{abstract}
\newpage
\tableofcontents
\newpage

\section{Introduction}

In the rapidly evolving field of quantum information science,
understanding and manipulating quantum states is of paramount
importance. Among the myriad phenomena that quantum mechanics
offers, local unitary equivalence and entanglement stand out as
fundamental yet intricate concepts. Local unitary equivalence, which
posits that certain quantum states are indistinguishable under local
operations and classical communication (LOCC), lies at the heart of
quantum state classification and resource theories. Entanglement, on
the other hand, serves as a cornerstone for quantum computing
\cite{Jozsa2003}, quantum cryptography \cite{Portmann2022}, and
various quantum communication protocols
\cite{Bennett1992,Bennett1993,Bouwmeester1997}, underscoring its
pivotal role in harnessing the power of quantum mechanics.

The local unitary equivalence
\cite{Makhlin2002,Kraus2010,Zhou2012,Jing2015,Martins2015}, defined
through local unitary transformations, holds significant importance
in quantum information science because the importance of local
unitary transformations lies in their crucial roles in quantum state
classification, manipulation, and algorithm design: In quantum state
classification \cite{Gour2013} it acts as a core tool, enabling
judgments of quantum state equivalence (identifying which states can
be inter-converted via local unitary transformations) while
preserving entanglement properties. In quantum state manipulation
\cite{Gour2011}, it is indispensable, allowing precise control over
local properties of quantum states without altering global system
characteristics -- acting on specific subsystems to modify their
states while preserving the inner product and norm of the overall
system state, thus finding wide use in state preparation,
manipulation, and measurement. In quantum algorithm design
\cite{Vartiainen2005} it is equally critical, serving as a basic
operational unit for constructing complex algorithms and optimizing
performance (e.g., adjusting search space structures in quantum
search algorithms to enhance efficiency).

The characterization of local unitary equivalence and the detection
of entanglement are crucial for advancing our understanding and
applications of quantum systems. Despite significant progress, these
tasks remain challenging due to the complex nature of quantum states
and the high-dimensional spaces they inhabit. The complex
interaction between local unitary transformations and global quantum
characteristics requires a refined method to distinguish equivalent
states and efficiently recognize entangled states.

Bargmann invariants, fundamental local unitary invariants of central
importance in quantum information, are associated with protocols
like quantum fingerprinting \cite{Buhrman2001} and concepts
including geometric phases \cite{Simon1993,Mukunda2003}. Their
applications span Kirkwood-Dirac quasiprobabilities, quantum
imaginarity witnesses
\cite{Kirkwood1933,Dirac1945,Bamber2014,Fernandes2024}, and
multipartite entanglement detection. Also termed multivariate traces
\cite{Oszmaniec2024}, they are amenable to estimation via
constant-depth circuits \cite{Quek2024}, ensuring compatibility with
near-term hardware and experimental feasibility. Acting as ``quantum
fingerprints", they determine state equivalence and enable
classification of high-dimensional multipartite states. Critically,
they capture nonlocal structures to detect entanglement and offer
multidimensional insights into quantum phenomena.

This paper develops a comprehensive framework using local unitary
Bargmann invariants to characterize multipartite quantum state
equivalence and detect entanglement. We integrate theoretical
foundations with algorithmic implementations to: (i) establish
precise conditions for local unitary equivalence, and (ii)
propose--for the first time--an entanglement detection protocol
based on Bargmann invariants. This approach advances methodologies
for analyzing complex quantum systems.

The paper is structured as follows: Firstly, we review fundamental
concepts of local unitary equivalence, establishing the groundwork
for subsequent analysis. Then we explores the theoretical
foundations of local unitary transformations and their role in
quantum state classification. After that, we develops entanglement
detection criteria based on local unitary-invariant Bargmann
invariants, critically examining their advantages and limitations.
Finally, we introduces novel methods and algorithms addressing
current challenges in characterizing local unitary equivalence and
detecting entanglement, while outlining promising research
directions. By advancing methodologies for these fundamental quantum
phenomena, this work aims to catalyze new developments in quantum
information science and technology. In the Appendixes we detail the
development process and key findings leading to the main
conclusions. When deriving these main results, we present some
essential tools that facilitate the obtainment of additional
findings. For instance, we establish a rigorous relationship
concerning the conversion between the Makhlin fundamental invariants
and LU Bargmann invariants. With these preparations, we can
calculate arbitrary locally unitary Bargmann invariants
$\Tr{\rho_{i_1}\cdots\rho_{i_N}}$, where each $\rho_{i_k}$ is from
the set $\set{\rho_{AB},\rho_A\ot\I_B,\I_A\ot\rho_B}$ for any
two-qubit state $\rho_{AB}$, up to ignoring dimensional factors.

\section{Bargmann invariant of a tuple of quantum states}

Before proceeding, let us fix notations used in this paper. Given
two tuples of $N$ states $\Psi=(\rho_1,\ldots,\rho_K)$ and
$\Psi'=(\rho'_1,\ldots,\rho'_K)$ acting on Hilbert space
$\complex^d$, if there exists a unitary $\bsU\in\U(d)$, the unitary
group acting on $\complex^d$, such that
$\rho'_i=\bsU\rho_i\bsU^\dagger$ for each $i=1,2,\ldots,K$, we say
$\Psi$ and $\Psi'$ are \emph{unitarily equivalent}. If there exists
a set of invariant properties  allows us to decide whether two
tuples of states are unitarily equivalent, this set is said to be
complete.

Consider a tuple of $K$ pure/mixed quantum states
$\Psi=(\rho_1,\ldots,\rho_K)$, where states $\rho_i$'s act on the
same underlying Hilbert space. The \emph{Bargmann invariant} (aka
multivariate traces \cite{Oszmaniec2024,Quek2024}) of this tuple of
states is defined as
\begin{eqnarray}
\Delta_{12\cdots K}(\Psi)=\Tr{\rho_1\rho_2\cdots\rho_K}.
\end{eqnarray}

Bargmann invariants can be used to describe the unitarily
equivalence between tuples of states. In fact, we have already known
the following result \cite{Procesi1976}: For two tuples of mixed
states on $\complex^d$, $\Psi=(\rho_1,\ldots,\rho_K)$ and
$\Psi'=(\rho'_1,\ldots,\rho'_K)$, both $\Psi$ and $\Psi'$ are
unitarily equivalent if and only if, for every $m\in\natural$ and
for every sequence $i_1,i_2,\ldots,i_m$ of numbers from
$\set{1,\ldots,K}$, the corresponding Bargmann invariants of degree
$m$ agree
\begin{eqnarray}
\Tr{\rho_{i_1}\rho_{i_2}\cdots\rho_{i_m}}=\Tr{\rho'_{i_1}\rho'_{i_2}\cdots\rho'_{i_m}}.
\end{eqnarray}

Recently, quantum circuits such as cycle test was introduced, which
enable the \emph{direct measurement} of complete sets of Bargmann
invariants for a tuple of quantum states \cite{Oszmaniec2024}.
Motivated by this result, we will investigate the locally unitary
equivalence of tuples of multipartite states using locally unitary
Bargmann invariants.

\section{Local unitary equivalence of multipartite states}

The same paradigm in the last section motivated the usage of
invariant polynomials in the context of classification of
entanglement classes subject to local unitary transformation. Let
$V:=\Herm(\complex^{d_1}\ot\cdots\ot \complex^{d_N})$, the Hermitian
matrices acting on the tensor space, and denote the local unitary
group by $\LU(\bsd)\equiv\U(d_1)\ot\cdots\ot\U(d_N)$, where
$\bsd:=(d_1,\ldots,d_N)$. The action of $\LU(\bsd)$ on $V$ is
defined by conjugation as $\tau_g(\bsX) = g \bsX g^{\dagger}$ for
all $g \in \LU(\bsd)$ and $\bsX \in V$. In fact, given two tuples of
multipartite states on $\complex^{d_1}\ot\cdots\ot\complex^{d_N}$,
$\Psi=(\rho_1,\ldots,\rho_K)$ and $\Psi'=(\rho'_1,\ldots,\rho'_K)$,
they are locally unitarily (LU) equivalent in the sense that
$\rho'_i =g\rho_i g^\dagger$ for all $i=1,\ldots,K$ and some
$g\in\LU(\bsd)$. That is, there exist a collection of $N$ unitary
operators $\bsU_j\in\U(d_j)(j=1,\ldots,N)$ such that
$g=\bsU_1\ot\cdots\ot\bsU_N$ and
\begin{eqnarray}
\rho'_i =
(\bsU_1\ot\cdots\ot\bsU_N)\rho_i(\bsU_1\ot\cdots\ot\bsU_N)^\dagger
\end{eqnarray}
for each $i=1,\ldots,K$. Clearly, when $K=1$, this problem is
reduced to a well-known locally unitary equivalence of two
multipartite states. Henceforth, we characterize local unitary
equivalence between two multipartite states through measurable
quantities expressible as linear combinations of local unitary
Bargmann invariants. The following result is essentially due to
Grassl \cite{Grassl1998}. For the reader's convenience, we provide
an independent proof here. For the detailed development, please see
Appendix~\ref{app:1}.
\begin{prop}\label{th:LUequiv}
For any two $N$-partite states $\rho$ and $\sigma$ acting on
$\complex^{d_1}\ot\cdots\ot \complex^{d_N}$, they are LU equivalent,
i.e., $\sigma=g\rho g^\dagger$ for some $g\in \LU(\bsd)$ for
$\bsd=(d_1,\ldots,d_N)$, if and only if, for arbitrary positive
integer $n$, it holds that
\begin{eqnarray}\label{eq:tracequantities}
\Tr{\sigma^{\ot n}\bP_{\bsd,n}(\boldsymbol{\pi})} = \Tr{\rho^{\ot
n}\bP_{\bsd,n}(\boldsymbol{\pi})},
\end{eqnarray}
where the meaning of $\bP_{\bsd,n}(\boldsymbol{\pi})$ will be
explained immediately in Eqs.~\eqref{eq:per1} and \eqref{eq:per2}
for all $\boldsymbol{\pi}:=(\pi_1,\ldots,\pi_N)\in S^N_n$, the
Cartesian product of $N$ copies of the permutation group of $n$
distinct elements.
\end{prop}
The sketch of proof is described here. Clearly $\rho\in V$, then
$g=\bsU_1\ot\cdots\ot\bsU_N\in \L\U(\bsd)$ acts on $\rho$ via
$\tau_g\rho =g\rho g^\dagger$. The space of all real polynomials on
$V$ is denoted by $\real[V]$. We will denote by $\real[V]_n$ the
space of real homogenous polynomials on $V$ of degree $n$. We have
already known that each homogeneous polynomials of degree $n$ are
mappings of the form $p(\bsX)=\Inner{\tilde\bsp}{\bsX^{\ot n}}$,
where $\bsX\in V$ and $\tilde\bsp\in V^{\ot n}$. Thus $\tilde \bsp$
defines an $\LU(\bsd)$-invariant polynomial $p\in
\real[V]^{\L\U(\bsd)}_n$ if and only if $\tau^{\ot n}_g\tilde
\bsp=\tilde \bsp$ for all $g\in \L\U(\bsd)$. Now for both $\rho$ and
$\sigma$ satisfying $\sigma=g\rho g^\dagger$ for some $g\in
\L\U(\bsd)$ if and only if $p(\rho)=p(\sigma)$ for $\forall p\in
\real[V]^{\L\U(\bsd)}$ \cite{Vrana2012}. By virtue of the above this
is equivalent to demanding that for all $g\in \L\U(\bsd)$
$$
\tilde \bsp=\tau^{\ot n}_g\tilde \bsp=\bQ_{\bsd,n}(\bar g)\tilde
\bsp\bQ_{\bsd,n}(\bar g),
$$
where $\bQ_{\bsd,n}(\bar g):=\bsU^{\ot n}_1\ot\cdots\ot\bsU^{\ot
n}_N$ for $\bar g:=(\bsU_1,\ldots,\bsU_N)$. This amounts to
requiring $[\bQ_{\bsd,n}(\bar g),\tilde\bsp]=0$ implying that
$\tilde\bsp\in\widetilde\cB'=\widetilde\cA$ by the generalized
Schur-Weyl duality. Thus $\tilde \bsp$ can be expanded into a linear
combination of $\bP_{\bsd,n}(\boldsymbol{\pi})$'s for
$\boldsymbol{\pi}=(\pi_1,\ldots,\pi_N)\in S^N_n$, where
\begin{eqnarray}\label{eq:per1}
\bP_{\bsd,n}(\boldsymbol{\pi}):=\bP_{d_1,n}(\pi_1)\ot\cdots\ot\bP_{d_N,n}(\pi_N)
\end{eqnarray}
for each permutation $\pi_j\in S_n(j=1,\ldots,N)$. Here, for
$(d,\pi)\in\set{(d_j,\pi_j):j=1,\ldots,N}$, $\bP_{d,n}(\pi)$ acting
on $(\complex^d)^{\ot n}$ via the action on computational basis
vectors is defined by
\begin{eqnarray}\label{eq:per2}
\bP_{d,n}(\pi)\ket{i_1\cdots i_n}:=\ket{i_{\pi^{-1}(1)}\cdots
i_{\pi^{-1}(n)}}.
\end{eqnarray}
By the generalized Schur-Weyl duality \cite{Brauer1937}, it is
easily seen that $\sigma=g\rho g^\dagger$ if and only if
$\Tr{\sigma^{\ot n}\bP_{\bsd,n}(\boldsymbol{\pi})} = \Tr{\rho^{\ot
n}\bP_{\bsd,n}(\boldsymbol{\pi})}$, where $n=1,2,\ldots$ and
$\boldsymbol{\pi}\in S^N_n$.

It will be seen that these quantities involved in
Eq.~\eqref{eq:tracequantities} can be shown to be a linear
combination of local unitary Bargmann invariants
$\Tr{\rho_{i_1}\cdots \rho_{i_n}}$ (when restricted to two-qubit
system), where each $\rho_{i_k}$'s is taken from the sequence of
states $\set{\rho_\Lambda:\Lambda\subseteq \set{1,\ldots,N}}$, where
$\rho_\Lambda=\Ptr{\bar\Lambda}{\rho}$, where
$\bar\Lambda:=\set{1,\ldots,N}\backslash\Lambda$. Due to the
measurability of Bargmann invariants, the above result in fact leads
an operational test for LU equivalence.

In order to derive a further result, let us focus on the special
bipartite case. Given two bipartite states $\rho_{AB}$ and
$\rho'_{AB}$ on $\complex^m\ot\complex^n(m=n=2)$, let
$\Psi=(\rho_{AB},\rho_A\ot\I_B,\I_A\ot\rho_B)$, where $\I_X(X=A,B)$
is the identity operator, and
$\Psi'=(\rho'_{AB},\rho'_A\ot\I_B,\I_A\ot\rho'_B)$, we will study
the locally unitary equivalence of two tuples $\Psi$ and $\Psi'$. It
is easily seen that $\rho_{AB}$ is LU equivalent to $\rho'_{AB}$ if
and only if $\Psi$ is LU equivalent to $\Psi'$.

Although a complete set of LU invariants of two-qubit states is
given already by Makhlin in 2002 \cite{Makhlin2002}, we would like
here to work out a complete set of LU invariants of two-qubit states
in terms of Bargmann invariants which are measurable quantities, of
interest to experimentalists.

\begin{thrm}\label{th:Bargmann-generators}
Given any two-qubit state
$\rho_{AB}\in\density{\complex^2\ot\complex^2}$. Denote
$\bsX_0=\rho_{AB}, \bsX_1=\rho_A\ot\I_B$, and
$\bsX_2=\I_A\ot\rho_B$. The set comprising of 18 local unitary
Bargmann invariants $B_k(k=1,\ldots,18)$ provides a complete
description of nonlocal properties of the two-qubit state
$\rho_{AB}$, where the meanings of $B_k$'s are given below:
\begin{eqnarray}\label{eq:18Bargmann-generators}
\begin{cases}
B_1:=\Tr{\bsX_0\bsX_1},B_2:=\Tr{\bsX_0\bsX_2},B_3:=\Tr{\bsX_0\bsX_1\bsX_2},B_4:=\Tr{\bsX^2_0},\\
B_5:=\Tr{\bsX^2_0\bsX_1\bsX_2},B_6:=\Tr{\bsX^3_0},B_7:=\Tr{\bsX^3_0\bsX_1},B_8:=\Tr{\bsX^3_0\bsX_2},\\
B_9:=\Tr{\bsX^3_0\bsX_1\bsX_2},B_{10}:=\Tr{\bsX^4_0},B_{11}:=\Tr{\bsX^2_0\bsX_1\bsX^2_0\bsX_1},B_{12}:=\Tr{\bsX^2_0\bsX_2\bsX^2_0\bsX_2},\\
B_{13}:=\Tr{\bsX_0\bsX_1\bsX_2\bsX^2_0\bsX_1},B_{14}:=\Tr{\bsX_0\bsX_1\bsX_2\bsX^2_0\bsX_2},
B_{15}:=\Tr{\bsX_0\bsX_1\bsX_2\bsX^3_0\bsX_1},\\
B_{16}:=\Tr{\bsX_0\bsX_1\bsX_2\bsX^3_0\bsX_2},
B_{17}:=\Tr{\bsX_0\bsX_1\bsX^2_0\bsX_1\bsX^3_0\bsX_1},
B_{18}:=\Tr{\bsX_0\bsX_2\bsX^2_0\bsX_2\bsX^3_0\bsX_2}.
\end{cases}
\end{eqnarray}
In other words, two states of a two-qubit system are LU equivalent
if and only if both states have equal values of all 18 LU Bargmann
invariants.
\end{thrm}
The specific expressions for all Makhlin invariants are analytically
expressed by using Bargmann invariants $B_k$'s given in
Lemma~\ref{lem:LB} of Appendix~\ref{app:2}. We should remark there
that the problem of a complete set of generators for the invariant
polynomial ring and the problem of finding a complete set to
distinguish local unitary orbits are not identical problems. For
example, in the case of a two-qubit system, the invariant polynomial
ring has 21 generators \cite{Grassl1998}, while a complete set for
distinguishing local unitary orbits consists of 18 elements
\cite{Makhlin2002}.

The proof of this theorem can be finished by finding analytical
relations between Makhlin invariants and Bargmann invariants
$B_k$'s. In other words, as generators of a complete set of LU
invariants, $B_k$'s are more important because $B_k$'s are
measurable by the recent proposed quantum circuit, the so-called
cycle test. Therefore, we can determine whether two \emph{unknown}
two-qubit states are LU equivalent if and only if they have the same
values on the 18 Bargmann generators by measurement. This
classification of states based on their local properties provides a
powerful toolkit for analyzing their global, non-local
characteristics. In the following section, we leverage these
invariants to address a central problem in quantum information:
determining whether a given state is entangled.

\section{Entanglement criterion via LU Bargmann invariants}

Having established a framework for classifying states under local
operations, we now turn to the problem of entanglement detection.
Since entanglement must be invariant under local unitary
transformations, the LU invariants discussed in the previous section
are natural candidates for constructing entanglement criteria.

From the connection between Makhlin invariants and Bargmann
invariants, we will get a physical and operational criterion in
entanglement detection. In fact, we get the following:

\begin{thrm}\label{th:ent-test}
A two-qubit state $\rho_{AB}$ is entangled if and only if the
following inequality holds true:
\begin{eqnarray}\label{eq:ent-test}
6(B_1+B_2-B_1B_2-B_4-B_{10})+12(B_5-B_3)+3B^2_4+4B_6<1,
\end{eqnarray}
where the meanings of $B_k$'s here are taken from
Eq.~\eqref{eq:18Bargmann-generators}. Explicitly,
Eq.~\eqref{eq:ent-test} can be equivalently rewritten as
\begin{eqnarray}\label{eq:Bargmannentvssep}
&&6\Br{\Tr{\rho^2_A}+\Tr{\rho^2_B}-\Tr{\rho^2_A}\Tr{\rho^2_B}-\Tr{\rho^2_{AB}}-\Tr{\rho^4_{AB}}}\notag\\
&&+12\Br{\Tr{\rho^2_{AB}(\rho_A\ot\rho_B)}-\Tr{\rho_{AB}(\rho_A\ot\rho_B)}}+3\Br{\Tr{\rho^2_{AB}}}^2+4\Tr{\rho^3_{AB}}<1.
\end{eqnarray}
\end{thrm}
It is already known that the fundamental equivalence is that, for
any two-qubit system, the necessary and sufficient criterion for
entanglement is the Peres-Horodecki (abbr. PPT) criterion
\cite{Horodecki1996}. Similarly, Theorem~\ref{th:ent-test} is also
equivalent to the PPT criterion. However, the significance of it
lies in its independence from any other observables; determining
whether a state is entangled requires only implementing a quantum
circuit, say one in \cite{Quek2024}, to measure 7 locally unitary
Bargmann invariants. The proof of the desired inequality is put in
Appendix~\ref{app:3}.

\section{Discussion}

In \cite{Yu2003}, the authors proposed a test for entanglement of
two-qubit states: A two-qubit $\rho$ is separable if and only if the
following inequality holds for \emph{all} sets of observables
$\bsA_i=\bsa_i\cdot\boldsymbol{\sigma}$ and
$\bsB_i=\bsb_i\cdot\boldsymbol{\sigma}$, where $i=1,2,3$, with the
same orientation:
\begin{eqnarray*}
\sqrt{\langle\bsA_1\ot\bsB_1+\bsA_2\ot\bsB_2\rangle^2_\rho+
\langle\bsA_3\ot\I+\I\ot\bsB_3\rangle^2_\rho}\leqslant
1+\langle\bsA_3\ot\bsB_3\rangle_\rho.
\end{eqnarray*}
We see from this criterion that, for an \emph{unknown} two-qubit, it
is hard to determine the separability of such state in practice
because one has to check their inequality for all sets of local
testing observables being complementary. The advantage of our
criterion in Eq.~\eqref{eq:ent-test} or
Eq.~\eqref{eq:Bargmannentvssep} indicates that in order to determine
separability/entanglement in an \emph{unknown} two-qubit state, it
suffices to measure only 7 locally unitarily Bargmann invariants for
such two-qubit state with the help of a quantum circuit of constant
depth \cite{Quek2024}.

Our approach to the entanglement criterion for two-qubit states can
be extended to another specialized composite quantum system, namely,
the qubit-qutrit system. However, the computational complexity of
determining a complete set of generators distinguishing locally
unitary orbits for qubit-qutrit states is tremendous, due to the
increased dimensionality of the Hilbert space
$\complex^2\ot\complex^3$, which results in a more complex local
unitary group $\SU(2)\times \SU(3)$. This growth in complexity leads
to a rapid expansion in the number of algebraically independent
polynomial invariants, rendering their computation prohibitively
expensive. Thus, identifying a minimal and sufficient set of these
invariants for entanglement detection---a task we accomplished for
two-qubit---becomes a non-trivial undertaking in hybrid-dimensional
spaces. Moreover, when employing LU Bargmann invariants to detect
entanglement, the required number of inequalities is no longer one
(i.e., multiple inequalities are needed). This is because, unlike
the two-qubit case where although the negativity of the constant
term (i.e., the determinant of the partial-transposed density
matrix) in the characteristic polynomial of the partial-transposed
density matrix serves as a necessary and sufficient entanglement
criterion, it becomes insufficient for the qubit-qutrit system where
a more complex set of inequalities for complete detection is needed.
Our future work will extend this approach by generalizing the cross
product from three-dimensional Euclidean space to higher dimensions
and formulating a product rule for two-qudit observables. This will
allow us to formalize the relationship between LU Bargmann
invariants and the triples (i.e., those triples consists of two
generalized Bloch vectors and correlation matrix) in the generalized
Bloch representation \cite{Kimura2003}.

\section{Conclusion}

In this work, we have explored the local unitary equivalence of
multipartite states using Bargmann invariants. We identified a
complete set of 18 Bargmann generators distinguishing local unitary
(LU) invariants for two-qubit states. Building on this foundation,
we propose a method to characterize entanglement in unknown
two-qubit states by measuring a subset of seven out of these 18
Bargmann generators. Our approach can be extended to
higher-dimensional state spaces. Our findings also inspire novel
experimental designs to test entanglement in unknown quantum states.
In future research, we plan to investigate the relationships between
the moments of the probability distribution of random measurements
\cite{Wyderka2023} and Bargmann invariants.

\subsubsection*{Acknowledgments}

This research is supported by Zhejiang Provincial Natural Science
Foundation of China under Grant No. LZ23A010005 and by NSFC under
Grant No.11971140. We are deeply grateful for the significant and
crucial comments provided by the anonymous reviewers, as they have
been instrumental in enhancing the quality of our paper. We extend
our sincere gratitude to Shao-Ming Fei, Ming Li, and Ming-Jing Zhao
for their insightful and valuable comments. LZ would like to thank
Felix Huber and Markus Grassl for pointing out the omitted
references in the arXiv version of this paper, and Markus Grassl
again for his helpful comments on Makhlin's invariants.



\newpage
\appendix
\appendixpage
\addappheadtotoc

The present appendix details the development process and key
findings leading to the main conclusions. When deriving these main
results, we present some essential tools that facilitate the
obtainment of additional findings. For instance, we establish a
rigorous relationship concerning the conversion between Makhlin's
fundamental invariants and LU Bargmann's invariants. With these
preparations, we can calculate arbitrary locally unitary Bargmann
invariants $\Tr{\rho_{i_1}\cdots\rho_{i_N}}$, where each
$\rho_{i_k}$ is from the set
$\set{\rho_{AB},\rho_A\ot\I_B,\I_A\ot\rho_B}$ for any two-qubit
state $\rho_{AB}$, up to ignoring dimensional factors.

\section{Proof of Proposition~\ref{th:LUequiv}}\label{app:1}

For the proof of Proposition~\ref{th:LUequiv} in the main text, we
used a lot of tools which cannot be explained in detail within the
confines of that proposition's discussion. Now in this section, we
will present a more comprehensive and detailed exploration of these
tools, providing the necessary background, definitions, and
explanations to fully understand their application in the proof.
This deeper dive will not only clarify the intricacies of the proof
but also enhance the reader's grasp of the underlying mathematical
concepts and techniques. By doing so, we aim to make the proof of
Proposition~\ref{th:LUequiv} more accessible and insightful for a
broader audience.

\subsection{Invariant theory}

Let $K$ be a compact group and let
\begin{eqnarray}
\Pi:K\ni g\mapsto \Pi_g\in\GL(V)
\end{eqnarray}
be a representation of $K$ in a finite dimensional real vector space
$V$. Since $K$ is compact, we can assume that $\Pi_g$ is an
orthogonal transformation. That is,
\begin{eqnarray}
\Pi:K\ni g\mapsto \Pi_g\in\O(V).
\end{eqnarray}
The space of all real polynomials on $V$ is denoted by $\real[V]$.
We will denote by $\real[V]_n$ the space of real homogeneous
polynomials on $V$ of degree $n$. Homogeneous polynomials of degree
$n$ are mappings of the form:
\begin{eqnarray}
p(\bsv)=\Inner{\tilde \bsp}{\bsv^{\ot n}}
\end{eqnarray}
where $\Inner{\cdot}{\cdot}$ is the $K$-invariant inner product in
$V^{\ot n}$ (induced by the inner product on $V$), and $\tilde
\bsp\in V^{\ot n}$ is a tensor encoding the polynomial $p$.

Invariant homogeneous polynomials of degree $n$ are polynomials that
must satisfy
\begin{eqnarray}
p(\Pi_{g^{-1}}\bsv)=p(\bsv)
\end{eqnarray}
for every $\bsv\in V$ and $g\in K$. This is equivalent to
\begin{eqnarray}
\Inner{\tilde \bsp}{\bsv^{\ot n}} = \Inner{\tilde
\bsp}{(\Pi_{g^{-1}}\bsv)^{\ot n}} = \Inner{\tilde \bsp}{\Pi^{\ot
n}_{g^{-1}}\bsv^{\ot n}} = \Inner{\Pi^{\ot n}_g\tilde
\bsp}{\bsv^{\ot n}},
\end{eqnarray}
which implies
\begin{eqnarray}
\Pi^{\ot n}_g\tilde \bsp = \tilde \bsp
\end{eqnarray}
for every $g\in K$.

Denote the set of all $K$-invariant polynomials by $\real[V]^K$. It
is well known result in invariant theory that in the case of compact
groups we can use invariant polynomials in $\real[V]^K$ to decide
about equivalence of elements of $V$ under the action of $K$.
\begin{prop}[\cite{Vrana2012}]\label{prop:vrana}
For $\bsu,\bsv\in V$, we have $\bsv=\Pi_g\bsu$, for some $g\in K$ if
and only if for every invariant polynomial $p\in\real[V]^K$, we have
$p(\bsv)=p(\bsu)$.
\end{prop}
Because every polynomial can be decomposed into the direct sum of
homogeneous polynomials, this implies
$\real[V]^K=\oplus^\infty_{n=1}\real[V]^K_n$. Then the above
Proposition~\ref{prop:vrana} can be restated as
\begin{prop}\label{prop:lin}
For $\bsu,\bsv\in V$, we have $\bsv=\Pi_g\bsu$, for some $g\in K$ if
and only if for every $K$-invariant homogeneous polynomial $p_n$ of
degree $n$, we have $p_n(\bsv)=p_n(\bsu)$, where $n=1,2,\ldots$
\end{prop}

\subsection{The generalized Schur-Weyl duality}

Consider a system of $n$ qudits, acting on $(\complex^d)^{\ot n}$
each with a standard local computational basis
$\set{\ket{i},i=1,\ldots,d}$. The Schur-Weyl duality relates
transforms on the system performed by local $d$-dimensional unitary
operations to those performed by permutation of the qudits. Recall
that the symmetric group $S_n$ is the group of all permutations of
$n$ objects. This group is naturally represented in our system by
\begin{eqnarray}\label{eq:P}
\bP_{d,n}(\pi)\ket{i_1\cdots i_n} := \ket{i_{\pi^{-1}(1)}\cdots
i_{\pi^{-1}(n)}},
\end{eqnarray}
where $\pi\in S_n$ is a permutation and $\ket{i_1\cdots i_n}$ is
shorthand for $\ket{i_1}\ot\cdots\ot\ket{i_n}$. Let $\U(d)$ denote
the group of $d\times d$ unitary operators. This group is naturally
represented in our system by
\begin{eqnarray}\label{eq:Q}
\bQ_{d,n}(\bsU)\ket{i_1\cdots i_n} := \bsU\ket{i_1}\ot\cdots\ot
\bsU\ket{i_n},
\end{eqnarray}
where $\bsU\in\U(d)$. In fact, $\bQ_{d,n}(\bsU):=\bsU^{\ot n}$,
which is called the \emph{collective action} of $\bsU\in\U(d)$. Thus
we have the following famous result:
\begin{thrm}[Schur, \cite{Zhang2024}]\label{th-schur}
Let $\cA = \spn_\complex\Set{\bP_{d,n}(\pi): \pi\in S_k}$ and $\cB =
\spn_\complex\Set{\bQ_{d,n}(\bsU): \bsU\in\mathsf{U}(d)}$. Then:
\begin{eqnarray}
\cA' = \cB\quad\text{and}\quad \cA = \cB'.
\end{eqnarray}
\end{thrm}

When treated as matrix algebras, such pairs $(\cA,\cB)$ are known as
\emph{dual reductive pairs} since the collective action of the
unitary group on the tensor space and the permutation action of
tensor factors are mutual commutants.

In fact, the above dual theorem by Schur can be generalized.
Consider the local unitary group
$\LU(\bsd)\equiv\U(d_1)\ot\cdots\ot\U(d_N)$, where
$\bsd:=(d_1,\ldots,d_N)$ are positive integer dimensions, which is a
subgroup of
$\GL(\bsd)\equiv\GL(d_1,\complex)\ot\cdots\ot\GL(d_N,\complex)$. Let
$V_i$ be a $d_i$-dimensional complex Hilbert space and
$V=V_1\ot\cdots\ot V_n$. Then $\LU(\bsd)$ acts on the vector space
$\End(V)=\ot^N_{i=1}\End(V_i)$, where $\End(V_i)$ is the set of all
endomorphisms from $V_i$ to itself, by
\begin{eqnarray}
\bsM\longmapsto g\bsM g^\dagger\quad(g=\bsU_1\ot\cdots\ot\bsU_N\in
\LU(\bsd),\bsM\in\End(V))
\end{eqnarray}
which is obtained by linear extension of the action:
$\ot^N_{i=1}\bsX_i\mapsto \ot^N_{i=1}\bsU_i\bsX_i\bsU^\dagger_i$,
where $\bsX_i\in\End(V_i)$ and $\bsU_i\in\U(d_i)$.

Consider the representation of $\LU(\bsd)$ on $\End(V^{\ot n})$,
defined by
\begin{eqnarray}\label{eq:qdk}
\bQ_{\bsd,n}(\bsU_1,\ldots,\bsU_N)\defeq
\bQ_{d_1,n}(\bsU_1)\ot\cdots\ot \bQ_{d_N,n}(\bsU_N),
\end{eqnarray}
where $\bQ_{d_i,n}(\bsU_i)=\bsU^{\ot n}_i$ for $\bsU_i\in \U(d_i)$.
Denote the $N$-fold Cartesian product $S^N_n:=S_n\times\cdots\times
S_n$ of the symmetric group $S_n$ of order $n$. The action of
$S^N_n$ on $\End(V^{\ot n})$ is defined by
\begin{eqnarray}\label{eq:pdk}
\bP_{\bsd,n}(\pi_1,\ldots,\pi_N)\defeq
\bP_{d_1,n}(\pi_1)\ot\cdots\ot\bP_{d_N,n}(\pi_N),
\end{eqnarray}
where $\bP_{d_i,n}(\pi_i)\in \End(V^{\ot n}_i)$ for $\pi_i\in  S_n$
with its definition taken from Eq.~\eqref{eq:P}.
\begin{thrm}[The generalized Schur-Weyl duality,
\cite{Grassl1998,Turner2017}]\label{th:GSW} Let
\begin{eqnarray}
\widetilde\cA&:=&\spn_{\complex}\Set{\bP_{\bsd,n}(\boldsymbol{\pi}):
\boldsymbol{\pi}\in S^N_n},\\
\widetilde\cB&:=&\spn_{\complex}\Set{\bQ_{\bsd,n}(g):
g\in\LU(\bsd)}.
\end{eqnarray}
Then it holds that
\begin{eqnarray}
\widetilde\cA'=\widetilde\cB\quad\text{and}\quad
\widetilde\cB'=\widetilde\cA.
\end{eqnarray}
\end{thrm}

\subsection{Proof of Proposition~\ref{th:LUequiv}}

Let $V=\Herm(\complex^{d_1}\ot\cdots\ot \complex^{d_N})$, the
Hermitian matrices acting on the tensor space, and denote the local
unitary group by $\LU(\bsd)\equiv\U(d_1)\ot\cdots\ot\U(d_N)$. Define
$\LU(\bsd)$ acts on $V$ by conjugation, i.e., for any $g\in
\LU(\bsd)$ and $\bsX\in V$, we get the conjugate action of
$\LU(\bsd)$ on $V$ via $\tau_g\bsX= g\bsX g^{\dagger}$. In fact,
given two tuples of multipartite states on
$\complex^{d_1}\ot\cdots\ot\complex^{d_N}$,
$\Psi=(\rho_1,\ldots,\rho_K)$ and $\Psi'=(\rho'_1,\ldots,\rho'_K)$,
they are locally unitarily (LU) equivalent in the sense that
$\rho'_i =g\rho_i g^\dagger$ for all $i=1,\ldots,K$ and some
$g\in\LU(\bsd)$. That is, there exist a collection of unitary
operators $\bsU_j\in\U(d_j)(j=1,\ldots,N)$ such that
$g=\bsU_1\ot\cdots\ot\bsU_N$.
\begin{eqnarray}
\rho'_i =
(\bsU_1\ot\cdots\ot\bsU_N)\rho_i(\bsU_1\ot\cdots\ot\bsU_N)^\dagger
\end{eqnarray}
for each $i=1,\ldots,K$.
\begin{proof}[Proof of Proposition~\ref{th:LUequiv}]
Clearly $\rho\in V$, then $g=\bsU_1\ot\cdots\ot\bsU_N\in \LU(\bsd)$
acts on $\rho$ via $\tau_g\rho =g\rho g^\dagger$. The space of all
real polynomials on $V$ is denoted by $\real[V]$. We will denote by
$\real[V]_n$ the space of real homogenous polynomials on $V$ of
degree $n$. We have already known that each homogeneous polynomials
of degree $n$ are mappings of the form
$p(\bsX)=\Inner{\tilde\bsp}{\bsX^{\ot n}}$, where $\bsX\in V$ and
$\tilde\bsp\in V^{\ot n}$. Thus $\tilde \bsp$ defines an
$\LU(\bsd)$-invariant polynomial $p\in \real[V]^{\LU(\bsd)}_n$ if
and only if $\tau^{\ot n}_g\tilde \bsp=\tilde \bsp$ for all $g\in
\LU(\bsd)$. Now for both $\rho$ and $\sigma$ satisfying
$\sigma=g\rho g^\dagger$ for some $g\in \LU(\bsd)$ if and only if
$p(\rho)=p(\sigma)$ for $\forall p\in \real[V]^{\LU(\bsd)}$ by
Proposition~\ref{prop:vrana}. By virtue of the above this is
equivalent to demanding that for all $g\in \LU(\bsd)$
$$
\tilde \bsp=\tau^{\ot n}_g\tilde \bsp=\bQ_{\bsd,n}(\bar g)\tilde
\bsp\bQ_{\bsd,n}(\bar g),
$$
where $\bar g:=(\bsU_1,\ldots,\bsU_N)$. This amounts to requiring
$[\bQ_{\bsd,n}(\bar g),\tilde\bsp]=0$ implying that
$\tilde\bsp\in\widetilde\cB'=\widetilde\cA$ by Theorem~\ref{th:GSW}.
Thus $\tilde \bsp$ can be expanded into a linear combination of
$\bP_{\bsd,n}(\boldsymbol{\pi})$'s for $\boldsymbol{\pi}\in S^N_n$.
By Theorem~\ref{th:GSW}, it is easily seen that $\sigma=g\rho
g^\dagger$ if and only if $\Tr{\sigma^{\ot
n}\bP_{\bsd,n}(\boldsymbol{\pi})} = \Tr{\rho^{\ot
n}\bP_{\bsd,n}(\boldsymbol{\pi})}$, where $n=1,2,\ldots$ and
$\boldsymbol{\pi}\in S^N_n$. This completes the proof.
\end{proof}

\section{Proof of Theorem~\ref{th:Bargmann-generators}}\label{app:2}

In this section, we first establish an intriguing formula
(Lemma~\ref{lem:A1}) concerning operator products. Subsequently, we
reformulate the 18 Makhlin invariants $I_k$'s using 18 LU invariant
generators, denoted as $L_k$'s (Proposition~\ref{prop:LI}). With
these foundational steps completed, we can express all 18 Bargmann
generators $B_k$'s as polynomials in terms of the 18 LU invariant
generators $L_k$'s (see Lemma~\ref{lem:LB}). Building on this, we
derive expressions for the $L_k$'s in terms of the $B_k$'s. Through
the interrelationships between the $L_k$'s and $B_k$'s, we deduce
that the set of 18 Bargmann invariants $B_k$'s constitutes a
complete set that determines the local unitary equivalence of
two-qubit states.

\subsection{Product formula for two-qubit observables}\label{sub:1}

Let us fix some notations used in this section. Firstly, we recall
the notion of the \emph{cross product} in the real Euclidean space
$\real^3$. We will make the convention by assuming that the cross
product of two row(column) vectors will be a row(column) vector
according to the definition of the cross product. For instance, for
two column vectors $\bsx=(x_1,x_2,x_3)^\t$ and
$\bsy=(y_1,y_2,y_3)^\t$ in $\real^3$, where $^\t$ means the
transpose, their cross product $\bsx\times\bsy$ is identified with
$$
\bsx\times\bsy = \Pa{\Abs{\begin{array}{cc}
                            x_2 & x_3 \\
                            y_2 & y_3
                          \end{array}
}, -\Abs{\begin{array}{cc}
           x_1 & x_3 \\
           y_1 & y_3
         \end{array}
},\Abs{\begin{array}{cc}
           x_1 & x_2 \\
           y_1 & y_2
         \end{array}
}}^\t.
$$
Moreover the cross product $\bsx^\t\times\bsy^\t$ is identified with
$$
\bsx^\t\times\bsy^\t = \Pa{\Abs{\begin{array}{cc}
                            x_2 & x_3 \\
                            y_2 & y_3
                          \end{array}
}, -\Abs{\begin{array}{cc}
           x_1 & x_3 \\
           y_1 & y_3
         \end{array}
},\Abs{\begin{array}{cc}
           x_1 & x_2 \\
           y_1 & y_2
         \end{array}
}}.
$$
According to this convention, we find that
$(\bsx\times\bsy)^\t=\bsx^\t\times\bsy^\t$.

In what follows, we will use exchangeably the notation of column
(row) vector $\bsx(\bsx^\t)$ and the Dirac notation ket (bra)
$\ket{\bsx}(\bra{\bsx})$. The inner products between two \emph{real}
3-dimensional column vectors $\bsx$ and $\bsy$ and two \emph{real}
$3\times3$ matrices $\bsM$ and $\bsN$, are defined by, respectively,
\begin{eqnarray*}
\Inner{\bsx}{\bsy}
:=\bsx^\t\bsy\quad\text{and}\quad\Inner{\bsM}{\bsN}
:=\Tr{\bsM^\t\bsN},
\end{eqnarray*}
where $\trace$ stands for the usual matrix trace. We often write
$\Inner{\bsx}{\bsM\bsy}$ as $\Innerm{\bsx}{\bsM}{\bsy}$. Denote
$\abs{\bsx}:=\sqrt{\Inner{\bsx}{\bsx}}$ and
$\norm{\bsM}:=\sqrt{\Inner{\bsM}{\bsM}}$.

We also use the notion of the \emph{cofactors} \cite{Horn2013} of
entries in a matrix is defined as follows.
\begin{definition}
For any (real or complex) square matrix $\bsM=(m_{ij})_{n\times n}$,
the so-called \emph{cofactor} of entry $m_{ij}$ is defined as the
factor $(-1)^{i+j}$ times the determinant of the $(n-1)\times(n-1)$
matrix (denoted by $\bsM[\hat i|\hat j]$) obtained by deleting the
$i$-th row and $j$-th column of $\bsM$. That is, the cofactor of
$m_{ij}$ is
\begin{eqnarray}
\widehat m_{ij}\defeq (-1)^{i+j}\det\Pa{\bsM[\hat i|\hat j]}.
\end{eqnarray}
Denote by $\widehat \bsM:=(\widehat m_{ij})_{n\times n}$, which is
called the \emph{cofactor matrix}. Then
$\bsM^*\defeq{\widehat\bsM}^\t$ is called the \emph{adjugate matrix}
of $\bsM$.
\end{definition}
In Linear Algebra, for any two square matrices $\bsM$ and $\bsN$ of
order $n$, it is well-known that
\begin{eqnarray}
\widehat{\bsM^\t} = (\widehat\bsM)^\t\quad \text{and}\quad
\widehat{\bsM\bsN}= \widehat\bsM\widehat\bsN.
\end{eqnarray}
Let the characteristic polynomial of the $n\times n$ matrix $\bsM$
be $f_n(\lambda)$. Then
\begin{eqnarray}\label{eq:chpoly}
f_n(x)=\sum^n_{k=0}(-1)^k e_k(\bsM)x^{n-k},
\end{eqnarray}
where
\begin{eqnarray}
\begin{cases}
e_0(\bsM)&\equiv1\\
e_1(\bsM)&=\Tr{\bsM}\\
&\vdots\\
e_{n-1}(\bsM)&=\Tr{\widehat\bsM}\\
e_n(\bsM)&=\det(\bsM).
\end{cases}
\end{eqnarray}
We can use Hamilton-Cayley theorem in Linear Algebra, together with
the continuity argument, to give a formula towards the computation
of adjugate matrix, which can be described as follows:
\begin{prop}\label{prop:adjugate}
For any $n\times n$ matrix $\bsM$, its adjugate matrix can be
determined by
\begin{eqnarray}
\bsM^* =\sum^{n-1}_{k=0} e_k(\bsM)(-\bsM)^{n-1-k}.
\end{eqnarray}
\end{prop}

\begin{proof}
Indeed, This indicates by Hamilton-Cayley Theorem that
\begin{eqnarray*}
\bsM^n-e_1(\bsM)\bsM^{n-1}+\cdots+(-1)^{n-1}e_{n-1}(\bsM)\bsM+(-1)^n\det(\bsM)\I_n=0.
\end{eqnarray*}
Thus
\begin{eqnarray*}
\Pa{\bsM^{n-1}-e_1(\bsM)\bsM^{n-2}+(-1)^{n-1}e_{n-1}(\bsM)\I_n}\bsM=(-1)^{n-1}\det(\bsM)\I_n=(-1)^{n-1}\bsM^*\bsM.
\end{eqnarray*}
Then
\begin{eqnarray*}
\bsM^*&=&(-\bsM)^{n-1}+e_1(\bsM)(-\bsM)^{n-2}+\cdots+
e_k(\bsM)(-\bsM)^{n-1-k}+e_{n-1}(\bsM)\I_n\\
&=&\sum^{n-1}_{k=0} e_k(\bsM)(-\bsM)^{n-1-k}
\end{eqnarray*}
holds true if $\bsM$ is invertible. By the continuity argument, this
holds true for all square matrix $\bsM$.
\end{proof}

\begin{cor}\label{lem:adjugate}
For any square matrix $\bsM\in\real^{3\times3}$, it holds that
\begin{enumerate}[(i)]
\item $\Tr{\widehat\bsM} = \frac{\Tr{\bsM}^2-\Tr{\bsM^2}}2$;
\item $\widehat\bsM^\t\widehat\bsM = (\bsM^\t\bsM)^2 -
\Inner{\bsM}{\bsM}\bsM^\t\bsM+\Inner{\widehat\bsM}{\widehat\bsM}\I_3$;
\item
$\Inner{\widehat\bsM}{\widehat\bsM}=\frac12\Pa{\Inner{\bsM}{\bsM}^2-\Inner{\bsM^\t\bsM}{\bsM^\t\bsM}}$;
\item
$\widehat{\widehat\bsM}=\bsM^4-c_2(\bsM)\bsM^2+c_1(\bsM)\bsM+c_0(\bsM)\I_3$,
where three coefficients $c_k(\bsM)(k=0,1,2)$ are identified with
\begin{eqnarray}
\begin{cases}
c_0(\bsM) &= \frac{-\Tr{\bsM}^4+2\Tr{\bsM}^2\Tr{\bsM^2}+\Tr{\bsM^2}^2-2\Tr{\bsM^4}}8,\\
c_1(\bsM) &= \frac{\Tr{\bsM}\Pa{\Tr{\bsM}^2-\Tr{\bsM^2}}}2,\\
c_2(\bsM) &= \frac{\Tr{\bsM}^2+\Tr{\bsM^2}}2.
\end{cases}
\end{eqnarray}
\end{enumerate}
\end{cor}

\begin{proof}
The proof is conceptually simple. We can also use
\textsc{Mathematica} to do this. In what follows, we give analytical
reasoning. By Proposition~\ref{prop:adjugate}, we see that
\begin{eqnarray}\label{eq:3rdadj}
\bsM^*=\widehat\bsM^\t = \bsM^2-\Tr{\bsM}\bsM+\Tr{\widehat\bsM}\I_3.
\end{eqnarray}
(i) By taking the traces on both sides, we get that
\begin{eqnarray*}
\Tr{\widehat\bsM} = \frac{\Tr{\bsM}^2-\Tr{\bsM^2}}2.
\end{eqnarray*}
(ii) Now we use $\bsM^\t\bsM$ to replace $\bsM$ in
Eq.~\eqref{eq:3rdadj}, then
\begin{eqnarray*}
\widehat\bsM^\t\widehat\bsM  &=& \widehat{\bsM^\t\bsM} =
(\bsM^\t\bsM)^2 - \Tr{\bsM^\t\bsM}\bsM^\t\bsM +
\Tr{\widehat{\bsM^\t\bsM}}\I_3\\
&=&(\bsM^\t\bsM)^2 - \Inner{\bsM}{\bsM}\bsM^\t\bsM
+\Inner{\widehat\bsM}{\widehat\bsM}\I_3.
\end{eqnarray*}
(iii) By taking the traces on both sides of the identity in (ii),
after simplifying it, we get the desired result.\\
(iv) Apparently,
\begin{eqnarray*}
\widehat{\widehat\bsM} &=& \Pa{\widehat\bsM^2 -
\Tr{\widehat\bsM}\widehat\bsM +
\Tr{\widehat{\widehat\bsM}}\I_3}^\t\\
&=&\Pa{\widehat\bsM^2}^\t - \Tr{\widehat\bsM}\widehat\bsM^\t +
\Tr{\widehat{\widehat\bsM}}\I_3,
\end{eqnarray*}
where
\begin{eqnarray*}
\Pa{\widehat\bsM^2}^\t &=& \widehat{\bsM^2}^\t =
\bsM^4-\Tr{\bsM^2}\bsM^2+\Tr{\widehat{\bsM^2}}\I_3.
\end{eqnarray*}
Thus substituting this into the expression of
$\widehat{\widehat\bsM}$, we get that
\begin{eqnarray*}
\widehat{\widehat\bsM} &=&
\Pa{\bsM^4-\Tr{\bsM^2}\bsM^2+\Tr{\widehat{\bsM^2}}\I_3} \\
&&- \Tr{\widehat\bsM}\Pa{\bsM^2-\Tr{\bsM}\bsM+\Tr{\widehat\bsM}\I_3}
+ \Tr{\widehat{\widehat\bsM}}\I_3\\
&=&\bsM^4 - \Br{\Tr{\bsM^2}+\Tr{\widehat\bsM}}\bsM^2 +
\Tr{\bsM}\Tr{\widehat\bsM}\bsM\\
&&+\Br{\Tr{\widehat{\bsM^2}}-\Tr{\widehat\bsM}^2+\Tr{\widehat{\widehat\bsM}}}\I_3.
\end{eqnarray*}
Using many times the result obtained in (i), finally we obtain the
desired identity.
\end{proof}

\subsubsection{Product formula}

As conventions, three Pauli matrices are given below:
\begin{eqnarray}
\sigma_1=\Pa{\begin{array}{cc}
0 & 1 \\
1 & 0
\end{array}},\quad \sigma_2=\Pa{\begin{array}{cc}
0 & -\mathrm{i} \\
\mathrm{i} & 0
\end{array}},\quad \sigma_3=\Pa{\begin{array}{cc}
1 & 0 \\
0 & -1
\end{array}}.
\end{eqnarray}
For any two-qubit \emph{observable} $\bsX$, we can decompose it as
\begin{eqnarray}
\bsX = t\I_4 + \bsr\cdot\bsigma\ot\I_2 + \I_2\ot\bss\cdot\bsigma +
\sum^3_{i,j=1}t_{ij}\sigma_i\ot\sigma_j,
\end{eqnarray}
where
$t\in\real,\bsr:=(r_1,r_2,r_3)^\t,\bss:=(s_1,s_2,s_3)^\t\in\real^3$,
and $\bsT:=(t_{ij})_{3\times3}\in\real^{3\times3}$. Here
$\bsr\cdot\bsigma:=\sum^3_{i=1}r_i\sigma_i$. By mimicking this
notation, we introduce the following notation:
$\bsF_k=(\varepsilon_{ijk})_{3\times 3}$, where
$\varepsilon_{ijk}:=\sign[(j-i)(k-i)(k-j)]$ for
$i,j,k\in[3]:=\set{1,2,3}$. Indeed,
\begin{eqnarray}
\bsF_1=\Pa{\begin{array}{ccc}
0 & 0 & 0\\
0 & 0 & 1\\
0 & -1 & 0
\end{array}},\quad \bsF_2=\Pa{\begin{array}{ccc}
0 & 0 & -1\\
0 & 0 & 0\\
1 & 0 & 0
\end{array}},\quad \bsF_3=\Pa{\begin{array}{ccc}
0 & 1 & 0\\
-1 & 0 & 0\\
0 & 0 & 0
\end{array}}.
\end{eqnarray}
Denote $\bsx\cdot\cF:=\sum^3_{k=1}x_k\bsF_k$, where
$\cF:=(\bsF_1,\bsF_2,\bsF_3)$. It is easily seen that the cross
product can be realized as
\begin{eqnarray}
\bsx\times \bsy =
\Pa{\Innerm{\bsx}{\bsF_1}{\bsy},\Innerm{\bsx}{\bsF_2}{\bsy},\Innerm{\bsx}{\bsF_3}{\bsy}}^\t.
\end{eqnarray}
For convenience, we parameterize $\bsX$ in the notation
$(t,\bsr,\bss,\bsT)$ for $\bsX$, denoted by
$\bsX\approx(t,\bsr,\bss,\bsT)$, and $(t',\bsr',\bss',\bsT')$ for
$\bsX'$, denoted by $\bsX'\approx(t',\bsr',\bss',\bsT')$,
respectively. Consider the product $\tilde \bsX:=\bsX\bsX'$ with
parameters $(\tilde t,\tilde\bsr,\tilde\bss,\tilde\bsT)$.

In order to describe our product formula for $\tilde\bsX$, we
introduce the following notations: Denote
\begin{eqnarray}\label{eq:Omegasymbol}
\Omega(\bsM,\bsN):=\Pa{\begin{array}{c}
      \bse^\t_2\bsM\times\bse^\t_3\bsN+\bse^\t_2\bsN\times\bse^\t_3\bsM  \\
      \bse^\t_3\bsM\times\bse^\t_1\bsN+\bse^\t_3\bsN\times\bse^\t_1\bsM\\
      \bse^\t_1\bsM\times\bse^\t_2\bsN+\bse^\t_1\bsN\times\bse^\t_2\bsM
    \end{array}
},
\end{eqnarray}
where $\bsM,\bsN\in\real^{3\times 3}$ and
$\set{\bse_1,\bse_2,\bse_3}$ is the computational basis of
$\real^3$, defined by $\bse_1=(1,0,0)^\t,\bse_2=(0,1,0)^\t$, and
$\bse_3=(0,0,1)^\t$. Clearly $\Omega$ is symmetric bilinear mapping
in the sense that $\Omega(\bsM,\bsN)=\Omega(\bsN,\bsM)$. Let
\begin{eqnarray}
\Psi(\bsx,\bsM,\bsy):=(\bsx\cdot\cF)^\t\bsM+\bsM(\bsy\cdot\cF),
\end{eqnarray}
where $\bsx,\bsy\in\real^3$ and $\bsM\in\real^{3\times3}$.

\begin{prop}\label{prop:Omega}
For the matrix $\Omega(\bsM,\bsN)$, its entries can be identified as
\begin{eqnarray}
\Omega(\bsM,\bsN)_{p,q} =-
\Inner{\bsF_p\bsM\bsF_q}{\bsN}\quad(\forall p,q\in\set{1,2,3}).
\end{eqnarray}
Moreover, it holds that
\begin{eqnarray}
\Omega(\bsM,\bsN) =
\frac12\sum^3_{i,j=1}\out{\bse_i\times\bse_j}{\bse^\t_i\bsM\times\bse^\t_j\bsN+\bse^\t_i\bsN\times\bse^\t_j\bsM}.
\end{eqnarray}
\end{prop}

\begin{proof}
For the first row of $\Omega(\bsM,\bsN)$, we find that
\begin{eqnarray*}
&&\bse^\t_2\bsM\times\bse^\t_3\bsN+\bse^\t_2\bsN\times\bse^\t_3\bsM\\
&&=\Pa{\Innerm{\bse_2}{\bsM\bsF_1\bsN^\t+\bsN\bsF_1\bsM^\t}{\bse_3},\Innerm{\bse_2}{\bsM\bsF_2\bsN^\t+\bsN\bsF_2\bsM^\t}{\bse_3},\Innerm{\bse_2}{\bsM\bsF_3\bsN^\t+\bsN\bsF_3\bsM^\t}{\bse_3}}.
\end{eqnarray*}
Next, we determine such three components as follows. In fact,
$\bsM\bsF_j\bsN^\t+\bsN\bsF_j\bsM^\t$ is skew symmetric, and thus it
can be decomposed as
\begin{eqnarray*}
\bsM\bsF_j\bsN^\t+\bsN\bsF_j\bsM^\t =
c^{(j)}_1\bsF_1+c^{(j)}_2\bsF_2+c^{(j)}_3\bsF_3.
\end{eqnarray*}
This implies that
\begin{eqnarray*}
\Tr{\bsF_i(\bsM\bsF_j\bsN^\t+\bsN\bsF_j\bsM^\t)} =
c^{(j)}_1\Tr{\bsF_i\bsF_1}+c^{(j)}_2\Tr{\bsF_i\bsF_2}+c^{(j)}_3\Tr{\bsF_i\bsF_3}.
\end{eqnarray*}
That is,
\begin{eqnarray*}
\Tr{\bsF_i\bsN\bsF_j\bsM^\t} =
-c^{(j)}_1\delta_{1i}-c^{(j)}_2\delta_{2i}-c^{(j)}_3\delta_{3i}\Longrightarrow
c^{(j)}_i=-\Tr{(\bsF_i\bsM\bsF_j)^\t\bsN}=-\Inner{\bsF_i\bsM\bsF_j}{\bsN}.
\end{eqnarray*}
From this observation, we get that
$\Innerm{\bse_2}{\bsM\bsF_j\bsN^\t+\bsN\bsF_j\bsM^\t}{\bse_3}=-\Tr{\bsF_1\bsM\bsF_j\bsN^\t}$,
which implies that
\begin{eqnarray*}
\bse^\t_2\bsM\times\bse^\t_3\bsN+\bse^\t_2\bsN\times\bse^\t_3\bsM
=-\Pa{\Inner{\bsF_1\bsM\bsF_1}{\bsN},\Inner{\bsF_1\bsM\bsF_2}{\bsN},\Inner{\bsF_1\bsM\bsF_3}{\bsN}}
\end{eqnarray*}
Similar procedures for second and third rows are performed,
respectively, and thus we get the desired result:
$\Omega(\bsM,\bsN)_{p,q}=-\Inner{\bsF_p\bsM\bsF_q}{\bsN}$. The
second item can be checked as follows: Clearly $i=j$,
$\out{\bse_i\times\bse_j}{\bse^\t_i\bsM\times\bse^\t_j\bsN+\bse^\t_i\bsN\times\bse^\t_j\bsM}=0$
due to the fact that $\bse_i\times\bse_j=0$ if $i=j$. Besides, for
$i\neq j$,
$$
\out{\bse_i\times\bse_j}{\bse^\t_i\bsM\times\bse^\t_j\bsN+\bse^\t_i\bsN\times\bse^\t_j\bsM}
=
\out{\bse_j\times\bse_i}{\bse^\t_j\bsM\times\bse^\t_i\bsN+\bse^\t_j\bsN\times\bse^\t_i\bsM}.
$$
It suffices to consider $(i,j)=(1,2),(1,3),(2,3)$. Note that
$\bse_1\times\bse_2=\bse_3,\bse_2\times\bse_3=\bse_1$, and
$\bse_3\times\bse_1=\bse_2$. Thus we get that
\begin{eqnarray*}
&&\frac12\sum^3_{i,j=1}\out{\bse_i\times\bse_j}{\bse^\t_i\bsM\times\bse^\t_j\bsN+\bse^\t_i\bsN\times\bse^\t_j\bsM}=\sum_{1\leqslant
i<j\leqslant
3}\out{\bse_i\times\bse_j}{\bse^\t_i\bsM\times\bse^\t_j\bsN+\bse^\t_i\bsN\times\bse^\t_j\bsM}\\
&&=\out{\bse_1\times\bse_2}{\bse^\t_1\bsM\times\bse^\t_2\bsN+\bse^\t_1\bsN\times\bse^\t_2\bsM}+\out{\bse_1\times\bse_3}{\bse^\t_1\bsM\times\bse^\t_3\bsN+\bse^\t_1\bsN\times\bse^\t_3\bsM}\\
&&~~~+\out{\bse_2\times\bse_3}{\bse^\t_2\bsM\times\bse^\t_3\bsN+\bse^\t_2\bsN\times\bse^\t_3\bsM}\\
&&=\out{\bse_3}{\bse^\t_1\bsM\times\bse^\t_2\bsN+\bse^\t_1\bsN\times\bse^\t_2\bsM}+\out{\bse_2}{\bse^\t_3\bsM\times\bse^\t_1\bsN+\bse^\t_3\bsN\times\bse^\t_1\bsM}\\
&&~~~+\out{\bse_1}{\bse^\t_2\bsM\times\bse^\t_3\bsN+\bse^\t_2\bsN\times\bse^\t_3\bsM},
\end{eqnarray*}
which implies the desired result when writing it in matrix form.
\end{proof}

We have the following formula for the product $\tilde\bsX=\bsX\bsX'$
of $\bsX$ and $\bsX'$.
\begin{lem}[Product formula of two-qubit observables]\label{lem:A1}
If $\bsX\approx(t,\bsr,\bss,\bsT)$ and
$\bsX'\approx(t',\bsr',\bss',\bsT')$, then $\tilde\bsX\approx(\tilde
t,\tilde\bsr,\tilde\bss,\tilde\bsT)$ is given by the following
formulae:
\begin{eqnarray}
\begin{cases}
\tilde t = tt' + \inner{\bsr}{\bsr'} + \inner{\bss}{\bss'} + \inner{\bsT}{\bsT'},\\
\tilde\bsr = t'\bsr+t\bsr'+\bsT'\bss+\bsT\bss'+\mathrm{i}\Pa{\bsr\times\bsr'+\sum^3_{i=1}\bsT\bse_i\times\bsT'\bse_i},\\
\tilde\bss = t'\bss+t\bss'+{\bsT'}^\t\bsr+\bsT^\t\bsr'+\mathrm{i}\Pa{\bss\times\bss'+\sum^3_{i=1}\bsT^\t\bse_i\times{\bsT'}^\t\bse_i},\\
\tilde\bsT = t'\bsT + t\bsT' + \out{\bsr}{\bss'} + \out{\bsr'}{\bss}
-\Omega(\bsT,\bsT')+\mathrm{i}\Pa{\Psi(\bsr,\bsT',\bss)-\Psi(\bsr',\bsT,\bss')}.
\end{cases}
\end{eqnarray}
Moreover, $\Tr{\bsX\bsX'}=4(tt' + \inner{\bsr}{\bsr'} +
\inner{\bss}{\bss'} + \inner{\bsT}{\bsT'})$.
\end{lem}

\begin{proof}
The proof is conceptually, but needs tedious algebraic computations.
Indeed,
\begin{eqnarray*}
\tilde t &=& \frac14\Tr{\tilde\bsX}=\frac14\Tr{\bsX\bsX'},\\
\tilde r_i &=& \frac14\Tr{\bsX\bsX'(\sigma_i\ot\I_2)},\\
\tilde s_j &=& \frac14\Tr{\bsX\bsX'(\I_2\ot\sigma_j)},\\
\tilde t_{ij}&=& \frac14\Tr{\bsX\bsX'(\sigma_i\ot\sigma_j)}.
\end{eqnarray*}
The next step is to check the correctness of the desired formula.
This can be done by using the symbolic computation of the
mathematical software \textsc{Mathematica}. Assume that
$\bsX\approx(t,\bsr,\bss,\bsT)$ and
$\bsX'\approx(t',\bsr',\bss',\bsT')$. Then
\begin{eqnarray*}
\bsX\bsX' &=&\Pa{t\I_4 + \bsr\cdot\bsigma\ot\I_2 +
\I_2\ot\bss\cdot\bsigma +
\sum^3_{i,j=1}t_{ij}\sigma_i\ot\sigma_j}\\
&&\times\Pa{t'\I_4 + \bsr'\cdot\bsigma\ot\I_2 +
\I_2\ot\bss'\cdot\bsigma +
\sum^3_{i,j=1}t'_{ij}\sigma_i\ot\sigma_j}\\
&=&\Pa{tt'\I_4 + t\bsr'\cdot\bsigma\ot\I_2 + \I_2\ot
t\bss'\cdot\bsigma + \sum^3_{i,j=1}tt'_{ij}\sigma_i\ot\sigma_j}\\
&&+\Pa{t'\bsr\cdot\bsigma\ot\I_2 +
(\bsr\cdot\bsigma)(\bsr'\cdot\bsigma)\ot\I_2 +
\bsr\cdot\bsigma\ot\bss'\cdot\bsigma +
\sum^3_{i,j=1}t'_{ij}(\bsr\cdot\bsigma)\sigma_i\ot\sigma_j}\\
&&+\Pa{\I_2\ot t'\bss\cdot\bsigma +
\bsr'\cdot\bsigma\ot\bss\cdot\bsigma +
\I_2\ot(\bss\cdot\bsigma)(\bss'\cdot\bsigma) +
\sum^3_{i,j=1}t'_{ij}\sigma_i\ot(\bss\cdot\bsigma)\sigma_j}\\
&&+\Pa{t'\sum^3_{i,j=1}t_{ij}\sigma_i\ot\sigma_j +
\sum^3_{i,j=1}t_{ij}\sigma_i(\bsr'\cdot\bsigma)\ot\sigma_j +
\sum^3_{i,j=1}t_{ij}\sigma_i\ot\sigma_j(\bss'\cdot\bsigma)}\\
&&+
\Pa{\sum^3_{i,j=1}t_{ij}\sigma_i\ot\sigma_j}\Pa{\sum^3_{i,j=1}t'_{ij}\sigma_i\ot\sigma_j}.
\end{eqnarray*}
Furthermore
\begin{eqnarray*}
\bsX\bsX' &=&tt'\I_4 + (t\bsr'+t'\bsr)\cdot\bsigma\ot\I_2 + \I_2\ot
(t\bss'+t'\bss)\cdot\bsigma + \sum^3_{i,j=1}(tt'_{ij}+t't_{ij})\sigma_i\ot\sigma_j\\
&&+\Pa{(\bsr\cdot\bsigma)(\bsr'\cdot\bsigma)\ot\I_2 +
\bsr\cdot\bsigma\ot\bss'\cdot\bsigma +
\sum^3_{i,j=1}t'_{ij}(\bsr\cdot\bsigma)\sigma_i\ot\sigma_j+\sum^3_{i,j=1}t_{ij}\sigma_i(\bsr'\cdot\bsigma)\ot\sigma_j}\\
&&+\Pa{\bsr'\cdot\bsigma\ot\bss\cdot\bsigma +
\I_2\ot(\bss\cdot\bsigma)(\bss'\cdot\bsigma) +
\sum^3_{i,j=1}t'_{ij}\sigma_i\ot(\bss\cdot\bsigma)\sigma_j+\sum^3_{i,j=1}t_{ij}\sigma_i\ot\sigma_j(\bss'\cdot\bsigma)}\\
&&+
\Pa{\sum^3_{i,j=1}t_{ij}\sigma_i\ot\sigma_j}\Pa{\sum^3_{i,j=1}t'_{ij}\sigma_i\ot\sigma_j}.
\end{eqnarray*}
Note that
$(\bsr\cdot\bsigma)(\bsr'\cdot\bsigma)\ot\I_2=\Inner{\bsr}{\bsr'}\I_4+\mathrm{i}(\bsr\times\bsr')\cdot\bsigma\ot\I_2$
and
$\I_2\ot(\bss\cdot\bsigma)(\bss'\cdot\bsigma)=\Inner{\bss}{\bss'}\I_4+\I_2\ot\mathrm{i}(\bss\times\bss')\cdot\bsigma$.
Then we see that
\begin{eqnarray*}
\bsX\bsX' &=&(tt'+\Inner{\bsr}{\bsr'}+\Inner{\bss}{\bss'})\I_4 +
(t\bsr'+t'\bsr+\mathrm{i}\bsr\times\bsr')\cdot\bsigma\ot\I_2  \\
&&+ \I_2\ot
(t\bss'+t'\bss+\mathrm{i}\bss\times\bss')\cdot\bsigma+ \sum^3_{i,j=1}(t\bsT'+t'\bsT)_{ij}\sigma_i\ot\sigma_j\\
&&+\Pa{\bsr\cdot\bsigma\ot\bss'\cdot\bsigma +
\sum^3_{i,j=1}t'_{ij}(\bsr\cdot\bsigma)\sigma_i\ot\sigma_j+\sum^3_{i,j=1}t_{ij}\sigma_i(\bsr'\cdot\bsigma)\ot\sigma_j}\\
&&+\Pa{\bsr'\cdot\bsigma\ot\bss\cdot\bsigma +
\sum^3_{i,j=1}t'_{ij}\sigma_i\ot(\bss\cdot\bsigma)\sigma_j+\sum^3_{i,j=1}t_{ij}\sigma_i\ot\sigma_j(\bss'\cdot\bsigma)}\\
&&+
\Pa{\sum^3_{i,j=1}t_{ij}\sigma_i\ot\sigma_j}\Pa{\sum^3_{i,j=1}t'_{ij}\sigma_i\ot\sigma_j}.
\end{eqnarray*}
Now we use the fact that
$\sigma_i\sigma_j=\mathrm{i}\sum^3_{k=1}\varepsilon_{ijk}\sigma_k+\delta_{ij}\I_2$
and get that
\begin{eqnarray*}
\bsr\cdot\bsigma\ot\bss'\cdot\bsigma=\sum^3_{i,j=1}(\out{\bsr}{\bss'})_{ij}\sigma_i\ot\sigma_j,\\
\bsr'\cdot\bsigma\ot\bss\cdot\bsigma=\sum^3_{i,j=1}(\out{\bsr'}{\bss})_{ij}\sigma_i\ot\sigma_j.
\end{eqnarray*}
We also have
\begin{eqnarray*}
&&\sum^3_{i,j=1}t'_{ij}(\bsr\cdot\bsigma)\sigma_i\ot\sigma_j =
\sum^3_{k=1}\sum^3_{i,j=1}t'_{ij}r_k\sigma_k\sigma_i\ot\sigma_j=\sum^3_{i,j,k=1}t'_{ij}r_k\Pa{\mathrm{i}\sum^3_{j'=1}\varepsilon_{kij'}\sigma_{j'}+\delta_{ki}\I_2}\ot\sigma_j\\
&&=\mathrm{i}\sum^3_{j',j=1}\Pa{\sum^3_{k,i=1}t'_{ij}\varepsilon_{kij'}r_k}\sigma_{j'}\ot\sigma_j
+ \I_2\ot\sum^3_{i,j,k=1}\Pa{t'_{ij}\delta_{ki}r_k}\sigma_j\\
&&=\mathrm{i}\sum^3_{j',j=1}\Pa{(\bsr\cdot\cF)^\t\bsT'}_{j'j}\sigma_{j'}\ot\sigma_j
+ \I_2\ot \Pa{{\bsT'}^\t\bsr}\cdot\bsigma
=\mathrm{i}\sum^3_{i,j=1}\Pa{(\bsr\cdot\cF)^\t\bsT'}_{ij}\sigma_i\ot\sigma_j
+ \I_2\ot \Pa{{\bsT'}^\t\bsr}\cdot\bsigma
\end{eqnarray*}
and
\begin{eqnarray*}
&&\sum^3_{i,j=1}t_{ij}\sigma_i(\bsr'\cdot\bsigma)\ot\sigma_j =
\sum^3_{k=1}\sum^3_{i,j=1}t_{ij}r'_k\sigma_i\sigma_k\ot\sigma_j=\sum^3_{i,j,k=1}t_{ij}r'_k\Pa{\mathrm{i}\sum^3_{j'=1}\varepsilon_{ikj'}\sigma_{j'}+\delta_{ik}\I_2}\ot\sigma_j\\
&&=-\mathrm{i}\sum^3_{j',j=1}\Pa{\sum^3_{i,k=1}t_{ij}\varepsilon_{ij'k}r'_k}\sigma_{j'}\ot\sigma_j
+ \I_2\ot\sum^3_{i,j,k=1}\Pa{t_{ij}\delta_{ik}r'_k}\sigma_j\\
&&=-\mathrm{i}\sum^3_{j',j=1}\Pa{(\bsr'\cdot\cF)^\t\bsT}_{j'j}\sigma_{j'}\ot\sigma_j
+ \I_2\ot \Pa{\bsT^\t\bsr}\cdot\bsigma
=-\mathrm{i}\sum^3_{i,j=1}\Pa{(\bsr'\cdot\cF)^\t\bsT}_{ij}\sigma_i\ot\sigma_j
+ \I_2\ot \Pa{\bsT^\t\bsr'}\cdot\bsigma.
\end{eqnarray*}
Similarly, we get that
\begin{eqnarray*}
\sum^3_{i,j=1}t'_{ij}\sigma_i\ot(\bss\cdot\bsigma)\sigma_j =
\mathrm{i}\sum^3_{i,j=1}\Pa{\bsT'(\bss\cdot\cF)}_{ij}\sigma_i\ot\sigma_j
+ (\bsT'\bss)\cdot\bsigma\ot\I_2,\\
\sum^3_{i,j=1}t_{ij}\sigma_i\ot\sigma_j(\bss'\cdot\bsigma) =
-\mathrm{i}\sum^3_{i,j=1}\Pa{\bsT(\bss'\cdot\cF)}_{ij}\sigma_i\ot\sigma_j
+ (\bsT\bss')\cdot\bsigma\ot\I_2.
\end{eqnarray*}
At last,
\begin{eqnarray*}
&&\Pa{\sum^3_{i,j=1}t_{ij}\sigma_i\ot\sigma_j}\Pa{\sum^3_{i,j=1}t'_{ij}\sigma_i\ot\sigma_j}=\sum^3_{i,j,k,l=1}t_{ij}t'_{kl}\sigma_i\sigma_k\ot\sigma_j\sigma_l\\
&&=\sum^3_{i,j,k,l=1}t_{ij}t'_{kl}\Pa{\mathrm{i}\sum^3_{p=1}\varepsilon_{ikp}\sigma_p+\delta_{ik}\I_2}\ot\Pa{\mathrm{i}\sum^3_{q=1}\varepsilon_{jlq}\sigma_q+\delta_{jl}\I_2}\\
&&=-\sum^3_{i,j,k,l,p,q=1}t_{ij}t'_{kl}\varepsilon_{ikp}\varepsilon_{jlq}\sigma_p\ot\sigma_q+\mathrm{i}\sum^3_{i,j,k,l,p=1}t_{ij}t'_{kl}\varepsilon_{ikp}\delta_{jl}\sigma_p\ot\I_2+\mathrm{i}\sum^3_{i,j,k,l,q=1}t_{ij}t'_{kl}\varepsilon_{jlq}\delta_{ik}\I_2\ot\sigma_q\\
&&~~~~~~+\sum^3_{i,j,k,l=1}t_{ij}t'_{kl}\delta_{ik}\delta_{jl}\I_2\ot\I_2\\
&&=\sum^3_{p,q=1}\Pa{\sum^3_{i,j,k,l=1}t_{ij}t'_{kl}\varepsilon_{ikp}\varepsilon_{ljq}}\sigma_p\ot\sigma_q+
\mathrm{i}\sum^3_{p=1}\Pa{\sum^3_{i,j,k,l=1}t_{ij}t'_{kl}\varepsilon_{ikp}\delta_{jl}}\sigma_p\ot\I_2\\
&&~~~~~~+\mathrm{i}\I_2\ot\sum^3_{q=1}\Pa{\sum^3_{i,j,k,l=1}t_{ij}t'_{kl}\varepsilon_{jlq}\delta_{ik}}\sigma_q+\Inner{\bsT}{\bsT'}\I_4\\
&&=-\sum^3_{p,q=1}\Omega(\bsT,\bsT')_{p,q}\sigma_p\ot\sigma_q+\mathrm{i}\Pa{\sum^3_{i=1}\bsT\bse_i\times\bsT'\bse_i}\cdot\bsigma\ot\I_2+\mathrm{i}\I_2\ot\Pa{\sum^3_{i=1}\bsT^\t\bse_i\times{\bsT'}^\t\bse_i}\cdot\bsigma\\
&&~~~~~~+\Inner{\bsT}{\bsT'}\I_4,
\end{eqnarray*}
where we used the facts that
\begin{enumerate}[(1)]
\item $\Omega(\bsT,\bsT')_{p,q}=-\Inner{\bsF_p\bsT\bsF_q}{\bsT'}$;
\item $(\sum^3_{j=1}\bsT\bse_j\times\bsT'\bse_j)\cdot\bsigma = \sum^3_{j=1} \Pa{\sum^3_{i,k,p=1}(\bsT\bse_j)_i(\bsT'\bse_j)_k\varepsilon_{ikp}\sigma_p} =
\sum^3_{j=1}\Pa{\sum^3_{i,k,p=1}(\bsT\bse_j)_i(\bsT'\bse_l)_k\varepsilon_{ikp}\delta_{jl}\sigma_p}=\sum^3_{i,j,k,l,p=1}t_{ij}t'_{kl}\varepsilon_{ikp}\delta_{jl}\sigma_p$;
\item
$\Pa{\sum^3_{i=1}\bsT^\t\bse_i\times{\bsT'}^\t\bse_i}\cdot\bsigma=\sum^3_{i,j,k,l,q=1}t_{ij}t'_{kl}\varepsilon_{jlq}\delta_{ik}\sigma_q$.
\end{enumerate}
We are done.
\end{proof}

The advantage of this product formula for two-qubit observables lies
in its independence from the components of vectors (or matrix
entries).

\begin{cor}\label{cor:B1}
The commutator $[\bsX,\bsX']:=\bsX\bsX'-\bsX'\bsX$ is identified as
\begin{eqnarray}
[\bsX,\bsX']&=&
2\mathrm{i}\Big[\Big(\bsr\times\bsr'+\sum^3_{i=1}\bsT\bse_i\times\bsT'\bse_i\Big)\cdot\sigma\ot\I_2
+
\I_2\ot\Big(\bss\times\bss'+\sum^3_{i=1}\bsT^\t\bse_i\times{\bsT'}^\t\bse_i\Big)\cdot\sigma\notag\\
&&~~~~~+\sum^3_{i,j=1}\Big(\Psi(\bsr,\bsT',\bss)-\Psi(\bsr',\bsT,\bss')\Big)_{ij}\sigma_i\ot\sigma_j\Big].
\end{eqnarray}
Moreover $[\bsX,\bsX']=0$ if and only if
\begin{eqnarray}
\begin{cases}
\bsr\times\bsr'+\sum^3_{i=1}\bsT\bse_i\times\bsT'\bse_i&=\zero,\\
\bss\times\bss'+\sum^3_{i=1}\bsT^\t\bse_i\times{\bsT'}^\t\bse_i&=\zero,\\
\Psi(\bsr,\bsT',\bss)&=\Psi(\bsr',\bsT,\bss').
\end{cases}
\end{eqnarray}
\end{cor}

\begin{prop}\label{prop:crossproperty}
It holds that
\begin{eqnarray}
\sum^3_{i=1}\bsA\bse_i\times\bsB\bse_i =
\sum^3_{i=1}(\bsA\bsB^\t\bse_i)\times\bse_i=\sum^3_{i=1}\bse_i\times(\bsB\bsA^\t\bse_i),
\end{eqnarray}
where $\bsA,\bsB\in\real^{3\times3}$.
\end{prop}

\begin{proof}
Indeed,
\begin{eqnarray*}
\sum^3_{i=1}\bsA\bse_i\times\bsB\bse_i
&=&\sum^3_{i=1}\bsA\bse_i\times\sum^3_{j=1}\proj{\bse_j}\bsB\bse_i =
\sum^3_{j=1}\sum^3_{i=1}\bsA\bse_i\Innerm{\bse_j}{\bsB}{\bse_i}\times
\bse_j\\
&=&\sum^3_{j=1}\sum^3_{i=1}\bsA\ket{\bse_i}
\Innerm{\bse_i}{\bsB^\t}{\bse_j}\times\bse_j=\sum^3_{j=1}\bsA\sum^3_{i=1}\ket{\bse_i}
\Innerm{\bse_i}{\bsB^\t}{\bse_j}\times\bse_j\\
&=&\sum^3_{j=1}\bsA\bsB^\t\bse_j\times\bse_j,
\end{eqnarray*}
completing the proof.
\end{proof}

\begin{cor}
For $\bsA,\bsB\in\real^{3\times3}$, we have
\begin{eqnarray}
\sum^3_{k=1}\bsA\bse_k\times\bsB\bse_k =
-\frac12\Tr{\cF(\bsA\bsB^\t-\bsB\bsA^\t)} =
-\frac12\sum^3_{k=1}\Tr{\bsF_k(\bsA\bsB^\t-\bsB\bsA^\t)}\bse_k,
\end{eqnarray}
where
$$
\Tr{\cF(\bsA\bsB^\t-\bsB\bsA^\t)}:=(\Tr{\bsF_1(\bsA\bsB^\t-\bsB\bsA^\t)},\Tr{\bsF_2(\bsA\bsB^\t-\bsB\bsA^\t)},\Tr{\bsF_3(\bsA\bsB^\t-\bsB\bsA^\t)}).
$$
\end{cor}

\begin{proof}
Indeed, by Proposition~\ref{prop:crossproperty}, we get that
\begin{eqnarray*}
\sum^3_{k=1}\bsA\bse_k\times\bsB\bse_k =
\sum^3_{i=1}(\bsA\bsB^\t\bse_i)\times\bse_i=\sum^3_{i=1}\bse_i\times(\bsB\bsA^\t\bse_i)
=-\sum^3_{i=1}(\bsB\bsA^\t\bse_i)\times\bse_i,
\end{eqnarray*}
implying that
\begin{eqnarray*}
&&\sum^3_{k=1}\bsA\bse_k\times\bsB\bse_k =-\frac12
\sum^3_{i=1}\bse_i\times[(\bsA\bsB^\t-\bsB\bsA^\t)\bse_i]\\
&&=-\frac12(\sum^3_{i=1}\Innerm{\bse_i}{\bsF_1(\bsA\bsB^\t-\bsB\bsA^\t)}{\bse_i},\sum^3_{i=1}\Innerm{\bse_i}{\bsF_2(\bsA\bsB^\t-\bsB\bsA^\t)}{\bse_i},\sum^3_{i=1}\Innerm{\bse_i}{\bsF_3(\bsA\bsB^\t-\bsB\bsA^\t)}{\bse_i})\\
&&=-\frac12(\Tr{\bsF_1(\bsA\bsB^\t-\bsB\bsA^\t)},\Tr{\bsF_2(\bsA\bsB^\t-\bsB\bsA^\t)},\Tr{\bsF_3(\bsA\bsB^\t-\bsB\bsA^\t)}).
\end{eqnarray*}
This can be written down in a simplified notation:
$$
\Tr{\cF(\bsA\bsB^\t-\bsB\bsA^\t)}:=(\Tr{\bsF_1(\bsA\bsB^\t-\bsB\bsA^\t)},\Tr{\bsF_2(\bsA\bsB^\t-\bsB\bsA^\t)},\Tr{\bsF_3(\bsA\bsB^\t-\bsB\bsA^\t)}).
$$
We are done.
\end{proof}

\subsubsection{Auxiliary results}

To establish a rigorous relationship between Makhlin's invariants
and the Bargmann invariants under local unitary (LU)
transformations, we need to perform detailed calculations.
Throughout this process, numerous intriguing insights and findings
will emerge, which can be immediately utilized for simplifications
and reductions.

\begin{lem}\label{lem:A2}
For two given vectors
$\bsx=(x_1,x_2,x_3)^\t,\bsy=(y_1,y_2,y_3)^\t\in\real^3$, it holds
that
\begin{enumerate}[(i)]
\item $(\bsx\cdot\cF)^\t=-\bsx\cdot\cF$;
\item $(\bsx\cdot\cF)^\t\bsy=\bsx\times\bsy$;
\item
$\bsx\cdot\cF=\sum^3_{j=1}\out{\bse_j\times\bsx}{\bse_j}=\sum^3_{j=1}\out{\bse_j}{\bsx\times\bse_j}$.
\item
$(\bsx\cdot\cF)^\t(\bsy\cdot\cF)=\sum^3_{j=1}\bsF_j\out{\bsx}{\bsy}\bsF^\t_j=\Inner{\bsy}{\bsx}\I_3-\out{\bsy}{\bsx}$
and thus $\Inner{\bsx\cdot\cF}{\bsy\cdot\cF}=2\Inner{\bsx}{\bsy}$;
\item $(\bsx\times\bsy)\cdot\cF=\out{\bsx}{\bsy}-\out{\bsy}{\bsx}$.
\end{enumerate}
\end{lem}

\begin{proof}
For the first item, it is trivial result. For the second item, in
fact, we can check this identity directly as follows:
\begin{eqnarray*}
(\bsx\cdot\cF)^\t\bsy &=& -\Pa{\begin{array}{ccc}
 0 & x_3 & -x_2 \\
 -x_3 & 0 & x_1 \\
 x_2 & -x_1 & 0 \\
\end{array}}
\Pa{\begin{array}{c}
y_1\\
y_2\\
y_3
\end{array}} = -\Pa{\begin{array}{c}
x_3y_2-x_2y_3\\
x_1y_3-x_3y_1\\
x_2y_1-x_1y_2
\end{array}}\\
&=&\Pa{\Abs{\begin{array}{cc}
x_2 & x_3\\
y_2 & y_3
\end{array}},-\Abs{\begin{array}{cc}
x_1 & x_3\\
y_1 & y_3
\end{array}},\Abs{\begin{array}{cc}
x_1 & x_2\\
y_1 & y_2
\end{array}}}^\t\\
&=&\bsx\times\bsy.
\end{eqnarray*}
The third item can be also calculated immediately. Indeed, note that
$(\bsx\cdot\cF)^\t\bse_j=\bsx\times \bse_j$ or in Dirac notation,
$$
(\bsx\cdot\cF)^\t\ket{\bse_j}=\ket{\bsx\times\bse_j},\quad j=1,2,3,
$$
we get that
\begin{eqnarray*}
\bsx\cdot\cF&=&-(\bsx\cdot\cF)^\t\sum^3_{j=1}\proj{\bse_j}=-\sum^3_{j=1}(\bsx\cdot\cF)^\t\proj{\bse_j}\\
&=&-\sum^3_{j=1}\out{\bsx\times
\bse_j}{\bse_j}=\sum^3_{j=1}\out{\bse_j\times\bsx}{\bse_j}.
\end{eqnarray*}
Analogously,
\begin{eqnarray*}
\bsx\cdot\cF&=&\sum^3_{j=1}\proj{\bse_j}(\bsx\cdot\cF)=\sum^3_{j=1}\out{\bse_j}{\bsx\times
\bse_j}.
\end{eqnarray*}
For the 4th item, Furthermore,
\begin{eqnarray*}
(\bsx\cdot\cF)^\t(\bsy\cdot\cF)&=&(\bsx\cdot\cF)^\t\sum^3_{j=1}\proj{\bse_j}(\bsy\cdot\cF)=\sum^3_{j=1}(\bsx\cdot\cF)^\t\proj{\bse_j}(\bsy\cdot\cF)\\
&=&\sum^3_{j=1}\out{\bsx\times
\bse_j}{\bsy\times\bse_j}=\sum^3_{j=1}\out{\bse_j\times\bsx}{\bse_j\times\bsy}=\sum^3_{j=1}\bsF_j\out{\bsx}{\bsy}\bsF^\t_j.
\end{eqnarray*}
Note that $\Inner{\bsF_i}{\bsF_j}=2\delta_{ij}$. We get that
$\Inner{\bsx\cdot\cF}{\bsy\cdot\cF}=2\Inner{\bsx}{\bsy}$. For the
last item, we see that
\begin{eqnarray*}
(\bsx\times\bsy)\cdot\cF &=&
\sum^3_{k=1}\Innerm{\bsx}{\bsF_k}{\bsy}\bsF_k=
\sum^3_{k=1}\Tr{\bsF_k\out{\bsy}{\bsx}}\bsF_k=-\sum^3_{k=1}\Tr{\bsF_k\out{\bsx}{\bsy}}\bsF_k\\
&=&-\sum^3_{k=1}\frac12\Tr{\bsF_k(\out{\bsx}{\bsy}-\out{\bsy}{\bsx})}\bsF_k
=\out{\bsx}{\bsy}-\out{\bsy}{\bsx}.
\end{eqnarray*}
This completes the proof.
\end{proof}

In fact, the second item in Lemma~\ref{lem:A2} can be viewed as the
implementation of cross product by matrix multiplication. This
observation is simple but very important throughout this paper.

Another important fact is paramount in the following development. In
fact,
\begin{lem}\label{lem:A3}
For arbitrary two matrices $\bsM,\bsN\in\real^{3\times3}$ and any
two vectors $\bsx,\bsy\in\real^3$, it holds that
\begin{enumerate}[(i)]
\item $\Omega(\bsM,\bsM)=2\widehat\bsM$.
\item $\bsM(\bsx\cdot\cF)\bsN^\t + \bsN(\bsx\cdot\cF)\bsM^\t = (\Omega(\bsM,\bsN)\bsx)\cdot\cF$. In particular, for $\bsM=\bsN$, we get that
\begin{eqnarray}
\bsM(\bsx\cdot\cF)\bsM^\t =
\frac12(\Omega(\bsM,\bsM)\bsx)\cdot\cF=(\widehat\bsM\bsx)\cdot\cF.
\end{eqnarray}
\item $\bsM^\t[(\bsM\bsx)\cdot\cF]\bsM =\det(\bsM) (\bsx\cdot\cF)$ and
$\bsM[(\bsM^\t\bsx)\cdot\cF]\bsM^\t =\det(\bsM) (\bsx\cdot\cF)$.
\item $\Inner{(\bsx\cdot\cF)\bsM}{\bsM(\bsy\cdot\cF)}=2\Innerm{\bsx}{\widehat\bsM}{\bsy}$.
\end{enumerate}
\end{lem}

\begin{proof}
For the first item, the proof can be obtained immediately by direct
computation. Indeed, using Proposition~\ref{prop:Omega}, for
$\bsM=\bsN$, we get that
\begin{eqnarray*}
\Omega(\bsM,\bsM)_{p,q} &=& -\Inner{\bsF_p\bsM\bsF_q}{\bsM} =
\Inner{\bsF_p\bsM}{\bsM\bsF_q} \\
&=& \Inner{(\bse_p\cdot\cF)\bsM}{\bsM(\bse_q\cdot\cF)} =
2\Innerm{\bse_p}{\widehat\bsM}{\bse_q},
\end{eqnarray*}
implying that $\Omega(\bsM,\bsM)=2\widehat\bsM$. For the second
item, it is easily seen that
\begin{eqnarray*}
\Pa{\bsM(\bsx\cdot\cF)\bsN^\t + \bsN(\bsx\cdot\cF)\bsM^\t}^\t =
-\Pa{\bsM(\bsx\cdot\cF)\bsN^\t + \bsN(\bsx\cdot\cF)\bsM^\t}.
\end{eqnarray*}
Thus it can be decomposed as
\begin{eqnarray*}
\bsM(\bsx\cdot\cF)\bsN^\t + \bsN(\bsx\cdot\cF)\bsM^\t=
\sum^3_{k=1}c_k(\bsM,\bsN)\bsF_k,
\end{eqnarray*}
where the coefficients $c_k$ can be identified with
\begin{eqnarray*}
c_k &=& -\frac12\Tr{\bsM(\bsx\cdot\cF)\bsN^\t\bsF_k}-\frac12\Tr{\bsN(\bsx\cdot\cF)\bsM^\t\bsF_k}\\
&=& \Inner{\bse_k}{\Omega(\bsM,\bsN)\bsx},
\end{eqnarray*}
implying that $\bsM(\bsx\cdot\cF)\bsN^\t + \bsN(\bsx\cdot\cF)\bsM^\t
= (\Omega(\bsM,\bsN)\bsx)\cdot\cF$. In particular, for $\bsM=\bsN$,
the desired identity follows immediately from
$\Omega(\bsM,\bsM)=2\widehat\bsM$. For the third item, we see from
the obtained result in (ii) that
\begin{eqnarray*}
\bsM^\t[(\bsM\bsx)\cdot\cF]\bsM =
({\widehat\bsM}^\t\bsM\bsx)\cdot\cF = \det(\bsM)(\bsx\cdot\cF).
\end{eqnarray*}
For the 4th item,
\begin{eqnarray*}
\Inner{(\bsx\cdot\cF)\bsM}{\bsM(\bsy\cdot\cF)}&=& \Tr{\bsM^\t(\bsx\cdot\cF)^\t\bsM(\bsy\cdot\cF)} = \Tr{(\bsx\cdot\cF)^\t\bsM(\bsy\cdot\cF)\bsM^\t}\\
&=&\Tr{(\bsx\cdot\cF)^\t[(\widehat\bsM\bsy)\cdot\cF]} = \Inner{\bsx\cdot\cF}{(\widehat\bsM\bsy)\cdot\cF}\\
&=&2\Innerm{\bsx}{\widehat\bsM}{\bsy}.
\end{eqnarray*}
In the first equality, we used the definition of Hilbert-Schmidt
inner product. For the second equality, we used the cyclicity of
trace. In the third equality, we used the obtained result in (ii).
In the last equality, we used the fact obtained in (iii) of
Lemma~\ref{lem:A2}.
\end{proof}

\begin{cor}\label{cor:B3a}
For an arbitrary invertible matrix $\bsL\in\real^{3\times3}$ and any
two vectors $\bsu,\bsv\in\real^3$, we have that
\begin{enumerate}[(i)]
\item $\bsL(\bsu\times\bsv) =\widehat\bsL\bsu\times
(\bsL^\t)^{-1}\bsv$, in fact, we see that
\begin{eqnarray}
\bsL(\bsu\times\bsv) &=& \widehat\bsL\bsu\times (\bsL^\t)^{-1}\bsv\\
&=& \det(\bsL)\Pa{(\bsL^\t)^{-1}\bsu\times (\bsL^\t)^{-1}\bsv}\\
&=& \frac1{\det(\bsL)}\Pa{\widehat\bsL\bsu\times \widehat\bsL\bsv}.
\end{eqnarray}
In particular, for $\bsR\in\SO(3)$, the special orthogonal group of
order $3$, we recover the well-known formula:
\begin{eqnarray}
\bsR(\bsu\times\bsv)=\bsR\bsu\times\bsR\bsv.
\end{eqnarray}
\item
$\bsL\bsu\times\bsL\bsv=\det(\bsL)(\bsL^\t)^{-1}(\bsu\times\bsv)=\widehat\bsL(\bsu\times\bsv)$.
\end{enumerate}
\end{cor}

\begin{proof}
It suffices to show that
$\bsL(\bsu\times\bsv)=\widehat\bsL\bsu\times (\bsL^\t)^{-1}\bsv$.
Indeed, via the fact that $\bsu\times\bsv=(\bsu\cdot\cF)^\t\bsv$,
using the result in (ii) of Lemma~\ref{lem:A3}
\begin{eqnarray*}
&&\bsL(\bsu\times\bsv)=\bsL(\bsu\cdot\cF)^\t\bsv =
\bsL(\bsu\cdot\cF)^\t\bsL^\t(\bsL^\t)^{-1}\bsv=[\bsL(\bsu\cdot\cF)\bsL^\t]^\t(\bsL^\t)^{-1}\bsv\\
&&=[(\widehat\bsL\bsu)\cdot\cF]^\t(\bsL^\t)^{-1}\bsv =
\widehat\bsL\bsu\times(\bsL^\t)^{-1}\bsv.
\end{eqnarray*}
Due to the fact that $\widehat\bsL\bsL^\t = \det(\bsL)\I_3$ and
$\bsL$ is invertible, we get the other two forms of this formula. In
particular, for $\bsL=\bsR\in\SO(3)$, then $\det(\bsR)=1$ and
$\widehat\bsR=\bsR$, which leads to the desired identity.
\end{proof}

The above results are obtained under the invertibility condition. In
fact, we can remove such condition, that is, the following
identities holds for any matrix $\bsL\in\real^{3\times3}$:
\begin{eqnarray}
\det(\bsL)\bsL(\bsu\times\bsv) &=&
\widehat\bsL\bsu\times\widehat\bsL\bsv,\\
\bsL\bsu\times\bsL\bsv&=&\widehat\bsL(\bsu\times\bsv).
\end{eqnarray}

\begin{cor}\label{cor:B3b}
For any two matrices $\bsM,\bsN\in\real^{3\times3}$ and any two
vectors $\bsu,\bsv\in\real^3$, it holds that
\begin{eqnarray}
\bsM\bsu\times\bsN\bsv + \bsN\bsu\times\bsM\bsv =
\Omega(\bsM,\bsN)(\bsu\times \bsv).
\end{eqnarray}
In particular, for $\bsM=\bsN$, we get that
$\bsM\bsu\times\bsM\bsv=\frac12\Omega(\bsM,\bsM)(\bsu\times\bsv)=\widehat\bsM(\bsu\times\bsv)$.
\end{cor}

\begin{proof}
In fact,
\begin{eqnarray*}
\bsM\bsu\times\bsN\bsv + \bsN\bsu\times\bsM\bsv &=&
\Pa{\Innerm{\bsM\bsu}{\bsF_1}{\bsN\bsv},\Innerm{\bsM\bsu}{\bsF_2}{\bsN\bsv},\Innerm{\bsM\bsu}{\bsF_3}{\bsN\bsv}}^\t\\
&&+\Pa{\Innerm{\bsN\bsu}{\bsF_1}{\bsM\bsv},\Innerm{\bsN\bsu}{\bsF_2}{\bsM\bsv},\Innerm{\bsN\bsu}{\bsF_3}{\bsM\bsv}}^\t,
\end{eqnarray*}
which is equal to
$$
\Pa{\Innerm{\bsu}{\bsM^\t\bsF_1\bsN+\bsN^\t\bsF_1\bsM}{\bsv},\Innerm{\bsu}{\bsM^\t\bsF_2\bsN+\bsN^\t\bsF_2\bsM}{\bsv},\Innerm{\bsu}{\bsM^\t\bsF_3\bsN+\bsN^\t\bsF_3\bsM}{\bsv}}^\t.
$$
Now we can easily check that
\begin{eqnarray*}
\bsM^\t\bsF_i\bsN+\bsN^\t\bsF_i\bsM=-\sum^3_{j=1}\Inner{\bsF_i\bsM\bsF_j}{\bsN}\bsF_j,
\end{eqnarray*}
implying that
\begin{eqnarray*}
&&\Innerm{\bsu}{\bsM^\t\bsF_i\bsN+\bsN^\t\bsF_i\bsM}{\bsv}=-\sum^3_{j=1}\Inner{\bsF_i\bsM\bsF_j}{\bsN}\Innerm{\bsu}{\bsF_j}{\bsv}=\sum^3_{j=1}\Omega(\bsM,\bsN)_{i,j}\Innerm{\bsu}{\bsF_j}{\bsv}\\
&&=[\Omega(\bsM,\bsN)(\bsu\times\bsv)]_i.
\end{eqnarray*}
That is, the desired result is true.
\end{proof}

\begin{cor}
For any two matrices $\bsM,\bsN\in\real^{3\times3}$, it holds that
\begin{eqnarray}
\Omega(\bsM,\bsN) =
\widehat{\bsM+\bsN}-\Pa{\widehat\bsM+\widehat\bsN}.
\end{eqnarray}
\end{cor}

\begin{proof}
Indeed, using Corollary~\ref{cor:B3b}, we get that
\begin{eqnarray*}
(\bsM+\bsN)\bsu\times(\bsM+\bsN)\bsv&=&\widehat{\bsM+\bsN}(\bsu\times\bsv),\\
\bsM\bsu\times \bsM\bsv&=&\widehat{\bsM}(\bsu\times\bsv),\\
\bsN\bsu\times \bsN\bsv&=&\widehat{\bsN}(\bsu\times\bsv),
\end{eqnarray*}
and thus
\begin{eqnarray*}
\bsM\bsu\times\bsN\bsv + \bsN\bsu\times\bsM\bsv
&=&(\bsM+\bsN)\bsu\times(\bsM+\bsN)\bsv - \bsM\bsu\times \bsM\bsv -
\bsN\bsu\times \bsN\bsv\\
&=&\widehat{\bsM+\bsN}(\bsu\times\bsv)-\widehat{\bsM}(\bsu\times\bsv)-\widehat{\bsN}(\bsu\times\bsv)\\
&=&
\Pa{\widehat{\bsM+\bsN}-\widehat{\bsM}-\widehat{\bsN}}(\bsu\times\bsv)
= \Omega(\bsM,\bsN)(\bsu\times\bsv),
\end{eqnarray*}
implying that
$\Omega(\bsM,\bsN)=\widehat{\bsM+\bsN}-\widehat{\bsM}-\widehat{\bsN}$.
\end{proof}

\begin{cor}\label{cor:B3c}
For any matrix $\bsM\in\real^{3\times3}$ and any vectors
$\bsx,\bsy\in\real^3$, it holds that
\begin{enumerate}[(i)]
\item $[(\bsM\bsy)\cdot\cF]\bsM=\widehat\bsM(\bsy\cdot\cF)$.
\item $\bsM[(\bsM^\t\bsx)\cdot\cF]=(\bsx\cdot\cF)\widehat\bsM$.
\end{enumerate}
\end{cor}

\begin{proof}
We can prove these results in two steps:
\begin{itemize}
\item Assume that $\bsM$ is invertible. Then by the result in (iii)
of Lemma~\ref{lem:A3}, we get that
\begin{eqnarray*}
\bsM^\t[(\bsM\bsy)\cdot\cF]\bsM = \det(\bsM)(\bsy\cdot\cF) =
\bsM^\t\widehat\bsM(\bsy\cdot\cF).
\end{eqnarray*}
Because $\bsM$ is invertible, i.e., $\bsM^\t$ is also invertible, we
get that
$$
[(\bsM\bsy)\cdot\cF]\bsM = \widehat\bsM(\bsy\cdot\cF).
$$
Analogously, we also have that
$\bsM[(\bsM^\t\bsx)\cdot\cF]=(\bsx\cdot\cF)\widehat\bsM$.
\item Now if $\bsM$ is not invertible, then we can take a net by
using Singular Value Decomposition such that $\bsM$ can be
approximated in any precision by such net. Indeed, via Singular
Value Decomposition, there there exist two orthogonal matrices
$\bsP$ and $\bsQ$ in $\O(3)$, the orthogonal group, such that
$\bsM=\bsP\Sigma\bsQ^\t$, where $\Sigma=\diag(m_1,m_2,m_3)$ consists
of singular values of $\bsM$. Let such net
$\set{\bsM_\epsilon:\epsilon>0}$ be given by
$\bsM_\epsilon=\bsP(\Sigma+\epsilon\I_3)\bsQ^\t$ for small enough
$\epsilon>0$. Now $\lim_{\epsilon\to0^+}\bsM_\epsilon=\bsM$ and
$$
[(\bsM_\epsilon\bsy)\cdot\cF]\bsM_\epsilon =
\widehat\bsM_\epsilon(\bsy\cdot\cF).
$$
The proof can be finished by taking the limit for $\epsilon\to0^+$
on both sides of the above expression due to the continuity argument
and the fact that
\begin{eqnarray}\label{eq:cofactormatrix}
\lim_{\epsilon\to0^+}\widehat\bsM_\epsilon=\widehat\bsM.
\end{eqnarray}
To this end, using the result (i) in Lemma~\ref{lem:adjugate}, we
see that
\begin{eqnarray*}
\lim_{\epsilon\to0^+}\Tr{\widehat\bsM_\epsilon} &=&
\frac12\Pa{\lim_{\epsilon\to0^+}\Tr{\bsM_\epsilon}^2-\lim_{\epsilon\to0^+}\Tr{\bsM_\epsilon^2}}\\
&=&
\frac12\Pa{\lim_{\epsilon\to0^+}\Tr{\bsM}^2-\lim_{\epsilon\to0^+}\Tr{\bsM^2}}
=\Tr{\widehat\bsM}.
\end{eqnarray*}

By Proposition~\ref{prop:adjugate}, we get
\begin{eqnarray*}
\lim_{\epsilon\to0^+}\widehat\bsM_\epsilon  &=&
\lim_{\epsilon\to0^+}
\Pa{\bsM^2_\epsilon-\Tr{\bsM_\epsilon}\bsM_\epsilon+\Tr{\widehat\bsM_\epsilon}\I_3}^\t\\
&=&
\Pa{\lim_{\epsilon\to0^+}\bsM^2_\epsilon-\lim_{\epsilon\to0^+}\Tr{\bsM_\epsilon}\lim_{\epsilon\to0^+}\bsM_\epsilon+\lim_{\epsilon\to0^+}\Tr{\widehat\bsM_\epsilon}\I_3}^\t\\
&=&\Pa{\bsM^2-\Tr{\bsM}\bsM+\Tr{\widehat\bsM}\I_3}^\t =
\widehat\bsM.
\end{eqnarray*}
\end{itemize}
The proof is complete.
\end{proof}
Next we summarize important properties concerning $\Omega$.
\begin{lem}
For $\Omega$, defined in Eq.~\eqref{eq:Omegasymbol}, it holds that
\begin{enumerate}[(i)]
\item $\Omega(\bsT,\out{\bsa}{\bsb}) = (\bsa\cdot\cF)\bsT(\bsb\cdot\cF)^\t$.
\item $\Omega(\bsM,\widehat\bsT)=\Tr{\bsM^\t\bsT}\bsT - \bsT\bsM^\t\bsT$, in particular, $\Omega(\bsT,\widehat\bsT) = \norm{\bsT}^2\bsT -
\bsT\bsT^\t\bsT$.
\item $\Omega(\bsT,\bsA\bsT) = \Tr{\bsA}\widehat\bsT - \bsA^\t\widehat\bsT
$; in particular, $\Omega(\bsT,\bsT\bsT^\t\bsT) =
\norm{\bsT}^2\widehat\bsT- \det(\bsT)\bsT$.
\item $\Omega(\bsT,\bsT\bsB)=\Tr{\bsB}\widehat\bsT - \widehat\bsT
\bsB^\t$.
\item $\Omega(\bsT,(\bsr\cdot\cF)\bsT(\bss\cdot\cF)^\t) = \Innerm{\bsr}{\bsT}{\bss}\bsT + \norm{\bsT}^2\out{\bsr}{\bss} -
(\out{\bsr}{\bss}\bsT^\t\bsT +\bsT\bsT^\t\out{\bsr}{\bss})$.
\item $\Omega(\bsT,\bsx\cdot\cF)=\bsT\bsx\cdot\cF+\out{\bsn}{\bsx}$, where $\bsn=\sum^3_{i=1}\bsT\bse_i\times\bse_i$ is determined from $\bsT-\bsT^\t=\bsn\cdot\cF$.
\item $\Omega(\bsx\cdot\cF,\bsy\cdot\cF)=\out{\bsx}{\bsy}+\out{\bsy}{\bsx}$. In particular, $\Omega(\bsx\cdot\cF,\bsx\cdot\cF)=2\proj{\bsx}$.
\end{enumerate}
\end{lem}

\begin{proof}
(i) Indeed, for $\bsT'=\out{\bsa}{\bsb}$, we see that
\begin{eqnarray*}
\Omega(\bsT,\bsT') &=& \Pa{\begin{array}{c}
          a_3\bse^\t_2\bsT\times\bsb^\t \\
          a_1\bse^\t_3\bsT\times\bsb^\t \\
          a_2\bse^\t_1\bsT\times\bsb^\t
        \end{array}
}-\Pa{\begin{array}{c}
          a_2\bse^\t_3\bsT\times\bsb^\t \\
          a_3\bse^\t_1\bsT\times\bsb^\t \\
          a_1\bse^\t_2\bsT\times\bsb^\t
        \end{array}
}\\
&=&\diag(a_3,a_1,a_2)\Pa{\begin{array}{ccc}
                           0 & 1 & 0 \\
                           0 & 0 & 1 \\
                           1 & 0 & 0
                         \end{array}
}\Pa{\begin{array}{c}
       \bse^\t_1\bsT\times\bsb^\t \\
       \bse^\t_2\bsT\times\bsb^\t \\
       \bse^\t_3\bsT\times\bsb^\t
     \end{array}
}\\
&&-\diag(a_2,a_3,a_1)\Pa{\begin{array}{ccc}
                           0 & 0 & 1 \\
                           1 & 0 & 0 \\
                           0 & 1 & 0
                         \end{array}
}\Pa{\begin{array}{c}
       \bse^\t_1\bsT\times\bsb^\t \\
       \bse^\t_2\bsT\times\bsb^\t \\
       \bse^\t_3\bsT\times\bsb^\t
     \end{array}
},
\end{eqnarray*}
which is just equal to $(\bsa\cdot\cF)\bsT(\bsb\cdot\cF)^\t$, where
we used the facts that
$$
\diag(a_3,a_1,a_2)\Pa{\begin{array}{ccc}
                           0 & 1 & 0 \\
                           0 & 0 & 1 \\
                           1 & 0 & 0
                         \end{array}
}-\diag(a_2,a_3,a_1)\Pa{\begin{array}{ccc}
                           0 & 0 & 1 \\
                           1 & 0 & 0 \\
                           0 & 1 & 0
                         \end{array}
}=\bsa\cdot\cF
$$
and $\Pa{\begin{array}{c}
       \bse^\t_1\bsT\times\bsb^\t \\
       \bse^\t_2\bsT\times\bsb^\t \\
       \bse^\t_3\bsT\times\bsb^\t
     \end{array}
}=\bsT(\bsb\cdot\cF)^\t$. \\
(ii) The correctness of this result can be directly checked by
\textsc{Mathematica}. In what follows, we infer it by analytical
method. In fact, using the result obtained in (i) previously,
\begin{eqnarray*}
\Omega(\out{\bsa}{\bsb},\widehat\bsT)
&=&\Omega(\widehat\bsT,\out{\bsa}{\bsb}) =
(\bsa\cdot\cF)\widehat\bsT(\bsb\cdot\cF)^\t\\
&=&(\bsa\cdot\cF)(\bsT\bsb\cdot\cF)^\t\bsT =
(\Innerm{\bsa}{\bsT}{\bsb}\I_3-\bsT\out{\bsb}{\bsa})\bsT\\
&=&\Tr{\out{\bsb}{\bsa}\bsT}\bsT - \bsT\out{\bsb}{\bsa}\bsT.
\end{eqnarray*}
Here in the third equality, we used the first property in
Corollary~\ref{cor:B3c}; and in the 4th equality, we used the third
property in Lemma~\ref{lem:A2}. Now using Singular Value
Decomposition of $\bsM$: $\bsM=\sum^3_{j=1}s_j\out{\bsa_j}{\bsb_j}$,
we can finish the proof:
$$
\Omega(\bsM,\widehat\bsT)=\Tr{\bsM^\t\bsT}\bsT - \bsT\bsM^\t\bsT.
$$
Indeed, by the bi-linearity of $\Omega(\cdot,\cdot)$,
\begin{eqnarray*}
\Omega(\bsT,\widehat\bsT) &=&\sum^3_{j=1}s_j
\Omega(\out{\bsa_j}{\bsb_j},\widehat\bsT) =\sum^3_{j=1}s_j
\Pa{\Tr{\out{\bsb_j}{\bsa_j}\bsT}\bsT -
\bsT\out{\bsb_j}{\bsa_j}\bsT}\\
&=& \Tr{\sum^3_{j=1}s_j\out{\bsb_j}{\bsa_j}\bsT}\bsT -
\bsT\sum^3_{j=1}s_j\out{\bsb_j}{\bsa_j}\bsT\\
&=&\Tr{\bsT^\t\bsT}\bsT-\bsT\bsT^\t\bsT =
\norm{\bsT}^2\bsT-\bsT\bsT^\t\bsT.
\end{eqnarray*}
(iii) We prove a stronger result: If $\bsT'=\bsA\bsT$ for any
$3\times 3$ real matrix $\bsA$, where $\bse^\t_k\bsA =
(a_{k1},a_{k2},a_{k3})(k\in[3])$,
\begin{eqnarray*}
&&\Omega(\bsT,\bsA\bsT) = \Pa{\begin{array}{c}
      \bse^\t_2\bsT\times(a_{31}\bse^\t_1\bsT+a_{33}\bse^\t_3\bsT)+(a_{21}\bse^\t_1\bsT+a_{22}\bse^\t_2\bsT)\times\bse^\t_3\bsT  \\
      \bse^\t_3\bsT\times(a_{11}\bse^\t_1\bsT+a_{12}\bse^\t_2\bsT)+(a_{32}\bse^\t_2\bsT+a_{33}\bse^\t_3\bsT)\times\bse^\t_1\bsT \\
      \bse^\t_1\bsT\times(a_{22}\bse^\t_2\bsT+a_{23}\bse^\t_3\bsT)+(a_{11}\bse^\t_1\bsT+a_{13}\bse^\t_3\bsT)\times\bse^\t_2\bsT
    \end{array}
}\\
&&=\Pa{\begin{array}{c}
      (a_{22}+a_{33})\bse^\t_2\bsT\times\bse^\t_3\bsT  \\
      (a_{11}+a_{33})\bse^\t_3\bsT\times\bse^\t_1\bsT\\
      (a_{11}+a_{22})\bse^\t_1\bsT\times\bse^\t_2\bsT
    \end{array}
}+\Pa{\begin{array}{c}
      a_{31}\bse^\t_2\bsT\times\bse^\t_1\bsT  \\
      a_{12}\bse^\t_3\bsT\times\bse^\t_2\bsT \\
      a_{23}\bse^\t_1\bsT\times\bse^\t_3\bsT
    \end{array}
}+\Pa{\begin{array}{c}
      a_{21}\bse^\t_1\bsT\times\bse^\t_3\bsT  \\
      a_{32}\bse^\t_2\bsT\times\bse^\t_1\bsT \\
      a_{13}\bse^\t_3\bsT\times\bse^\t_2\bsT
    \end{array}
}\\
&&=[\Tr{\bsA}-\diag(a_{11},a_{22},a_{33})]\widehat\bsT
-\Pa{\begin{array}{c}
      a_{31}\bse^\t_1\bsT\times\bse^\t_2\bsT  \\
      a_{12}\bse^\t_2\bsT\times\bse^\t_3\bsT \\
      a_{23}\bse^\t_3\bsT\times\bse^\t_1\bsT
    \end{array}
}-\Pa{\begin{array}{c}
      a_{21}\bse^\t_3\bsT\times\bse^\t_1\bsT  \\
      a_{32}\bse^\t_1\bsT\times\bse^\t_2\bsT \\
      a_{13}\bse^\t_2\bsT\times\bse^\t_3\bsT
    \end{array}
}\\
&&=\Tr{\bsA}\widehat\bsT -\diag(a_{11},a_{22},a_{33})\widehat\bsT -
\Pa{\begin{array}{ccc}
                                                                      0 & a_{21} & a_{31} \\
                                                                      a_{12} & 0 & a_{32} \\
                                                                      a_{13} & a_{23} &
                                                                      0
                                                                    \end{array}}\widehat\bsT
=\Tr{\bsA}\widehat\bsT - \bsA^\t\widehat\bsT.
\end{eqnarray*}
That is, for $\bsT'=\bsA\bsT$,
\begin{eqnarray}
\Omega(\bsT,\bsA\bsT) = \Tr{\bsA}\widehat\bsT - \bsA^\t\widehat\bsT.
\end{eqnarray}
Letting in the above $\bsA=\bsT\bsT^\t$, we get that
$$
\Tr{\bsA}\widehat\bsT - \bsA^\t\widehat\bsT =
\Tr{\bsT\bsT^\t}\widehat\bsT -
\bsT\bsT^\t\widehat\bsT=\Inner{\bsT}{\bsT}\widehat\bsT-
\det(\bsT)\bsT.
$$
The another approach to this result can be described as follows.
Indeed, $\bsA$ can be decomposed as
$\bsA=\sum^3_{i=1}s_i\out{\bsx_i}{\bsy_i}$ by Singular Value
Decomposition. Then
\begin{eqnarray*}
\Omega(\bsT,\bsA\bsT) &=&
\sum^3_{i=1}s_i\Omega(\bsT,\out{\bsx_i}{\bsy_i}\bsT) =
\sum^3_{j=1}s_i (\bsx_i\cdot\cF)\bsT((\bsT^\t\bsy_j)\cdot\cF)^\t\\
&=& \sum^3_{i=1}s_i
(\bsx_i\cdot\cF)(\bsy_i\cdot\cF)^\t\widehat\bsT=\sum^3_{i=1}s_i\Pa{\Inner{\bsy_i}{\bsx_i}\I_3-\out{\bsy_i}{\bsx_i}}\widehat\bsT\\
&=&\Tr{\bsA}\widehat\bsT - \bsA^\t\widehat\bsT.
\end{eqnarray*}
(iv) By Singular Value Decomposition of $\bsB$,
$\bsB=\sum^3_{i=1}s_j\out{\bsx_j}{\bsy_j}$. Now
\begin{eqnarray*}
\Omega(\bsT,\bsT\bsB) &=&
\sum^3_{j=1}s_j\Omega(\bsT,\bsT\out{\bsx_j}{\bsy_j}) =
\sum^3_{j=1}s_j\Omega(\bsT,\out{\bsT\bsx_j}{\bsy_j}).
\end{eqnarray*}
Using (i), we get that
\begin{eqnarray*}
\Omega(\bsT,\bsT\bsB) &=&
\sum^3_{j=1}s_j\Omega(\bsT,\out{\bsT\bsx_j}{\bsy_j}) =
\sum^3_{j=1}s_j ((\bsT\bsx_j)\cdot\cF)\bsT(\bsy_j\cdot\cF)^\t\\
&=& \sum^3_{j=1}s_j
\widehat\bsT(\bsx_j\cdot\cF)(\bsy_j\cdot\cF)^\t=\widehat\bsT
\sum^3_{j=1}s_j (\bsx_j\cdot\cF)(\bsy_j\cdot\cF)^\t\\
&=&\widehat\bsT\sum^3_{j=1}s_j\Pa{\Inner{\bsy_j}{\bsx_j}\I_3-\out{\bsy_j}{\bsx_j}}=\Tr{\bsB}\widehat\bsT
- \widehat\bsT\bsB^\t.
\end{eqnarray*}
(v) Note that
$$
(\bsr\cdot\cF)\bsT(\bss\cdot\cF)^\t=\Pa{\begin{array}{c}
                                            \bse^\t_1(\bsr\cdot\cF)\bsT\times\bss^\t \\
                                            \bse^\t_2(\bsr\cdot\cF)\bsT\times\bss^\t \\
                                            \bse^\t_3(\bsr\cdot\cF)\bsT\times\bss^\t
                                          \end{array}
}
$$
implies that
$$
\bse^\t_k(\bsr\cdot\cF)\bsT(\bss\cdot\cF)^\t=\bse^\t_k(\bsr\cdot\cF)\bsT\times\bss^\t.
$$
Using the facts that
\begin{eqnarray*}
(\bsu\times\bsv)\times\bsw = \Inner{\bsw}{\bsu}\bsv -
\Inner{\bsw}{\bsv}\bsu,\\
\bsu\times(\bsv\times\bsw) = \Inner{\bsu}{\bsw}\bsv -
\Inner{\bsu}{\bsv}\bsw,
\end{eqnarray*}
we get that
\begin{eqnarray*}
&&\Omega(\bsT,(\bsr\cdot\cF)\bsT(\bss\cdot\cF)^\t) =
\Pa{\begin{array}{c}
 (\bse^\t_2(\bsr\cdot\cF)\bsT\times\bss^\t)\times\bse^\t_3\bsT
 + \bse^\t_2\bsT\times(\bse^\t_3(\bsr\cdot\cF)\bsT\times\bss^\t) \\
 (\bse^\t_3(\bsr\cdot\cF)\bsT\times\bss^\t)\times\bse^\t_1\bsT
 + \bse^\t_3\bsT\times(\bse^\t_1(\bsr\cdot\cF)\bsT\times\bss^\t)\\
(\bse^\t_1(\bsr\cdot\cF)\bsT\times\bss^\t)\times\bse^\t_2\bsT
 + \bse^\t_1\bsT\times(\bse^\t_2(\bsr\cdot\cF)\bsT\times\bss^\t)
                                                         \end{array}
}\\
&&= - \Pa{\begin{array}{c}
            (\bsT\bss)_3\bse^\t_2(\bsr\cdot\cF)\bsT - (\bsT\bss)_2\bse^\t_3(\bsr\cdot\cF)\bsT\\
             (\bsT\bss)_1\bse^\t_3(\bsr\cdot\cF)\bsT - (\bsT\bss)_3\bse^\t_1(\bsr\cdot\cF)\bsT\\
            (\bsT\bss)_2\bse^\t_1(\bsr\cdot\cF)\bsT - (\bsT\bss)_1\bse^\t_2(\bsr\cdot\cF)\bsT
          \end{array}
}+\Pa{\begin{array}{c}
        \Innerm{\bse_2}{\set{\bsr\cdot\cF,\bsT\bsT^\t}}{\bse_3} \\
        \Innerm{\bse_3}{\set{\bsr\cdot\cF,\bsT\bsT^\t}}{\bse_1} \\
        \Innerm{\bse_1}{\set{\bsr\cdot\cF,\bsT\bsT^\t}}{\bse_2}
      \end{array}
}\bss^\t\\
&&=[(\bsT\bss)\cdot\cF]^\t(\bsr\cdot\cF)\bsT+(\Inner{\bsT}{\bsT}\I_3-\bsT\bsT^\t)\out{\bsr}{\bss}\\
&&=\Innerm{\bsr}{\bsT}{\bss}\bsT - \out{\bsr}{\bss}\bsT^\t\bsT +
(\norm{\bsT}^2\I_3-\bsT\bsT^\t)\out{\bsr}{\bss}.
\end{eqnarray*}
Here $\Set{\bsA,\bsB}:=\bsA\bsB+\bsB\bsA$. Other items can be
checked by direct calculation. This completes the proof.
\end{proof}

\begin{lem}
For $\Omega$, defined in Eq.~\eqref{eq:Omegasymbol}, it holds that
\begin{enumerate}[(i)]
\item $\Omega(\bsA\widehat\bsT,\bsB) =\Omega(\bsA,\bsB\bsT^\t)\bsT$.
\item $\Omega(\widehat\bsT\bsA,\bsB) =\bsT\Omega(\bsA,\bsT^\t\bsB)$.
\end{enumerate}
\end{lem}

\begin{proof}
For the first item, note that
\begin{eqnarray*}
&&\Omega(\bsA\widehat\bsT,\out{\bsa}{\bsb}) =
(\bsa\cdot\cF)\bsA\widehat\bsT(\bsb\cdot\cF)^\t =
(\bsa\cdot\cF)\bsA(\bsT\bsb\cdot\cF)^\t\bsT\\
&&=\Omega(\bsA,\out{\bsa}{\bsb}\bsT^\t)\bsT,
\end{eqnarray*}
implying that
\begin{eqnarray*}
\Omega(\bsA\widehat\bsT,\bsB) =\Omega(\bsA,\bsB\bsT^\t)\bsT.
\end{eqnarray*}
For the second item, we see that
\begin{eqnarray*}
&&\Omega(\widehat\bsT\bsA,\out{\bsa}{\bsb}) =
(\bsa\cdot\cF)\widehat\bsT\bsA(\bsb\cdot\cF)^\t =
\bsT(\bsT^\t\bsa\cdot\cF)\bsA(\bsb\cdot\cF)^\t\\
&&=\bsT\Omega(\bsA,\bsT^\t\out{\bsa}{\bsb}),
\end{eqnarray*}
implying that $\Omega(\widehat\bsT\bsA,\bsB)
=\bsT\Omega(\bsA,\bsT^\t\bsB)$.
\end{proof}

\subsubsection{Recurrence relation for the matrix power}

Let $\bsX^1\approx
\Pa{t^{(1)},\bsr^{(1)},\bss^{(1)},\bsT^{(1)}}=(t,\bsr,\bss,\bsT)$
and $\bsX^k\approx \Pa{t^{(k)},\bsr^{(k)},\bss^{(k)},\bsT^{(k)}}$,
i.e.,
\begin{eqnarray}
\bsX^k =t^{(k)}\I_4
+\bsr^{(k)}\cdot\bsigma\ot\I_2+\I_2\ot\bss^{(k)}\cdot\bsigma
+\sum^3_{i,j=1}t^{(k)}_{ij}\sigma_i\ot\sigma_j\quad(k\geqslant1),
\end{eqnarray}
where $\bsT^{(k)}:=\Pa{t^{(k)}_{ij}}_{3\times 3}$. By
Lemma~\ref{lem:A1}, we get that
\begin{cor}\label{cor:nthpower}
The recurrence relations of coefficients between
$\bsX^{k+1}=\bsX^k\bsX\approx
\Pa{t^{(k+1)},\bsr^{(k+1)},\bss^{(k+1)},\bsT^{(k+1)}}$ and
$\bsX^k\approx \Pa{t^{(k)},\bsr^{(k)},\bss^{(k)},\bsT^{(k)}}$ can be
identified as:
\begin{eqnarray}
\begin{cases}
t^{(k+1)} &= t^{(k)}t+\Inner{\bsr^{(k)}}{\bsr}+\Inner{\bss^{(k)}}{\bss}+\Inner{\bsT^{(k)}}{\bsT}\\
\bsr^{(k+1)} &= t\bsr^{(k)}+t^{(k)}\bsr+
\bsT\bss^{(k)}+\bsT^{(k)}\bss\\
\bss^{(k+1)} &= t\bss^{(k)}+t^{(k)}\bss+
\bsT^\t\bsr^{(k)}+{\bsT^{(k)}}^\t\bsr\\
\bsT^{(k+1)}
&=\out{\bsr^{(k)}}{\bss}+\out{\bsr}{\bss^{(k)}}+t\bsT^{(k)}+t^{(k)}\bsT
- \Omega(\bsT^{(k)},\bsT)
\end{cases}
\end{eqnarray}
where $k\geqslant1$.
\end{cor}

\begin{proof}
Using Corollary~\ref{cor:B1}, we see that $[\bsX^k,\bsX]=0$ if and
only if
\begin{eqnarray*}
\bsr^{(k)}\times\bsr+\sum^3_{i=1}\bsT^{(k)}\bse_i\times\bsT\bse_i&=&\zero,\\
\bss^{(k)}\times\bss+\sum^3_{i=1}{\bsT^{(k)}}^\t\bse_i\times\bsT^\t\bse_i&=&\zero,\\
\Psi(\bsr^{(k)},\bsT,\bss^{(k)})&=&\Psi(\bsr,\bsT^{(k)},\bss).
\end{eqnarray*}
The recurrence relation is obtained immediately.
\end{proof}

Using the previous results, we can list here the coefficients of
$\bsX^k(1\leqslant k\leqslant4)$ below:
\begin{enumerate}[(a)]
\item For $k=2$,
$\bsX^2\approx\Pa{t^{(2)},\bsr^{(2)},\bss^{(2)},\bsT^{(2)}}$ can be
identified as
\begin{eqnarray*}
\begin{cases}
t^{(2)} &= t^2+\abs{\bsr}^2+\abs{\bss}^2+\norm{\bsT}^2,\\
\bsr^{(2)} &= 2t\cdot\bsr + 2\bsT\bss,\\
\bss^{(2)} &= 2t\cdot\bss + 2\bsT^\t\bsr,\\
\bsT^{(2)} &= 2t\cdot\bsT+2\out{\bsr}{\bss}-2\widehat\bsT.
\end{cases}
\end{eqnarray*}
\item For $k=3$,
$\bsX^3\approx\Pa{t^{(3)},\bsr^{(3)},\bss^{(3)},\bsT^{(3)}}$ can be
identified as
\begin{eqnarray*}
\begin{cases}
t^{(3)} &= t^3+3t(\abs{\bsr}^2+\abs{\bss}^2+\norm{\bsT}^2)+6(\Innerm{\bsr}{\bsT}{\bss}-\det(\bsT)),\\
\bsr^{(3)} &=
\Pa{3t^2+\abs{\bsr}^2+3\abs{\bss}^2+\norm{\bsT}^2}\bsr+2\bsT\bsT^\t\bsr+6t\bsT\bss-2\widehat\bsT\bss,\\
\bss^{(3)} &=
\Pa{3t^2+3\abs{\bsr}^2+\abs{\bss}^2+\norm{\bsT}^2}\bss+2\bsT^\t\bsT\bss+6t\bsT^\t\bsr-2\widehat\bsT^\t\bsr,\\
\bsT^{(3)}
&=\Pa{3t^2+\abs{\bsr}^2+\abs{\bss}^2+3\norm{\bsT}^2}\bsT+6t(\out{\bsr}{\bss}-\widehat\bsT)\\
&~~~+2\Pa{\proj{\bsr}\bsT+\bsT\proj{\bss}-\bsT\bsT^\t\bsT-\Omega(\bsT,\out{\bsr}{\bss})}.
\end{cases}
\end{eqnarray*}
\item For $k=4$,
$\bsX^4\approx\Pa{t^{(4)},\bsr^{(4)},\bss^{(4)},\bsT^{(4)}}$ can be
identified as
\begin{eqnarray*}
\begin{cases}
t^{(4)} &= t^4+\abs{\bsr}^4+\abs{\bss}^4+\norm{\bsT}^4+6\abs{\bsr}^2\abs{\bss}^2+6t^2(\abs{\bsr}^2+\abs{\bss}^2+\norm{\bsT}^2)\\
&~~~+2(\abs{\bsr}^2+\abs{\bss}^2)\norm{\bsT}^2+4(\Innerm{\bsr}{\bsT\bsT^\t}{\bsr}+\Innerm{\bss}{\bsT^\t\bsT}{\bss}+\Inner{\widehat\bsT}{\widehat\bsT})\\
&~~~+24t(\Innerm{\bsr}{\bsT}{\bss}-\det(\bsT))-8\Innerm{\bsr}{\widehat\bsT}{\bss},\\
\bsr^{(4)} &=4\Big[\Pa{t(t^2+\abs{\bsr}^2+3\abs{\bss}^2+\norm{\bsT}^2)+2\Innerm{\bsr}{\bsT}{\bss}-2\det(\bsT)}\bsr\\
&~~~+(3t^2+\abs{\bsr}^2+\abs{\bss}^2+\norm{\bsT}^2)\bsT\bss+2t\bsT\bsT^\t\bsr-2t\widehat\bsT\bss\Big],\\
\bss^{(4)} &=4\Big[\Pa{t(t^2+3\abs{\bsr}^2+\abs{\bss}^2+\norm{\bsT}^2)+2\Innerm{\bsr}{\bsT}{\bss}-2\det(\bsT)}\bss\\
&~~~+(3t^2+\abs{\bsr}^2+\abs{\bss}^2+\norm{\bsT}^2)\bsT^\t\bsr+2t\bsT^\t\bsT\bss-2t\widehat\bsT^\t\bsr\Big],\\
\bsT^{(4)}&=4\Big[\Pa{t(t^2+\abs{\bsr}^2+\abs{\bss}^2+3\norm{\bsT}^2)+2\Innerm{\bsr}{\bsT}{\bss}-2\det(\bsT)}\bsT\\
&~~~+(3t^2+\abs{\bsr}^2+\abs{\bss}^2+\norm{\bsT}^2)(\out{\bsr}{\bss}-\widehat\bsT)
\\
&~~~+2t\Pa{\proj{\bsr}\bsT+\bsT\proj{\bss}-\bsT\bsT^\t\bsT-\Omega(\bsT,\out{\bsr}{\bss})}\Big]
\end{cases}
\end{eqnarray*}
\end{enumerate}
For instance, we give the details in calculating $\bsT^{(4)}$:
\begin{eqnarray*}
\bsT^{(4)} &=& \out{\bsr^{(3)}}{\bss} + \out{\bsr}{\bss^{(3)}} +
t\bsT^{(3)} + t^{(3)}\bsT - \Omega(\bsT^{(3)},\bsT).
\end{eqnarray*}
In what follows, we calculate
it term by term:
\begin{enumerate}[(i)]
\item $\out{\bsr^{(3)}}{\bss} =
(3t^2+\abs{\bsr}^2+3\abs{\bss}^2+\norm{\bsT}^2)\out{\bsr}{\bss}+2\bsT\bsT^\t\out{\bsr}{\bss}+2(3t\bsT-\widehat\bsT)\proj{\bss}$
\item $\out{\bsr}{\bss^{(3)}} =(3t^2+3\abs{\bsr}^2+\abs{\bss}^2+\norm{\bsT}^2)\out{\bsr}{\bss}+2\out{\bsr}{\bss}\bsT^\t\bsT+2\proj{\bsr}(3t\bsT-\widehat\bsT)$
\item $t\bsT^{(3)} = 2t(\proj{\bsr}\bsT+\bsT\proj{\bss})+t(3t^2+\abs{\bsr}^2+\abs{\bss}^2+3\norm{\bsT}^2)\bsT + 6t^2(\out{\bsr}{\bss}-\widehat\bsT)-2t[(\bsr\cdot\cF)\bsT(\bss\cdot\cF)^\t+\bsT\bsT^\t\bsT]$
\item $t^{(3)}\bsT = \Br{t^3+3t(\abs{\bsr}^2+\abs{\bss}^2+\norm{\bsT}^2)+6(\Innerm{\bsr}{\bsT}{\bss}-\det(\bsT))}\bsT$
\item Now we calculate $\Omega(\bsT^{(3)},\bsT)$. Indeed,
\begin{eqnarray*}
\Omega(\bsT^{(3)},\bsT) &=&
2[\Omega(\proj{\bsr}\bsT,\bsT)+\Omega(\bsT\proj{\bss},\bsT)]+(3t^2+\abs{\bsr}^2+\abs{\bss}^2+3\norm{\bsT}^2)\Omega(\bsT,\bsT)\\
&& 6t[\Omega(\out{\bsr}{\bss},\bsT) -
\Omega(\widehat\bsT,\bsT)]-2[\Omega((\bsr\cdot\cF)\bsT(\bss\cdot\cF)^\t,\bsT)+\Omega(\bsT\bsT^\t\bsT,\bsT)]\\
&=&
2[(\bsr\cdot\cF)\bsT((\bsT^\t\bsr)\cdot\cF)^\t+((\bsT\bss)\cdot\cF)\bsT(\bss\cdot\cF)^\t]+2(3t^2+\abs{\bsr}^2+\abs{\bss}^2+3\norm{\bsT}^2)\widehat\bsT\\
&&+6t[(\bsr\cdot\cF)\bsT(\bss\cdot\cF)^\t -
\Inner{\bsT}{\bsT}\bsT+\bsT\bsT^\t\bsT]\\
&&-2[\Innerm{\bsr}{\bsT}{\bss}\bsT - \out{\bsr}{\bss}\bsT^\t\bsT +
(\norm{\bsT}^2\I_3-\bsT\bsT^\t)\out{\bsr}{\bss}]-2[\Inner{\bsT}{\bsT}\widehat\bsT-
\det(\bsT)\bsT].
\end{eqnarray*}
\end{enumerate}
Thus
\begin{eqnarray*}
\bsT^{(4)}
&=&4\Big[\Pa{t^3+t(\abs{\bsr}^2+\abs{\bss}^2+3\norm{\bsT}^2)+2\Innerm{\bsr}{\bsT}{\bss}-2\det(\bsT)}\bsT
\\
&&
+(3t^2+\abs{\bsr}^2+\abs{\bss}^2+\norm{\bsT}^2)(\out{\bsr}{\bss}-\widehat\bsT)
\\
&&+2t\Pa{\proj{\bsr}\bsT+\bsT\proj{\bss}-\bsT\bsT^\t\bsT-(\bsr\cdot\cF)\bsT(\bss\cdot\cF)^\t}\Big]\\
&=&4\Big[\Pa{t^3+t(\abs{\bsr}^2+\abs{\bss}^2+3\norm{\bsT}^2)+2\Innerm{\bsr}{\bsT}{\bss}-2\det(\bsT)}\bsT
\\
&&
+(3t^2+\abs{\bsr}^2+\abs{\bss}^2+\norm{\bsT}^2)(\out{\bsr}{\bss}-\widehat\bsT)
\\
&&+2t\Pa{\proj{\bsr}\bsT+\bsT\proj{\bss}-\bsT\bsT^\t\bsT-\Omega(\bsT,\out{\bsr}{\bss})}\Big].
\end{eqnarray*}

\subsection{Some results about products involved two-qubit states}\label{sub:2}

We have already known that
\begin{eqnarray}
\bsX^k&\approx& \Pa{t^{(k)},\bsr^{(k)},\bss^{(k)},\bsT^{(k)}},\\
\rho_A\ot\I_B&\approx& \Pa{\frac12,\frac{\bsa}2,\zero,\zero},\\
\I_A\ot\rho_B&\approx& \Pa{\frac12,\zero,\frac{\bsb}2,\zero},\\
\rho_A\ot\rho_B&\approx&
\Pa{\frac14,\frac{\bsa}4,\frac{\bsb}4,\frac{\out{\bsa}{\bsb}}4}
\end{eqnarray}

\begin{prop}
Let $\rho^k_{AB}\approx
\frac1{4^k}\Pa{c^{(k)},\bsx^{(k)},\bsy^{(k)},\bsZ^{(k)}}$, where
$k=2,3,4$. We have the following results:
\begin{enumerate}[(i)]
\item For $k=2$,
\begin{eqnarray}
\begin{cases}
c^{(2)} &= 1+\abs{\bsa}^2+\abs{\bsb}^2+\norm{\bsC}^2,\\
\bsx^{(2)}&=2\bsa+2\bsC\bsb,\\
\bsy^{(2)}&=2\bsb+2\bsC^\t\bsa,\\
\bsZ^{(2)}&=2\Pa{\bsC+\out{\bsa}{\bsb}-\widehat\bsC}.
\end{cases}
\end{eqnarray}
\item For $k=3$,
\begin{eqnarray}
\begin{cases}
c^{(3)} &= 1+3\Pa{\abs{\bsa}^2+\abs{\bsb}^2+\norm{\bsC}^2}+6\Pa{\Innerm{\bsa}{\bsC}{\bsb}-\det(\bsC)},\\
\bsx^{(3)}&=\Pa{3+\abs{\bsa}^2+3\abs{\bsb}^2+\norm{\bsC}^2}\bsa+2\bsC\bsC^\t\bsa+6\bsC\bsb-2\widehat\bsC\bsb,\\
\bsy^{(3)}&=\Pa{3+3\abs{\bsa}^2+\abs{\bsb}^2+\norm{\bsC}^2}\bsb+2\bsC^\t\bsC\bsb+6\bsC^\t\bsa-2{\widehat\bsC}^\t\bsa,\\
\bsZ^{(3)}&=\Pa{3+\abs{\bsa}^2+\abs{\bsb}^2+3\norm{\bsC}^2}\bsC+6\Pa{\out{\bsa}{\bsb}-\widehat\bsC}\\
&~~~+2\Pa{\proj{\bsa}\bsC+\bsC\proj{\bsb}-\bsC\bsC^\t\bsC-\Omega(\bsC,\out{\bsa}{\bsb})}.
\end{cases}
\end{eqnarray}
\item For $k=4$,
\begin{eqnarray}
\begin{cases}
c^{(4)} &=
1+6\Pa{\abs{\bsa}^2+\abs{\bsb}^2+\abs{\bsa}^2\abs{\bsb}^2}+\abs{\bsa}^4+\abs{\bsb}^4+\norm{\bsC}^4+24\Innerm{\bsa}{\bsC}{\bsb}\\
&~~~+2\norm{\bsC}^2\Pa{3+\abs{\bsa}^2+\abs{\bsb}^2}
+4\Innerm{\bsa}{\bsC\bsC^\t}{\bsa}+4\Innerm{\bsb}{\bsC^\t\bsC}{\bsb}+4\abs{\widehat\bsC}^2\\
&~~~-8\Innerm{\bsa}{\widehat\bsC}{\bsb}-24\det(\bsC),\\
\bsx^{(4)}&=4\Pa{1+\abs{\bsa}^2+3\abs{\bsb}^2+\norm{\bsC}^2+2\Innerm{\bsa}{\bsC}{\bsb}-2\det(\bsC)}\bsa+8\bsC\bsC^\t\bsa\\
&~~~+4\Pa{3+\abs{\bsa}^2+\abs{\bsb}^2+\norm{\bsC}^2}\bsC\bsb-8\widehat\bsC\bsb,\\
\bsy^{(4)}&=4\Pa{1+3\abs{\bsa}^2+\abs{\bsb}^2+\norm{\bsC}^2+2\Innerm{\bsa}{\bsC}{\bsb}-2\det(\bsC)}\bsb+8\bsC^\t\bsC\bsb\\
&~~~+4\Pa{3+\abs{\bsa}^2+\abs{\bsb}^2+\norm{\bsC}^2}\bsC^\t\bsa-8\widehat\bsC^\t\bsa,\\
\bsZ^{(4)}&=4\Pa{1+\abs{\bsa}^2+\abs{\bsb}^2+3\norm{\bsC}^2+2\Innerm{\bsa}{\bsC}{\bsb}
-2\det(\bsC)}\bsC\\
&~~~+4\Pa{3+\abs{\bsa}^2+\abs{\bsb}^2+\norm{\bsC}^2}(\out{\bsa}{\bsb}-\widehat\bsC)\\
&~~~+8\Pa{\proj{\bsa}\bsC+\bsC\proj{\bsb}-\bsC\bsC^\t\bsC-\Omega(\bsC,\out{\bsa}{\bsb})}.
\end{cases}
\end{eqnarray}
\end{enumerate}
\end{prop}

\begin{proof}
The proof follows immediately when we let
$$
(t,\bsr,\bss,\bsT)=\frac14\Pa{1,\bsa,\bsb,\bsC}
$$
in Corollary~\ref{cor:nthpower}.
\end{proof}

\begin{prop}
Let
\begin{eqnarray*}
\bsX^k(\rho_A\ot\I_B)\approx \frac12(\tilde c^{(k)}_A,\tilde
\bsx^{(k)}_A,\tilde\bsy^{(k)}_A,\tilde \bsZ^{(k)}_A),\\
\bsX^k(\I_A\ot\rho_B)\approx \frac12(\tilde c^{(k)}_B,\tilde
\bsx^{(k)}_B,\tilde\bsy^{(k)}_B,\tilde \bsZ^{(k)}_B),\\
\bsX^k(\rho_A\ot\rho_B)\approx \frac14(\tilde c^{(k)}_{AB},\tilde
\bsx^{(k)}_{AB},\tilde\bsy^{(k)}_{AB},\tilde \bsZ^{(k)}_{AB}).
\end{eqnarray*}
We have the following results:
\begin{enumerate}[(i)]
\item $\bsX^k(\rho_A\ot\I_B)\approx \frac12(\tilde c^{(k)}_A,\tilde
\bsx^{(k)}_A,\tilde\bsy^{(k)}_A,\tilde \bsZ^{(k)}_A)$ is determined
by
\begin{eqnarray}
\begin{cases}
\tilde c^{(k)}_A &= t^{(k)}+\Inner{\bsr^{(k)}}{\bsa},\\
\tilde \bsx^{(k)}_A&= \bsr^{(k)}+t^{(k)}\bsa+\mathrm{i}\bsr^{(k)}\times\bsa,\\
\tilde\bsy^{(k)}_A&= \bss^{(k)}+{\bsT^{(k)}}^\t\bsa,\\
\tilde
\bsZ^{(k)}_A&=\out{\bsa}{\bss^{(k)}}+\bsT^{(k)}-\mathrm{i}(\bsa\cdot\cF)^\t\bsT^{(k)},
\end{cases}
\end{eqnarray}
\item $\bsX^k(\I_A\ot\rho_B)\approx \frac12(\tilde c^{(k)}_B,\tilde
\bsx^{(k)}_B,\tilde\bsy^{(k)}_B,\tilde \bsZ^{(k)}_B)$ is determined
by
\begin{eqnarray}
\begin{cases}
\tilde c^{(k)}_B &= t^{(k)}+\Inner{\bss^{(k)}}{\bsb},\\
\tilde \bsx^{(k)}_B&= \bsr^{(k)}+\bsT^{(k)}\bsb,\\
\tilde\bsy^{(k)}_B&= \bss^{(k)}+t^{(k)}\bsb+\mathrm{i}\bss^{(k)}\times\bsb,\\
\tilde
\bsZ^{(k)}_B&=\out{\bsr^{(k)}}{\bsb}+\bsT^{(k)}-\mathrm{i}\bsT^{(k)}(\bsb\cdot\cF),
\end{cases}
\end{eqnarray}
\item $\bsX^k(\rho_A\ot\rho_B)\approx \frac14(\tilde c^{(k)}_{AB},\tilde
\bsx^{(k)}_{AB},\tilde\bsy^{(k)}_{AB},\tilde \bsZ^{(k)}_{AB})$ is
determined by
\begin{eqnarray}
\begin{cases}
\tilde c^{(k)}_{AB} &= t^{(k)}+\Inner{\bsr^{(k)}}{\bsa}+\Inner{\bss^{(k)}}{\bsb}+\Innerm{\bsa}{\bsT^{(k)}}{\bsb},\\
\tilde \bsx^{(k)}_{AB}&= \bsr^{(k)}+(t^{(k)}+\Inner{\bss^{(k)}}{\bsb})\bsa+\bsT^{(k)}\bsb + \mathrm{i}(\bsr^{(k)}\times\bsa+\bsT^{(k)}\bsb\times\bsa),\\
\tilde\bsy^{(k)}_{AB}&= \bss^{(k)}+(t^{(k)}+\Inner{\bsr^{(k)}}{\bsa})\bsb + {\bsT^{(k)}}^\t\bsa + \mathrm{i}(\bss^{(k)}\times\bsb+{\bsT^{(k)}}^\t\bsa\times\bsb),\\
\tilde
\bsZ^{(k)}_{AB}&=\out{\bsr^{(k)}}{\bsb}+\out{\bsa}{\bss^{(k)}}+t^{(k)}\out{\bsa}{\bsb}+\bsT^{(k)}-
\Omega(\bsT^{(k)},\out{\bsa}{\bsb})\\
&~~~+\mathrm{i}\Pa{\Psi(\bsr^{(k)},\out{\bsa}{\bsb},\bss^{(k)})-\Psi(\bsa,\bsT^{(k)},\bsb)}.
\end{cases}
\end{eqnarray}
\end{enumerate}
\end{prop}

\begin{prop}
Let
\begin{eqnarray*}
\rho^k_{AB}(\rho_A\ot\I_B)&\approx&\frac1{2\cdot 4^k}(c^{(k)}_A,\bsx^{(k)}_A,\bsy^{(k)}_A,\bsZ^{(k)}_A),\\
\rho^k_{AB}(\I_A\ot\rho_B)&\approx&\frac1{2\cdot 4^k}(c^{(k)}_B,\bsx^{(k)}_B,\bsy^{(k)}_B,\bsZ^{(k)}_B),\\
\rho^k_{AB}(\rho_A\ot\rho_B)&\approx&\frac1{4^{k+1}}(c^{(k)}_{AB},\bsx^{(k)}_{AB},\bsy^{(k)}_{AB},\bsZ^{(k)}_{AB}).
\end{eqnarray*}
Then we get the following statements:
\begin{enumerate}[(i)]
\item For $k=1$, it holds that
\begin{eqnarray*}
\begin{cases}
c^{(1)}_A &= 1+\abs{\bsa}^2,\\
\bsx^{(1)}_A &= 2\bsa,\\
\bsy^{(1)}_A &= \bsb+\bsC^\t\bsa,\\
\bsZ^{(1)}_A &=
\bsC+\out{\bsa}{\bsb}-\mathrm{i}(\bsa\cdot\cF)^\t\bsC,
\end{cases}\qquad
\begin{cases}
c^{(1)}_B &= 1+\abs{\bsb}^2,\\
\bsx^{(1)}_B &= \bsa+\bsC\bsb,\\
\bsy^{(1)}_B &= 2\bsb,\\
\bsZ^{(1)}_B &= \bsC+\out{\bsa}{\bsb}-\mathrm{i}\bsC(\bsb\cdot\cF).
\end{cases}
\end{eqnarray*}
and
\begin{eqnarray*}
\begin{cases}
c^{(1)}_{AB} &= 1+\abs{\bsa}^2+\abs{\bsb}^2+\Innerm{\bsa}{\bsC}{\bsb}\\
\bsx^{(1)}_{AB} &= (2+\abs{\bsb}^2)\bsa+\bsC\bsb+\mathrm{i}\bsC\bsb\times\bsa\\
\bsy^{(1)}_{AB} &= (2+\abs{\bsa}^2)\bsb+\bsC^\t\bsa+\mathrm{i}\bsC^\t\bsa\times\bsb\\
\bsZ^{(1)}_{AB} &=
\bsC+3\out{\bsa}{\bsb}-\Omega(\bsC,\out{\bsa}{\bsb})-\mathrm{i}\Psi(\bsa,\bsC,\bsb).
\end{cases}
\end{eqnarray*}
\item For $k=2$, it holds that
\begin{eqnarray*}
\begin{cases}
c^{(2)}_A &= 1+3\abs{\bsa}^2+\abs{\bsb}^2+\norm{\bsC}^2+2\Innerm{\bsa}{\bsC}{\bsb},\\
\bsx^{(2)}_A &= (3+\abs{\bsa}^2+\abs{\bsb}^2+\norm{\bsC}^2)\bsa+2\bsC\bsb+2\mathrm{i}\bsC\bsb\times\bsa,\\
\bsy^{(2)}_A &= 2(1+\abs{\bsa}^2)\bsb+4\bsC^\t\bsa-2\widehat\bsC^\t\bsa,\\
\bsZ^{(2)}_A &=2(
\bsC-\widehat\bsC)+2\proj{\bsa}\bsC+4\out{\bsa}{\bsb}-2\mathrm{i}(\bsa\cdot\cF)^\t(\bsC-\widehat\bsC),
\end{cases}
\end{eqnarray*}
\begin{eqnarray*}
\begin{cases}
c^{(2)}_B &= 1+\abs{\bsa}^2+3\abs{\bsb}^2+\norm{\bsC}^2+2\Innerm{\bsa}{\bsC}{\bsb},\\
\bsx^{(2)}_B &= 2(1+\abs{\bsb}^2)\bsa+4\bsC\bsb-2\widehat\bsC\bsb, \\
\bsy^{(2)}_B &= (3+\abs{\bsa}^2+\abs{\bsb}^2+\norm{\bsC}^2)\bsb+2\bsC^\t\bsa+2\mathrm{i}\bsC^\t\bsa\times\bsb,\\
\bsZ^{(2)}_B &=2(
\bsC-\widehat\bsC)+2\bsC\proj{\bsb}+4\out{\bsa}{\bsb}-2\mathrm{i}(\bsC-\widehat\bsC)(\bsb\cdot\cF).
\end{cases}
\end{eqnarray*}
and
\begin{eqnarray*}
\begin{cases}
c^{(2)}_{AB} &= 1+3\abs{\bsa}^2+3\abs{\bsb}^2+2\abs{\bsa}^2\abs{\bsb}^2+\norm{\bsC}^2+6\Innerm{\bsa}{\bsC}{\bsb}-2\Innerm{\bsa}{\widehat\bsC}{\bsb},\\
\bsx^{(2)}_{AB} &= (3+\abs{\bsa}^2+5\abs{\bsb}^2+\norm{\bsC}^2+2\Innerm{\bsa}{\bsC}{\bsb})\bsa+2(2\bsC\bsb-\widehat\bsC\bsb)+2\mathrm{i}(2\bsC\bsb-\widehat\bsC\bsb)\times\bsa,\\
\bsy^{(2)}_{AB} &= (3+5\abs{\bsa}^2+\abs{\bsb}^2+\norm{\bsC}^2+2\Innerm{\bsa}{\bsC}{\bsb})\bsb+2(2\bsC^\t\bsa-\widehat\bsC^\t\bsa)+2\mathrm{i}(2\bsC^\t\bsa-\widehat\bsC^\t\bsa)\times\bsb,\\
\bsZ^{(2)}_{AB} &=
2(\bsC-\widehat\bsC)+2(\proj{\bsa}\bsC+\bsC\proj{\bsb})+(7+\abs{\bsa}^2+\abs{\bsb}^2+\norm{\bsC}^2)\out{\bsa}{\bsb}\\
&~~~-2\Omega(\bsC-\widehat\bsC,\out{\bsa}{\bsb})+2\mathrm{i}(\Psi(\bsC\bsb,\out{\bsa}{\bsb},\bsC^\t\bsa)
- \Psi(\bsa,\bsC-\widehat\bsC,\bsb)).
\end{cases}
\end{eqnarray*}
\item For $k=3$, it holds that
\begin{eqnarray*}
\begin{cases}
c^{(3)}_A &= 1+6\abs{\bsa}^2+\abs{\bsa}^4+3\abs{\bsb}^2(1+\abs{\bsa}^2)+(3+\abs{\bsa}^2)\norm{\bsC}^2+12\Innerm{\bsa}{\bsC}{\bsb}+2\Innerm{\bsa}{\bsC\bsC^\t}{\bsa}\\
&~~~-6\det(\bsC)-2\Innerm{\bsa}{\widehat\bsC}{\bsb},\\
\bsx^{(3)}_A &= (4+4\abs{\bsa}^2+6\abs{\bsb}^2+4\norm{\bsC}^2+6\Innerm{\bsa}{\bsC}{\bsb}-6\det(\bsC))\bsa+2(\bsC\bsC^\t\bsa+3\bsC\bsb-\widehat\bsC\bsb)\\
&~~~+2\mathrm{i}(\bsC\bsC^\t\bsa+3\bsC\bsb-\widehat\bsC\bsb)\times\bsa,\\
\bsy^{(3)}_A &= (3+9\abs{\bsa}^2+\abs{\bsb}^2+\norm{\bsC}^2+2\Innerm{\bsa}{\bsC}{\bsb})\bsb+2\bsC^\t\bsC\bsb+(9+3\abs{\bsa}^2+\abs{\bsb}^2+3\norm{\bsC}^2)\bsC^\t\bsa\\
&~~~-2\bsC^\t\bsC\bsC^\t\bsa-8\widehat\bsC^\t\bsa,\\
\bsZ^{(3)}_A
&=(3+\abs{\bsa}^2+\abs{\bsb}^2+3\norm{\bsC}^2)\bsC-6\widehat\bsC+2\proj{\bsa}(4\bsC-\widehat\bsC)
+2\bsC\proj{\bsb}+2\out{\bsa}{\bsb}\bsC^\t\bsC\\
&~~~+(9+3\abs{\bsa}^2+\abs{\bsb}^2+\norm{\bsC}^2)\out{\bsa}{\bsb}-2\bsC\bsC^\t\bsC-2\Omega(\bsC,\out{\bsa}{\bsb})\\
&~~~-\mathrm{i}(\bsa\cdot\cF)^\t\Br{(3+\abs{\bsa}^2+\abs{\bsb}^2+3\norm{\bsC}^2)\bsC-6\widehat\bsC+2(\bsC\proj{\bsb}-\bsC\bsC^\t\bsC-\Omega(\bsC,\out{\bsa}{\bsb}))},
\end{cases}
\end{eqnarray*}
\begin{eqnarray*}
\begin{cases}
c^{(3)}_B &= 1+6\abs{\bsb}^2+\abs{\bsb}^4+3\abs{\bsa}^2(1+\abs{\bsb}^2)+(3+\abs{\bsb}^2)\norm{\bsC}^2+12\Innerm{\bsa}{\bsC}{\bsb}+2\Innerm{\bsb}{\bsC^\t\bsC}{\bsb}\\
&~~~-6\det(\bsC)-2\Innerm{\bsa}{\widehat\bsC}{\bsb},\\
\bsx^{(3)}_B &= (3+\abs{\bsa}^2+9\abs{\bsb}^2+\norm{\bsC}^2+2\Innerm{\bsa}{\bsC}{\bsb})\bsa+2\bsC\bsC^\t\bsa+(9+\abs{\bsa}^2+3\abs{\bsb}^2+3\norm{\bsC}^2)\bsC\bsb\\
&~~~-2\bsC\bsC^\t\bsC\bsb-8\widehat\bsC\bsb,\\
\bsy^{(3)}_B &= (4+6\abs{\bsa}^2+4\abs{\bsb}^2+4\norm{\bsC}^2+6\Innerm{\bsa}{\bsC}{\bsb}-6\det(\bsC))\bsb+2(\bsC^\t\bsC\bsb+3\bsC^\t\bsa-\widehat\bsC^\t\bsa)\\
&~~~+2\mathrm{i}(\bsC^\t\bsC\bsb+3\bsC^\t\bsa-\widehat\bsC^\t\bsa)\times\bsb,\\
\bsZ^{(3)}_B
&=(3+\abs{\bsa}^2+\abs{\bsb}^2+3\norm{\bsC}^2)\bsC-6\widehat\bsC+2(4\bsC-\widehat\bsC)\proj{\bsb}
+2\proj{\bsa}\bsC+2\bsC\bsC^\t\out{\bsa}{\bsb}\\
&~~~+(9+\abs{\bsa}^2+3\abs{\bsb}^2+\norm{\bsC}^2)\out{\bsa}{\bsb}-2\bsC\bsC^\t\bsC-2\Omega(\bsC,\out{\bsa}{\bsb})\\
&~~~-\mathrm{i}\Br{(3+\abs{\bsa}^2+\abs{\bsb}^2+3\norm{\bsC}^2)\bsC-6\widehat\bsC+2(\proj{\bsa}\bsC-\bsC\bsC^\t\bsC-\Omega(\bsC,\out{\bsa}{\bsb}))}(\bsb\cdot\cF),
\end{cases}
\end{eqnarray*}
and
\begin{eqnarray*}
\begin{cases}
c^{(3)}_{AB} &= 1+6(\abs{\bsa}^2+\abs{\bsb}^2)+12\abs{\bsa}^2\abs{\bsb}^2+\abs{\bsa}^4+\abs{\bsb}^4+(3+\abs{\bsa}^2+\abs{\bsb}^2)\norm{\bsC}^2\\
&~~~+3(7+\abs{\bsa}^2+\abs{\bsb}^2+\norm{\bsC}^2)\Innerm{\bsa}{\bsC}{\bsb}-10\Innerm{\bsa}{\widehat\bsC}{\bsb}-6\det(\bsC)\\
&~~~+2(\Innerm{\bsa}{\bsC\bsC^\t}{\bsa}+\Innerm{\bsb}{\bsC^\t\bsC}{\bsb}-\Innerm{\bsa}{\bsC\bsC^\t\bsC}{\bsb}),\\
\bsx^{(3)}_{AB} &= \Big(4+4\abs{\bsa}^2+15\abs{\bsb}^2+\abs{\bsb}^4+4\norm{\bsC}^2+(3\abs{\bsa}^2+\norm{\bsC}^2)\abs{\bsb}^2+14\Innerm{\bsa}{\bsC}{\bsb}+2\Innerm{\bsb}{\bsC^\t\bsC}{\bsb}\\
&~~~-2\Innerm{\bsa}{\widehat\bsC}{\bsb}-6\det(\bsC)\Big)\bsa+2\bsC\bsC^\t\bsa + (9+\abs{\bsa}^2+3\abs{\bsb}^2+3\norm{\bsC}^2)\bsC\bsb-8\widehat\bsC\bsb-2\bsC\bsC^\t\bsC\bsb\\
&~~~+\mathrm{i}\Pa{2\bsC\bsC^\t\bsa+(9+\abs{\bsa}^2+3\abs{\bsb}^2+3\norm{\bsC}^2)\bsC\bsb-8\widehat\bsC\bsb-2\bsC\bsC^\t\bsC\bsb}\times\bsa,\\
\bsy^{(3)}_{AB} &= \Big(4+15\abs{\bsa}^2+\abs{\bsa}^4+4\abs{\bsb}^2+4\norm{\bsC}^2+(3\abs{\bsb}^2+\norm{\bsC}^2)\abs{\bsa}^2+14\Innerm{\bsa}{\bsC}{\bsb}+2\Innerm{\bsa}{\bsC\bsC^\t}{\bsa}\\
&~~~-2\Innerm{\bsa}{\widehat\bsC}{\bsb}-6\det(\bsC)\Big)\bsb+2\bsC^\t\bsC\bsb + (9+3\abs{\bsa}^2+\abs{\bsb}^2+3\norm{\bsC}^2)\bsC^\t\bsa-8\widehat\bsC^\t\bsa-2\bsC^\t\bsC\bsC^\t\bsa\\
&~~~+\mathrm{i}\Pa{2\bsC^\t\bsC\bsb+(9+3\abs{\bsa}^2+\abs{\bsb}^2+3\norm{\bsC}^2)\bsC^\t\bsa-8\widehat\bsC^\t\bsa-2\bsC^\t\bsC\bsC^\t\bsa}\times\bsb,\\
\bsZ^{(3)}_{AB}
&=(3+\abs{\bsa}^2+\abs{\bsb}^2+2\abs{\bsa}^2\abs{\bsb}^2+3\norm{\bsC}^2)\bsC+2(4-\abs{\bsb}^2)\proj{\bsa}\bsC+2(4-\abs{\bsa}^2)\bsC\proj{\bsb}\\
&~~~+\Pa{13+7(\abs{\bsa}^2+\abs{\bsb}^2)+5\norm{\bsC}^2+8\Innerm{\bsa}{\bsC}{\bsb}-6\det(\bsC)}\out{\bsa}{\bsb}-6\widehat\bsC\\
&~~~+2\Pa{\bsC\bsC^\t\out{\bsa}{\bsb}+\out{\bsa}{\bsb}\bsC^\t\bsC -
\bsC\bsC^\t\bsC-\proj{\bsa}\widehat\bsC-\widehat\bsC\proj{\bsb}}\\
&~~~-\Omega((5+\abs{\bsa}^2+\abs{\bsb}^2+3\norm{\bsC}^2)\bsC-2\bsC\bsC^\t\bsC-6\widehat\bsC,\out{\bsa}{\bsb})\\
&~~~+\mathrm{i}\Psi(2\bsC\bsC^\t\bsa+6\bsC\bsb-2\widehat\bsC\bsb,\out{\bsa}{\bsb},2\bsC^\t\bsC\bsb+6\bsC^\t\bsa-2\widehat\bsC^\t\bsa)\\
&~~~-\mathrm{i}\Pa{(3+\abs{\bsa}^2+\abs{\bsb}^2+3\norm{\bsC}^2)\Psi(\bsa,\bsC,\bsb)-6\Psi(\bsa,\widehat\bsC,\bsb)-2\Psi(\bsa,\bsC\bsC^\t\bsC,\bsb)}\\
&~~~+2\mathrm{i}\Pa{\abs{\bsb}^2(\bsa\cdot\cF)\bsC+\abs{\bsa}^2\bsC(\bsb\cdot\cF)^\t}.
\end{cases}
\end{eqnarray*}
\end{enumerate}
\end{prop}

\subsection{Revisiting local unitary invariants}\label{sub:3}

For any two-qubit state $\rho_{AB}$, decomposed as
\begin{eqnarray}\label{eq:2qubit}
\rho_{AB} =
\frac14\Pa{\I\ot\I+\bsa\cdot\boldsymbol{\sigma}\ot\I+\I\ot
\bsb\cdot\boldsymbol{\sigma}+\sum^3_{i,j=1}c_{ij}\sigma_i\ot\sigma_j},
\end{eqnarray}
where $\bsa=(a_1,a_2,a_3)^\t$ and $\bsb= (b_1,b_2,b_3)^\t$ are in
$\real^3$, and $\bsC=(c_{ij})_{3\times3}\in\real^{3\times 3}$. Its
two reduced states are given by, respectively $\rho_A =
\frac12(\I_2+\bsa\cdot\bsigma)$ and $\rho_B =
\frac12(\I_2+\bsb\cdot\bsigma)$. In 2002, Makhlin had published the
following well-known result\footnote{Here we reformulate those 18 LU
invariants for our convenience. They are also termed Makhlin's
invariants.}:
\begin{prop}[\cite{Makhlin2002}]
For any mixed two-qubit states
$\rho_{AB},\rho'_{AB}\in\density{\complex^2\ot\complex^2}$, both are
LU equivalent if and only if the following 18-tuple
$(I_1,\ldots,I_{18})$ are the same for both $\rho_{AB}$ and
$\rho'_{AB}$, where
\begin{eqnarray*}
I_1=\det(\bsC),I_2=\Inner{\bsC}{\bsC},I_3=\Inner{\bsC^\t\bsC}{\bsC^\t\bsC},\\
I_4=\Inner{\bsa}{\bsa},I_5=\Innerm{\bsa}{\bsC\bsC^\t}{\bsa},I_6=\Innerm{\bsa}{(\bsC\bsC^\t)^2}{\bsa},\\
I_7=\Inner{\bsb}{\bsb},I_8=\Innerm{\bsb}{\bsC^\t\bsC}{\bsb},I_9=\Innerm{\bsb}{(\bsC^\t\bsC)^2}{\bsb},\\
I_{10}=\bsa\cdot(\bsC\bsC^\t\bsa\times
(\bsC\bsC^\t)^2\bsa),I_{11}=\bsb\cdot(\bsC^\t\bsC\bsb\times
(\bsC^\t\bsC)^2\bsb),\\
I_{12}=\Innerm{\bsa}{\bsC}{\bsb},I_{13}=\Innerm{\bsa}{\bsC\bsC^\t\bsC}{\bsb},I_{14}=\Inner{(\bsa\cdot\cF)\bsC}{\bsC(\bsb\cdot\cF)},\\
I_{15}=\bsa\cdot(\bsC\bsC^\t\bsa\times\bsC\bsb),I_{16}=\bsC^\t\bsa\cdot(\bsb\times\bsC^\t\bsC\bsb),\\
I_{17}=\bsC^\t\bsa\cdot(\bsC^\t\bsC\bsC^\t
\bsa\times\bsb),I_{18}=\bsa\cdot(\bsC\bsb\times\bsC\bsC^\t\bsC\bsb).
\end{eqnarray*}
\end{prop}

We should remark here that, in invariant theory, there is the notion
of so-called separating invariants, which in general might generate
a proper subalgebra of the full algebra of invariant polynomials. In
other words, the subalgebra of separating invariant polynomials is
generally not the full algebra of invariant polynomials. From
\cite{Grassl1998}, we see that 21 invariant polynomials which were
shown to be non-redundant, i.e., none of them can be expressed as a
polynomial in the others. Moreover, such 21 polynomials are indeed
generating the full algebra of invariant polynomials. Although 18
Makhlin invariants \cite{Makhlin2002} are sufficient to discriminate
the orbits with respect to LU transformation, they are just
separating invariants which generates a proper subalgebra of the
full algebra of invariant polynomials.

Here we deliberately omit the constant factor in Makhlin's
invariants. For our purposes, we will give another 18-tuple of
invariants in replacement of Makhlin's invariants.
\begin{prop}[\cite{Makhlin2002}]\label{prop:LI}
For any mixed two-qubit states
$\rho_{AB},\rho'_{AB}\in\density{\complex^2\ot\complex^2}$, both are
LU equivalent if and only if the following 18-tuple
$(L_1,\ldots,L_{18})$ are the same for both $\rho_{AB}$ and
$\rho'_{AB}$, where
\begin{eqnarray*}
L_1=\det(\bsC),L_2=\Inner{\bsC}{\bsC},L_3=\Inner{\widehat\bsC}{\widehat\bsC},\\
L_4=\Inner{\bsa}{\bsa},L_5=\Innerm{\bsa}{\bsC\bsC^\t}{\bsa},L_6=\Innerm{\bsa}{\widehat{\bsC\bsC^\t}}{\bsa},\\
L_7=\Inner{\bsb}{\bsb},L_8=\Innerm{\bsb}{\bsC^\t\bsC}{\bsb},L_9=\Innerm{\bsb}{\widehat{\bsC^\t\bsC}}{\bsb},\\
L_{10}=\bsa\cdot(\bsC\bsC^\t\bsa\times
\widehat{\bsC\bsC^\t}\bsa),L_{11}=\bsb\cdot(\bsC^\t\bsC\bsb\times
\widehat{\bsC^\t\bsC}\bsb),\\
L_{12}=\Innerm{\bsa}{\bsC}{\bsb},L_{13}=\Innerm{\bsa}{\bsC\bsC^\t\bsC}{\bsb},L_{14}=\Innerm{\bsa}{\widehat\bsC}{\bsb},\\
L_{15}=\bsb\cdot(\bsC^\t\bsa\times \widehat\bsC^\t\bsa),L_{16}=\bsa\cdot(\bsC\bsb\times\widehat\bsC\bsb),\\
L_{17}=\widehat\bsC\bsb\cdot(\bsa\times\bsC\bsC^\t
\bsa),L_{18}=\widehat\bsC^\t\bsa\cdot(\bsb\times\bsC^\t\bsC\bsb).
\end{eqnarray*}
\end{prop}

\begin{proof}
Note that we can find out the following relations
\begin{enumerate}[(1)]
\item $I_k=L_k$, where $k\in\set{1,2,4,5,7,8,10,11,12,13,14,17,18}$
\item $I_3=L^2_2-2L_3$
\item $I_6=L_6 + L_2L_5 - L_3L_4$
\item $I_9=L_9+L_2L_8 - L_3L_7$
\item $I_k=-L_k$, where $k\in\set{15,16}$
\end{enumerate}
Indeed, the first one is trivial. For the 2nd item, note that
$2\Inner{\widehat\bsC}{\widehat\bsC}=\Inner{\bsC}{\bsC}^2-\Inner{\bsC^\t\bsC}{\bsC^\t\bsC}$.
This implies that the desired result. For the third item,
$$
\widehat{\bsC\bsC^\t} =
(\bsC\bsC^\t)^2-\Inner{\bsC}{\bsC}\bsC\bsC^\t+\Inner{\widehat\bsC}{\widehat\bsC}\I_3
$$
implying that
\begin{eqnarray*}
\Innerm{\bsa}{\widehat{\bsC\bsC^\t}}{\bsa} =
\Innerm{\bsa}{(\bsC\bsC^\t)^2}{\bsa} -
\Inner{\bsC}{\bsC}\Innerm{\bsa}{\bsC\bsC^\t}{\bsa} +
\Inner{\widehat\bsC}{\widehat\bsC}\iinner{\bsa}{\bsa}.
\end{eqnarray*}
That is,
\begin{eqnarray}
L_6 = I_6 - L_2L_5 + L_3L_4.
\end{eqnarray}
For the 4th item,
$$
\widehat{\bsC^\t\bsC} =
(\bsC^\t\bsC)^2-\Inner{\bsC}{\bsC}\bsC^\t\bsC+\Inner{\widehat\bsC}{\widehat\bsC}\I_3
$$
implying that
\begin{eqnarray}
\Innerm{\bsb}{\widehat{\bsC^\t\bsC}}{\bsb} =
\Innerm{\bsb}{(\bsC^\t\bsC)^2}{\bsb} -
\Inner{\bsC}{\bsC}\Innerm{\bsb}{\bsC^\t\bsC}{\bsb} +
\Inner{\widehat\bsC}{\widehat\bsC}\iinner{\bsb}{\bsb}.
\end{eqnarray}
That is,
\begin{eqnarray}
L_9 = I_9 - L_2L_8 + L_3L_7.
\end{eqnarray}
For the equality of $I_{10/11}=L_{10/11}$
$$
(\bsC\bsC^\t)^2\bsa =
\widehat{\bsC\bsC^\t}\bsa+\Inner{\bsC}{\bsC}\bsC\bsC^\t\bsa-\Inner{\widehat\bsC}{\widehat\bsC}\bsa.
$$
Then
$$
\bsC\bsC^\t\bsa\times (\bsC\bsC^\t)^2\bsa =
\bsC\bsC^\t\bsa\times\widehat{\bsC\bsC^\t}\bsa
-\Inner{\widehat\bsC}{\widehat\bsC}\bsC\bsC^\t\bsa\times\bsa,
$$
implying that
\begin{eqnarray*}
I_{10} =\bsa\cdot(\bsC\bsC^\t\bsa\times(\bsC\bsC^\t)^2\bsa) =
\bsa\cdot\Pa{\bsC\bsC^\t\bsa\times\widehat{\bsC\bsC^\t}\bsa}=L_{10}.
\end{eqnarray*}

$$
(\bsC^\t\bsC)^2\bsb =
\widehat{\bsC^\t\bsC}\bsb+\Inner{\bsC}{\bsC}\bsC^\t\bsC\bsb-\Inner{\widehat\bsC}{\widehat\bsC}\bsb.
$$
Then
$$
\bsC^\t\bsC\bsb\times (\bsC^\t\bsC)^2\bsb =
\bsC^\t\bsC\bsb\times\widehat{\bsC^\t\bsC}\bsb
-\Inner{\widehat\bsC}{\widehat\bsC}\bsC^\t\bsC\bsb\times\bsb,
$$
implying that
\begin{eqnarray*}
I_{11} =\bsb\cdot(\bsC^\t\bsC\bsb\times(\bsC^\t\bsC)^2\bsb) =
\bsb\cdot\Pa{\bsC^\t\bsC\bsb\times\widehat{\bsC^\t\bsC}\bsb}=L_{11}.
\end{eqnarray*}
For the 5th item,
\begin{eqnarray*}
I_{15} &=& \bsa\cdot(\bsC\bsC^\t\bsa\times \bsC\bsb) = \Inner{\bsa}{\bsC\bsC^\t\bsa\times \bsC\bsb} = \Inner{\bsa}{\widehat\bsC(\bsC^\t\bsa\times \bsb)} \\
&=& \Inner{\widehat\bsC^\t\bsa}{\bsC^\t\bsa\times \bsb} =
\bsb\cdot(\widehat\bsC^\t\bsa\times\bsC^\t\bsa) =
-\bsb\cdot(\bsC^\t\bsa\times\widehat\bsC^\t\bsa) = -L_{15}.
\end{eqnarray*}
Similarly, we get that $I_{16}=-L_{16}$. Indeed,
\begin{eqnarray*}
I_{16} &=& \bsC^\t\bsa\cdot(\bsb\times \bsC^\t\bsC\bsb) = \bsb\cdot(\bsC^\t\bsC\bsb\times\bsC^\t\bsa) = \Inner{\bsb}{\bsC^\t\bsC\bsb\times\bsC^\t\bsa}\\
&=& \Inner{\bsb}{\widehat\bsC^\t(\bsC\bsb\times\bsa)} = \Inner{\widehat\bsC\bsb}{(\bsC\bsb\times\bsa)} = \bsa\cdot(\widehat\bsC\bsb\times\bsC\bsb)\\
&=&-\bsa\cdot(\bsC\bsb\times \widehat\bsC\bsb)=-L_{16}.
\end{eqnarray*}
We also note that
\begin{eqnarray*}
I_{17} &=& \bsb\cdot(\bsC^\t\bsa\times \bsC^\t\bsC\bsC^\t\bsa) =
\Inner{\bsb}{\widehat\bsC^\t(\bsa\times \bsC\bsC^\t\bsa)} =
\Inner{\widehat\bsC\bsb}{\bsa\times \bsC\bsC^\t\bsa}=L_{17}
\end{eqnarray*}
and
\begin{eqnarray*}
I_{18} &=& \bsa\cdot(\bsC\bsb\times \bsC\bsC^\t\bsC\bsb) =
\Inner{\bsa}{\widehat\bsC(\bsb\times \bsC^\t\bsC\bsb)} =
\Inner{\widehat\bsC^\t\bsa}{\bsb\times \bsC^\t\bsC\bsb}=L_{18}.
\end{eqnarray*}
From the above discussion, we can see that the invariant ring
generated by 18 Makhlin's invariants $I_k(k=1,\ldots,18)$ can also
be generated by our proposed 18 invariants $L_k(k=1,\ldots,18)$.
\end{proof}

Based on this observation, we can infer the following results:
\begin{lem}\label{lem:LB}
For any two-qubit state $\rho_{AB}$ decomposed as in
Eq.~\eqref{eq:2qubit} above, let
$\bsX_0=\rho_{AB},\bsX_1=\rho_A\ot\I_2$, and $\bsX_2=\I_2\ot\rho_B$,
it holds that
\begin{enumerate}[(1)]
\item $B_1=\Tr{\bsX_0\bsX_1}=\frac{1+L_4}{2}$.
\item $B_2=\Tr{\bsX_0\bsX_2}=\frac{1+L_7}{2}$.
\item $B_3=\Tr{\bsX_0\bsX_1\bsX_2} =
\frac{1+L_4+L_7+L_{12}}4$.
\item $B_4=\Tr{\bsX^2_0}=\frac{1+L_2+L_4+L_7}4$.
\item $B_5=\Tr{\bsX^2_0\bsX_1\bsX_2} = \frac{1+L_2+3L_4+3L_7+2L_4L_7+6L_{12} -
2L_{14}}{16}$.
\item
$B_6=\Tr{\bsX^3_0}=\frac{1-6L_1+3L_2+3L_4+3L_7+6L_{12}}{16}$.
\item $B_7=\Tr{\bsX^3_0\bsX_1} =
\frac{1-6L_1+L_2(3+L_4)+6L_4+L^2_4+2L_5+3(1+L_4)L_7+12L_{12}-2L_{14}}{32}$.
\item $B_8=\Tr{\bsX^3_0\bsX_2} =
\frac{1-6L_1+L_2(3+L_7)+6L_7+L^2_7+2L_8+3(1+L_7)L_4+12L_{12}-2L_{14}}{32}$.
\item $B_9=\Tr{\bsX^3_0\bsX_1\bsX_2}$ is given by
\begin{eqnarray*}
B_9 &=& \frac1{64}\Big[1+6(L_4+L_7)+12L_4L_7+L^2_4+L^2_7+(3+L_4+L_7)L_2 \\
&&~~~~~+3\Pa{7+L_2+L_4+L_7}L_{12}+2(L_5+L_8)-6L_1-2L_{13}-10L_{14}\Big].
\end{eqnarray*}
\item $B_{10}=\Tr{\bsX^4_0}$ is given by
\begin{eqnarray*}
B_{10}&=& \frac1{64}\Big[1+6\Pa{L_4+L_7+L_4L_7}
+L^2_4+L^2_7+(6+L_2+2L_4+2L_7)L_2+24L_{12}\\
&&~~~~~+4(L_3+L_5+L_8-2L_{14}-6L_1) \Big].
\end{eqnarray*}
\item $B_{11}=\Tr{\bsX^2_0\bsX_1\bsX^2_0\bsX_1}$ is
given by
\begin{eqnarray*}
B_{11}&=&\frac1{256}\Big[8L^2_{12}+8L_{12}(6+6L_4+L_7+L_2)+4(7+L_4)L_5-8(3+L_4)L_{14}+8L_6+4(1-L_4)L_8\\
&&~~~~~~-8(3+L_4)L_1+4(1-L_4)L_3+(1+L_4)L^2_2+2(1+L_4)(3+L_4+L_7)L_2\\
&&~~~~~~+\Pa{1+15L_4+15L^2_4+L^3_4+6L_7+36L_4L_7+6L^2_4L_7+L^2_7+L_4L^2_7}\Big].
\end{eqnarray*}
\item $B_{12}=\Tr{\bsX^2_0\bsX_2\bsX^2_0\bsX_2}$ is
given by
\begin{eqnarray*}
B_{12}&=&\frac1{256}\Big[8L^2_{12}+8L_{12}(6+6L_7+L_4+L_2)+4(7+L_7)L_8-8(3+L_7)L_{14}+8L_9+4(1-L_7)L_5\\
&&~~~~~~-8(3+L_7)L_1+4(1-L_7)L_3+(1+L_7)L^2_2+2(1+L_7)(3+L_4+L_7)L_2\\
&&~~~~~~+\Pa{1+15L_7+15L^2_7+L^3_7+6L_4+36L_4L_7+6L^2_7L_4+L^2_4+L_7L^2_4}\Big].
\end{eqnarray*}
\item $B_{13}=\Tr{\bsX_0\bsX_1\bsX_2\bsX^2_0\bsX_1}$ is
given by
\begin{eqnarray*}
B_{13}&=&\frac1{128}\Big[4L^2_{12}+L_{12}(30+6L_7+18L_4+2L_2)+(3+L_4+L_7+L_4L_7)L_2\\
&&~~~~~~+2(1-L_4)L_8+8L_5-2(5+L_4)L_{14}-2(3-L_4)L_1\\
&&~~~~~~+4\mathrm{i}L_{15}+\Pa{1+6L_7+L^2_7+10L_4+27L_4L_7+L^2_7L_4+5L^2_4+3L_7L^2_4}\Big].
\end{eqnarray*}
\item $B_{14}=\Tr{\bsX_0\bsX_1\bsX_2\bsX^2_0\bsX_2}$ is
given by
\begin{eqnarray*}
B_{14}&=&\frac1{128}\Big[4L^2_{12}+L_{12}(30+6L_4+18L_7+2L_2)+(3+L_4+L_7+L_4L_7)L_2\\
&&~~~~~~+2(1-L_7)L_5+8L_8-2(5+L_7)L_{14}-2(3-L_7)L_1\\
&&~~~~~~+4\mathrm{i}L_{16}+\Pa{1+6L_4+L^2_4+10L_7+27L_4L_7+L^2_4L_7+5L^2_7+3L_4L^2_7}\Big].
\end{eqnarray*}
\item $B_{15}=\Tr{\bsX_0\bsX_1\bsX_2\bsX^3_0\bsX_1}$ is given by
\begin{eqnarray*}
B_{15} &=& \frac1{512}\Big[1+L^3_4+26 L_7 L^2_4+15 L^2_4+13 L^2_7
L_4+76 L_4 L_7+15 L_4+5 L^2_7+4 L_3+10 L_7-L^2_2(L_4-1)\\
&&~~~~~~+26 L_5+12 L_8+6(L_4+L_7) L_5-4 L_6+68 L_{12}+88 L_4
L_{12}+4 L^2_4 L_{12}+4 L_5 L_{12}+44 L_7 L_{12}\\
&&~~~~~~+12 L_4 L_7 L_{12}+28 L^2_{12}-4 L_1(6+2 L_7+L_4 (L_7+4)+3
L_{12})\\
&&~~~~~~+2 L_2(3+L_5+3 L_7+12 L_{12}+L_4(4 L_7+2 L_{12}+5))\\
&&~~~~~~-4 L_4 L_{13}-4 L_{13}-12 L_4 L_{14}-4 L_7 L_{14}-4 L_{12}
L_{14}-44 L_{14}+ \mathrm{i}(16L_{15}-4L_{16}+4L_{17})\Big].
\end{eqnarray*}
\item $B_{16}=\Tr{\bsX_0\bsX_1\bsX_2\bsX^3_0\bsX_2}$ is given by
\begin{eqnarray*}
B_{16} &=& \frac1{512}\Big[1+L^3_7+26 L_4 L^2_7+15 L^2_7+13 L^2_4
L_7+76 L_4 L_7+15 L_7+5 L^2_4+4 L_3+10 L_4-L^2_2(L_7-1)\\
&&~~~~~~+26 L_8+12 L_5+6(L_4+L_7) L_8-4 L_9+68 L_{12}+88 L_7
L_{12}+4 L^2_7 L_{12}+4 L_8 L_{12}+44 L_4 L_{12}\\
&&~~~~~~+12 L_4 L_7 L_{12}+28 L^2_{12}-4 L_1(6+2 L_4+L_7 (L_4+4)+3
L_{12})\\
&&~~~~~~+2 L_2(3+L_8+3 L_4+12 L_{12}+L_7(4 L_4+2 L_{12}+5))\\
&&~~~~~~-4 L_7 L_{13}-4 L_{13}-12 L_7 L_{14}-4 L_4 L_{14}-4 L_{12}
L_{14}-44 L_{14}+ \mathrm{i} (16L_{16}-4L_{15}+4L_{18})\Big].
\end{eqnarray*}
\item $B_{17}=\Tr{\bsX_0\bsX_1\bsX^2_0\bsX_1\bsX^3_0\bsX_1}$ is given by
{\scriptsize\begin{eqnarray*} B_{17}&=& \frac1{8192}\Big[1+L^3_7+15
L^2_7+12 L_3 L_7+15 L_7+48
L^2_1+60 L_3-L^3_2(L_4-1)+36 L_4+48 L_3 L_4+315 L_4L_7\\
&&~~~~~~+ 3 L_4 L^3_7+150 L_4 L^2_7+75 L^2_4 L^2_7+126 L^2_4-12 L_3
L^2_4+525 L^2_4 L_7+9 L^4_4+84 L^3_4+105L^3_4 L_7\\
&&~~~~~~+60 L_8+48 L_4 L_8+12 L_7 L_8-4 L_4 L_7 L_8-4 L_9 L_4-12 L_8
L^2_4+4 L_5 L^2_7+224 L_5+132 L_5 L_7\\
&&~~~~~~+448 L_4 L_5-4 L_5 L_8+108 L_4 L_7 L_5+96 L^2_4 L_5+32
L^2_5+8L_4 L_6 -12 L_7 L_6+24 L_6\\
&&~~~~~~+24 L_3 L_{12}+300 L_7 L_{12}+210 L_{12}+18 L_{12}
L^2_7+1050 L_4 L_{12}+24 L_8 L_{12}+600 L_4 L_7 L_{12}-8 L_3 L_4
L_{12}\\
&&~~~~~~+30 L^3_4 L_{12}+60 L^2_4 L_7 L_{12}+630 L^2_4 L_{12}+6
L_4L^2_7 L_{12}-8 L_4 L_8 L_{12}+336 L_5 L_{12}+48 L_4 L_5
L_{12}\\
&&~~~~~~-16 L_6 L_{12}+8 L_5 L_7 L_{12}+16 L^3_{12}+552 L^2_{12}+312
L_4 L^2_{12}+40 L_7 L^2_{12}+8 L_{12} L_{13}\\
&&~~~~~~+L^2_2(4 L_5+(L_4+3) L_7-L_4(L_4+2(L_{12}-9))+18
L_{12}+15)\\
&&~~~~~~+L_2((5 L_4+3)L^2_7+2(21 L^2_4+84 L_4+4 L_5+2 (L_4+9)
L_{12}+15) L_7+36 L^2_{12}-4 L_3(L_4-3)\\
&&~~~~~~+91 L_4+132 L_5+L_4(L_4 (9 L_4+77)+44 L_5-4 L_8)-8 L_6+12
L_8+4(3 L_4(L_4+22)+2 L_5)L_{12}\\
&&~~~~~~+300 L_{12}-96 L_{14}+15)-400 L_4 L_{14}-48 L_4 L_7
L_{14}-96 L_7 L_{14}-240 L_{14}-64 L^2_4 L_{14}-48 L_5 L_{14}\\
&&~~~~~~+20 L^2_{14}-32 L_4 L_{12} L_{14}-256 L_{12} L_{14}+8
L_1(L_2((L_4-6) L_4-15)-10 L_5+((L_4-16) L_4-15) L_7\\
&&~~~~~~-48 L_{12}+6 L_{14}-L_4(L_4(L_4+27)+12 L_{12}-2
L_{14}+61)-15)+16 \mathrm{i} (L_4 L_{18}-L_{12} L_{15}-L_{10})\Big].
\end{eqnarray*}}
\item $B_{18}=\Tr{\bsX_0\bsX_2\bsX^2_0\bsX_2\bsX^3_0\bsX_2}$ is given by
{\scriptsize\begin{eqnarray*} B_{18}&=& \frac1{8192}\Big[1+L^3_4+15
L^2_4+12 L_3 L_4+15 L_4+48
L^2_1+60 L_3-L^3_2(L_7-1)+36 L_7+48 L_3 L_7+315 L_4L_7\\
&&~~~~~~+ 3 L_7 L^3_4+150 L_7 L^2_4+75 L^2_4 L^2_7+126 L^2_7-12 L_3
L^2_7+525 L^2_7 L_4+9 L^4_7+84 L^3_7+105L^3_7 L_4\\
&&~~~~~~+60 L_5+48 L_7 L_5+12 L_4 L_5-4 L_4 L_7 L_5-4 L_6 L_7-12 L_5
L^2_7+4 L_8 L^2_4+224 L_8+132 L_8 L_4\\
&&~~~~~~+448 L_7 L_8-4 L_5 L_8+108 L_4 L_7 L_8+96 L^2_7 L_8+32
L^2_8+8 L_7 L_9-12 L_4 L_9+24 L_9\\
&&~~~~~~+24 L_3 L_{12}+300 L_4 L_{12}+210 L_{12}+18 L_{12}
L^2_4+1050 L_7 L_{12}+24 L_5 L_{12}+600 L_4 L_7 L_{12}-8 L_3 L_7
L_{12}\\
&&~~~~~~+30 L^3_7 L_{12}+60 L^2_7 L_4 L_{12}+630 L^2_7 L_{12}+6
L_7L^2_4 L_{12}-8 L_7 L_5 L_{12}+336 L_8 L_{12}+48 L_7 L_8 L_{12}\\
&&~~~~~~-16 L_9 L_{12}+8 L_4 L_8 L_{12}+16 L^3_{12}+552 L^2_{12}+312
L_7 L^2_{12}+40 L_4 L^2_{12}+8 L_{12} L_{13}\\
&&~~~~~~+L^2_2(4 L_8+(L_7+3) L_4-L_7(L_7+2(L_{12}-9))+18
L_{12}+15)\\
&&~~~~~~+L_2((5 L_7+3)L^2_4+2(21 L^2_7+84 L_7+4 L_8+2 (L_7+9)
L_{12}+15) L_4+36 L^2_{12}-4 L_3(L_7-3)\\
&&~~~~~~+91 L_7+132 L_8+L_7(L_7 (9 L_7+77)+44 L_8-4 L_5)-8 L_9+12
L_5+4(3 L_7(L_7+22)+2 L_8)L_{12}\\
&&~~~~~~+300 L_{12}-96 L_{14}+15)-400 L_7 L_{14}-48 L_4 L_7
L_{14}-96 L_4 L_{14}-240 L_{14}-64 L^2_7 L_{14}-48 L_8 L_{14}\\
&&~~~~~~+20 L^2_{14}-32 L_7 L_{12} L_{14}-256 L_{12} L_{14}+8
L_1(L_2((L_7-6) L_7-15)-10 L_8+((L_7-16) L_7-15) L_4\\
&&~~~~~~-48 L_{12}+6 L_{14}-L_7(L_7(L_7+27)+12 L_{12}-2
L_{14}+61)-15)+16 \mathrm{i} (L_7 L_{17}-L_{12} L_{16}-L_{11})\Big].
\end{eqnarray*}}
\end{enumerate}
Those Makhlin invariants $L_k$'s can be also expressed by using
Bargmann invariants $B_k$'s below:
\begin{enumerate}
\item[(1')] $L_1= \frac23(1-3B_1-3B_2+6B_3+3B_4-4B_6)$.
\item[(2')] $L_2 = 1-2B_1-2B_2+4B_4$.
\item[(3')] $L_3 = 4(1+B_1 B_2-3 B_1-3 B_2+6 B_3+B_4-B^2_4+B_1 B_4+B_2 B_4-4 B_7-4 B_8+4
B_{10})$.
\item[(4')] $L_4 = 2B_1-1$.
\item[(5')] $L_5 = 2B_2 + 4B_4 - 4B_1B_4 - 8B_5 - 8B_6 + 16B_7 - 1$.
\item[(6')] $L_6 = \frac43 \Big(1+4B_1 - 9B^2_1-18B_1B_2+6B^2_1B_2 - 3B^2_2 + 24B_1B_3 + 12B_2B_3 - 12B^2_3 - 12B_4 + 6B_2B_4+ 18B_1B_4+6B^2_1B_4-12B_3B_4 + 3B^2_4 - 6B_1B^2_4 + 12B_5 - 12B_1B_5 + 20B_6 - 4B_1B_6 - 24B_7 - 24B_1B_7 -12B_{10} + 12B_1B_{10} +
24B_{11}\Big)$.
\item[(7')] $L_7 = 2B_2-1$.
\item[(8')] $L_8 = 2B_1 - 4B_2B_4 + 4B_4 - 8B_5 - 8B_6 + 16B_8 - 1$.
\item[(9')] $L_9 = \frac43 \Big(1+4B_2 - 9B^2_2-18B_1B_2+6B_1B^2_2 - 3B^2_1 + 24B_2B_3 + 12B_1B_3 - 12B^2_3 - 12B_4 + 6B_1B_4+ 18B_2B_4+6B^2_2B_4-12B_3B_4 + 3B^2_4 - 6B_2B^2_4 + 12B_5 - 12B_2B_5 + 20B_6 - 4B_2B_6 - 24B_8- 24B_2B_8 -12B_{10} + 12B_2B_{10} +
24B_{12}\Big)$.
\item[(10')] $L_{10}=\frac23\mathrm{i}\Big(27-97 B_1+114 B^2_1-46 B^3_1 -81B_2+178 B_1 B_2-64 B^2_1 B_2+78 B^2_2+108 B_1 B^2_2+18 B^3_2+172 B_3-368B_1B_3+168 B^2_1B_3-384B_2B_3-288 B_1 B_2 B_3-144 B^2_2 B_3+456 B^2_3-288 B^3_3+120 B_1 B^2_3+360 B_2 B^2_3-18 B_4+54 B_1 B_4+54 B_2 B_4-137 B^2_1B_4 +48B^3_1 B_4 -390 B_1 B_2 B_4-141 B^2_2 B_4+72 B_1 B^2_2 B_4-108 B_3 B_4+660 B_1 B_3 B_4+48 B^2_1 B_3 B_4+540 B_2 B_3 B_4-192 B_1 B_2 B_3 B_4-480 B^2_3 B_4-129 B^2_4+261 B_1 B^2_4+72 B^2_1 B^2_4+81 B_2 B^2_4-48 B_1 B_2 B^2_4-144 B_3 B^2_4-12 B_5-68 B_1 B_5+96 B^2_1 B_5+36 B_2 B_5-144 B_1 B_2 B_5+144 B_3 B_5-96 B_1 B_3 B_5+60 B_4 B_5+60 B_1 B_4 B_5-36 B_2 B_4 B_5+96 B_3 B_4 B_5+48 B^2_4 B_5+88 B_6-92 B_1 B_6-32 B^3_1 B_6-228 B_2 B_6-40 B_1 B_2 B_6+64 B^2_1 B_2 B_6+488 B_3 B_6+32 B_1 B_3 B_6+336 B_4 B_6+32 B^2_1 B_4 B_6-36 B_2 B_4 B_6+96 B_3 B_4 B_6+48 B^2_4 B_6+48 B_5 B_6-64 B_1 B_5 B_6-40 B^2_6-16 B_7+132 B_1 B_7+96 B^2_1 B_7+228 B_2 B_7+384 B_1 B_2 B_7-624 B_3 B_7-768 B_1 B_3 B_7-552 B_4 B_7-96 B_1 B_4 B_7+72 B_2 B_4 B_7-192 B_3 B_4 B_7-96 B^2_4 B_7+144 B_5 B_7+400 B_6 B_7-768 B^2_7+36 B_1 B_8-96 B^2_1 B_8+84 B_2 B_8-48 B_1 B_2 B_8-240 B_3 B_8+288 B_1 B_3 B_8-24 B_4 B_8+72 B_1 B_4 B_8-48 B_5 B_8-48 B_6 B_8+96 B_7 B_8+24 B_9-48 B_1 B_9-144 B_2 B_9+192 B_1 B_2 B_9+96 B_3 B_9-90 B_{10}+138 B_1 B_{10}-96 B^2_1B_{10}+162 B_2 B_{10}-144 B_1 B_2 B_{10}-288 B_3 B_{10}+192 B_1 B_3 B_{10}-144 B_4 B_{10}+96 B_1 B_4 B_{10}-12 B_{11}-192 B_1 B_{11}-72B_1B_{12}-72B_2B_{11}+192B_3B_{11}+96 B_4 B_{11}+36 B_{12}+96 B_{13}-192 B_1 B_{13}-96 B_2 B_{13}+192 B_3 B_{13}-192 B_{14}+384 B_1 B_{14}-384 B_1 B_{16}+192 B_{16}+768
B_{17}\Big)$.
\item[(11')] $L_{11}=\frac23\mathrm{i}\Big(27-97 B_2+114 B^2_2-46 B^3_2 -81B_1+178 B_1 B_2-64 B^2_2 B_1+78 B^2_1+108 B_2 B^2_1+18 B^3_1+172 B_3-368B_2B_3+168 B^2_2B_3-384B_1B_3-288 B_1 B_2 B_3-144 B^2_1 B_3+456 B^2_3-288 B^3_3+120 B_2 B^2_3+360 B_1 B^2_3-18 B_4+54 B_1 B_4+54 B_2 B_4-137 B^2_2B_4 +48B^3_2 B_4 -390 B_1 B_2 B_4-141 B^2_1 B_4+72 B^2_1 B_2 B_4-108 B_3 B_4+660 B_2 B_3 B_4+48 B^2_2 B_3 B_4+540 B_1 B_3 B_4-192 B_1 B_2 B_3 B_4-480 B^2_3 B_4-129 B^2_4+261 B_2 B^2_4+72 B^2_2 B^2_4+81 B_1 B^2_4-48 B_1 B_2 B^2_4-144 B_3 B^2_4-12 B_5-68 B_2 B_5+96 B^2_2 B_5+36 B_1 B_5-144 B_1 B_2 B_5+144 B_3 B_5-96 B_2 B_3 B_5+60 B_4 B_5+60 B_2 B_4 B_5-36 B_1 B_4 B_5+96 B_3 B_4 B_5+48 B^2_4 B_5+88 B_6-92 B_2 B_6-32 B^3_2 B_6-228 B_1 B_6-40 B_1 B_2 B_6+64 B^2_2 B_1 B_6+488 B_3 B_6+32 B_2 B_3 B_6+336 B_4 B_6+32 B^2_2 B_4 B_6-36 B_1 B_4 B_6+96 B_3 B_4 B_6+48 B^2_4 B_6+48 B_5 B_6-64 B_2 B_5 B_6-40 B^2_6-16 B_8+132 B_2 B_8+96 B^2_2 B_8+228 B_1 B_8+384 B_1 B_2 B_8-624 B_3 B_8-768 B_2 B_3 B_8-552 B_4 B_8-96 B_2 B_4 B_8+72 B_1 B_4 B_8-192 B_3 B_4 B_8-96 B^2_4 B_8+144 B_5 B_8+400 B_6 B_8-768 B^2_8+36 B_2 B_7-96 B^2_2 B_7+84 B_1 B_7-48 B_1 B_2 B_7-240 B_3 B_7+288 B_2 B_3 B_7-24 B_4 B_7+72 B_2 B_4 B_7-48 B_5 B_7-48 B_6 B_7+96 B_7 B_8+24 B_9-48 B_2 B_9-144 B_1 B_9+192 B_1 B_2 B_9+96 B_3 B_9-90 B_{10}+138 B_2 B_{10}-96 B^2_2B_{10}+162 B_1 B_{10}-144 B_1 B_2 B_{10}-288 B_3 B_{10}+192 B_2 B_3 B_{10}-144 B_4 B_{10}+96 B_2 B_4 B_{10}-12 B_{12}-192 B_2 B_{12}-72B_1B_{12}-72B_2B_{11}+192B_3B_{12}+96 B_4 B_{12}+36 B_{11}+96 B_{14}-192 B_2 B_{14}-96 B_1 B_{14}+192 B_3 B_{14}-192 B_{13}+384 B_2 B_{13}-384 B_2 B_{15}+192 B_{15}+768
B_{18}\Big)$.
\item[(12')] $L_{12} = 1 - 2B_1 - 2B_2 + 4B_3$.
\item[(13')] $L_{13} = 12(B_1+B_2) - 12(B_1+B_2)B_4 - 36B_3 + 24B_3B_4 + 24
B_5- 8B_6 + 16(B_7+B_8)- 32B_9 - 3$.
\item[(14')] $L_{14} = 2(1-3B_1-3B_2+2B_1B_2+6 B_3+B_4-4 B_5)$.
\item[(15')] $L_{15} =\frac43\mathrm{i}\Big(-1+5B_1-6 B^2_1+3B_2+3 B_1 B_2-12 B_3+6 B_1 B_3-6 B_2 B_3+12 B^2_3+6 B_4-12 B_1 B_4-6 B_2 B_4+6 B_1 B_2 B_4+6 B_3 B_4-6 B_5+12 B_1 B_5-14 B_6+4 B_1 B_6+24 B_7+12 B_8-12 B_1 B_8-24
B_{13}\Big)$.
\item[(16')] $L_{16} =\frac43\mathrm{i}\Big(-1+5B_2-6 B^2_2+3B_1+3 B_1 B_2-12 B_3+6 B_2 B_3-6 B_1 B_3+12 B^2_3+6 B_4-12 B_2 B_4-6 B_1 B_4+6 B_1 B_2 B_4+6 B_3 B_4-6 B_5+12 B_2 B_5-14 B_6+4 B_2 B_6+24 B_8+12
B_7-12 B_2 B_7-24 B_{14}\Big)$.
\item[(17')] $L_{17}=\frac43\mathrm{i}\Big(-9+15B_1+6B^2_1+19B_2-11B_1B_2-24B_3-6B_1B_3+6B_2B_3-12B^2_3+18B_4-24B_1B_4+6B^2_1B_4-30B_2B_4+12B_1B_2B_4+42 B_3 B_4-24 B_1 B_3 B_4+6B^2_4-6B_1B^2_4+6 B_5-48 B_1 B_5+12 B_2 B_5-12 B_4 B_5-18 B_6-4 B_2 B_6+8 B_1 B_2 B_6-12 B_4 B_6-12 B_7-12 B_2 B_7+48 B_3 B_7+24 B_4 B_7+24 B_1 B_8+48 B_1 B_9+24 B_{10}-12 B_1 B_{10}-24 B_{11}+96 B_{13}-24 B_{14}-96
B_{15}\Big)$.
\item[(18')] $L_{18}=\frac43\mathrm{i}\Big(-9+15B_2+6B^2_2+19B_1-11B_1B_2-24B_3-6B_2B_3+6B_1B_3-12B^2_3+18B_4-24B_2B_4+6B^2_2B_4-30B_1B_4+12B_1B_2B_4+42 B_3 B_4-24 B_2 B_3 B_4+6B^2_4-6B_2B^2_4+6 B_5-48 B_2 B_5+12 B_1 B_5-12 B_4 B_5-18 B_6-4 B_1 B_6+8 B_1 B_2 B_6-12 B_4 B_6-12 B_8-12 B_1 B_8+48 B_3 B_8+24 B_4 B_8+24 B_2 B_7+48 B_2 B_9+24 B_{10}-12 B_2 B_{10}-24 B_{12}+96 B_{14}-24 B_{13}-96
B_{16}\Big)$.
\end{enumerate}
\end{lem}

\begin{proof}
The correctness of all of these results can be checked by the
mathematical software \textsc{Mathematica}. We remark here that
deriving these results is more challenging than verifying them. All
materials preceding this lemma serve as preparations for simplifying
the calculations in the proof of this lemma. In fact, we expand
$B_k$'s by using the Bloch decomposition of $\rho_{AB}$. Through
tedious algebraic computations and simplifications, utilizing the
results from Subsections~\ref{sub:1},~\ref{sub:2}, and~\ref{sub:3},
we obtain the desired results.
\end{proof}

\subsection{Proof of Theorem~\ref{th:Bargmann-generators}}
With the above preparations, now we can present the proof of
Theorem~\ref{th:Bargmann-generators}.
\begin{proof}[Proof of Theorem~\ref{th:Bargmann-generators}]
We have already known that the set comprising of 18 Makhlin's
fundamental invariants $I_k$'s, where $I_k$'s can be generated by
$L_k$'s in Proposition~\ref{prop:LI}, provides a complete
description of nonlocal properties of the two-qubit state
\cite{Makhlin2002}. This amounts to say that the set of 18
invariants $L_k$'s can completely determine the local unitary orbit
of the two-qubit state. From Lemma~\ref{lem:LB}, we see that $L_k$'s
can be generated by $B_k$'s. Therefore, the set of 18 local unitary
Bargmann invariants $B_k$'s can determine the local unitary orbit of
the two-qubit state. That is, two states of a two-qubit system are
LU equivalent if and only if both states have equal values of all 18
LU Bargmann invariants.
\end{proof}

\section{Proof of Theorem~\ref{th:ent-test}}\label{app:3}

\subsection{Entanglement criterion by Makhlin's invariants}

Let the partial trace with respect to either one subsystem of
$\rho_{AB}$ be given by $\rho^\Gamma_{AB}=\rho^{\t_A}_{AB}$ or
$\rho^{\t_B}_{AB}$. We have the following result:
\begin{lem}\label{lem:eigenvalues}
All eigenvalues of the operator $\bsX:=4\rho_{AB}-\I_2\ot\I_2$ are
determined by its characteristic polynomial equation
$x^4+px^2+qx+r=0$, where
\begin{eqnarray}
\begin{cases}
p = -2(L_2+ L_4+ L_7),\\
q = -8(L_{12}-L_1), \\
r = L^2_2+2(L_4+ L_7)L_2+(L_4- L_7)^2-4(L_3+L_5+L_8)+8L_{14}.
\end{cases}
\end{eqnarray}
Here the meaning of $L_k$'s can be found in
Proposition~\ref{prop:LI}.
\end{lem}

\begin{proof}
The proof is obtained by direct and tedious computations. It is
omitted here.
\end{proof}
We remark here that the correctness of the above result can also be
checked by employing symbolic computation function of
\textsc{Mathematica}. Apparently, getting this result is more
difficult than checking the correctness of it. Based on the above
result presented in Lemma~\ref{lem:eigenvalues}, we can derive the
following characterization of entanglement in two-qubit system.
Basically, it is another equivalent reformulation of Positive
Partial-Tranpose criteria for two-qubit system. More importantly,
our reformulation can be viewed as the first criterion using locally
unitary invariants.

For any two-qubit state $\rho_{AB}$, parameterized as in
Eq.~\eqref{eq:2qubit}, note that
\begin{eqnarray}
\Tr{\rho^2_A}=\frac{1+L_4}2,\quad \Tr{\rho^2_B}=\frac{1+L_7}2,\quad
\Tr{\rho^2_{AB}}=\frac{1+L_2+L_4+L_7}4,
\end{eqnarray}
from the facts that $\Tr{\rho^2_A},\Tr{\rho^2_B}\in[\frac12,1]$ and
$\Tr{\rho^2_{AB}}\in[\frac14,1]$, we get that
\begin{eqnarray}\label{eq:app1}
\begin{cases}
0\leqslant L_4\leqslant1,\\
0\leqslant L_7\leqslant1,\\
0\leqslant L_2+ L_4+ L_7\leqslant 3.
\end{cases}
\end{eqnarray}
It follows from Lemma~\ref{lem:eigenvalues}, we get the
characteristic polynomial equation is given by
\begin{eqnarray}
\lambda^4-\lambda^3+\frac{p+6}{16}\lambda^2-\frac{2p-q+4}{64}\lambda+\frac{p-q+r+1}{256}=0.
\end{eqnarray}
Recall a result in \cite{Kimura2003}: Consider an algebraic equation
of degree $N\geqslant1$,
\begin{eqnarray}
\prod^N_{k=1}(x-x_k)=\sum^N_{\ell=0}(-1)^\ell e_\ell
x^{N-\ell}=0\quad(e_0=1),
\end{eqnarray}
which has only real roots $x_k\in\real(k=1,\ldots,N)$.  The
necessary and sufficient condition that all the roots $x_k$'s to be
non-negative is that all the coefficients $e_\ell$'s are
non-negative. That is,
\begin{eqnarray}
(\forall k\in[N]: x_k\geqslant0)\Longleftrightarrow (\forall
\ell\in[N]: e_\ell\geqslant0, e_0\equiv1).
\end{eqnarray}
From the above result, we can present a following result about the
positivity of Hermitian matrix $\bsX$:
\begin{prop}
For a Hermitian complex matrix $\bsX\in\complex^{N\times N}$, denote
$p_k(\bsX):=\Tr{\bsX^k}$, then its characteristic polynomial is
given by
$$
\det(x\I_N-\bsX)=\sum^N_{k=0}(-1)^k e_k(\bsX) x^{N-k},
$$
where
\begin{eqnarray*}
e_k (\bsX) = \frac1{k!}\Abs{\begin{array}{ccccc}
             p_1(\bsX) & 1 & 0 & \cdots & 0 \\
             p_2(\bsX) & p_1(\bsX) & 2 & \cdots & 0 \\
             \vdots & \vdots & \vdots & \ddots & \vdots \\
             p_{k-1}(\bsX) & p_{k-2}(\bsX) & p_{k-3}(\bsX) & \cdots & k-1\\
             p_k(\bsX) & p_{k-1}(\bsX) & p_{k-2}(\bsX) & \cdots & p_1(\bsX)
           \end{array}
}\quad (k\geqslant1).
\end{eqnarray*}
Then we have
\begin{eqnarray*}
\bsX\geqslant\zero \Longleftrightarrow \Abs{\begin{array}{ccccc}
             p_1(\bsX) & 1 & 0 & \cdots & 0 \\
             p_2(\bsX) & p_1(\bsX) & 2 & \cdots & 0 \\
             \vdots & \vdots & \vdots & \ddots & \vdots \\
             p_{k-1}(\bsX) & p_{k-2}(\bsX) & p_{k-3}(\bsX) & \cdots & k-1\\
             p_k(\bsX) & p_{k-1}(\bsX) & p_{k-2}(\bsX) & \cdots & p_1(\bsX)
           \end{array}
}\geqslant0\quad(k=1,2,\ldots,N).
\end{eqnarray*}
\end{prop}

\begin{proof}
Since $\bsX$ is Hermitian matrix, it follows that its characteristic
polynomial $\det(x\I_N-\bsX)=\sum^N_{k=0}(-1)^k e_k(\bsX) x^{N-k}$
has only real roots. These real roots are non-negative if and only
if $\bsX\geqslant\zero$. Therefore $\bsX\geqslant\zero$ if and only
if $e_k(\bsX)\geqslant0$, where $k=1,\ldots,N$
\end{proof}
From the above result, the non-negativeness of $\rho_{AB}$ is
guaranteed by the following inequalities \cite{Kimura2003}:
\begin{eqnarray}\label{eq:app2}
\begin{cases}
p+6&\geqslant0\\
2p-q+4&\geqslant0\\
p-q+r+1&\geqslant0
\end{cases}\Longleftrightarrow \begin{cases}
p&\geqslant-6\\
q&\leqslant2p+4\\
r&\geqslant q-p-1.
\end{cases}
\end{eqnarray}
Based on both Eq.~\eqref{eq:app1} and Eq.~\eqref{eq:app2}, we can
summarize the above discussion into the following result:
\begin{prop}
For any Hermitian matrix $\rho_{AB}$ of fixed trace one,
parameterized as
\begin{eqnarray}
\rho_{AB} =
\frac14\Pa{\I_2\ot\I_2+\bsa\cdot\boldsymbol{\sigma}\ot\I_2+\I_2\ot
\bsb\cdot\boldsymbol{\sigma}+\sum^3_{i,j=1}c_{ij}\sigma_i\ot\sigma_j},
\end{eqnarray}
where $\bsa=(a_1,a_2,a_3)^\t$ and $\bsb= (b_1,b_2,b_3)^\t$ are in
$\real^3$, and $\bsC=(c_{ij})_{3\times3}\in\real^{3\times 3}$, the
necessary and sufficient condition for the non-negativeness
$\rho_{AB}\geqslant\zero$ if and only if the following inequalities
concerning the 3-tuple $(\bsa,\bsb,\bsC)$ are true:
\begin{eqnarray}
\begin{cases}
0\leqslant L_4\leqslant1,\\
0\leqslant L_7\leqslant1,\\
0\leqslant L_2+L_4+L_7\leqslant 3,\\
L_2+L_4+L_7\leqslant
1+2(L_{12}-L_1),\\
(L_2+L_4+L_7-1)^2-4(L_3+L_4L_7+L_5+L_8)+8(L_{12}+L_{14}-L_1)\geqslant
0.
\end{cases}
\end{eqnarray}
\end{prop}

The above constraints about the 3-tuple $(\bsa,\bsb,\bsC)$ can be
equivalently to reformulated via locally unitary Bargmann
invariants:
\begin{eqnarray}
\begin{cases}
1+2\Tr{\rho^3_{AB}}&\geqslant 3\Tr{\rho^2_{AB}},\\
1+3[\Tr{\rho^2_{AB}}]^2+8\Tr{\rho^3_{AB}}&\geqslant6\Tr{\rho^4_{AB}}+6\Tr{\rho^2_{AB}}.
\end{cases}
\end{eqnarray}

\begin{lem}[Detection of entanglement via locally unitary
invariants]\label{th:entvssep} For any given two-qubit state
$\rho_{AB}$, parameterized as in Eq.~\eqref{eq:2qubit}, which is
entangled if and only if 9 invariants of 18 Makhlin invariants are
satisfying the following inequality:
\begin{eqnarray}\label{eq:ent}
&&1+\Pa{\abs{\bsa}^2-\abs{\bsb}^2}^2+2(\abs{\bsa}^2+\abs{\bsb}^2)\Inner{\bsC}{\bsC}+2\Inner{\bsC^\t\bsC}{\bsC^\t\bsC}
+ 8(\Innerm{\bsa}{\bsC}{\bsb}+\det(\bsC))\notag\\
&&<\Inner{\bsC}{\bsC}^2+2\Pa{\abs{\bsa}^2+\abs{\bsb}^2+\Inner{\bsC}{\bsC}}+4\Pa{\Innerm{\bsa}{\bsC\bsC^\t}{\bsa}+\Innerm{\bsb}{\bsC^\t\bsC}{\bsb}}+8\Innerm{\bsa}{\widehat\bsC}{\bsb}.
\end{eqnarray}
\end{lem}

\begin{proof}
All eigenvalues of the operator
$\bsY:=4\rho^\Gamma_{AB}-\I_2\ot\I_2$ are determined by its
characteristic polynomial equation $y^4+\tilde py^2+\tilde qy+\tilde
r=0$, where
\begin{eqnarray*}
\tilde p &=&
-2\Pa{\abs{\bsa}^2+\abs{\bsb}^2+\Inner{\bsC}{\bsC}},\quad
\tilde q = -8\Pa{\Innerm{\bsa}{\bsC}{\bsb}+\det(\bsC)}, \\
\tilde r &=&
\Pa{\abs{\bsa}^2-\abs{\bsb}^2}^2+2(\abs{\bsa}^2+\abs{\bsb}^2)\Inner{\bsC}{\bsC}+2\Inner{\bsC^\t\bsC}{\bsC^\t\bsC}-\Inner{\bsC}{\bsC}^2\\
&&
-4\Pa{\Innerm{\bsa}{\bsC\bsC^\t}{\bsa}+\Innerm{\bsb}{\bsC^\t\bsC}{\bsb}}-8\Innerm{\bsa}{\widehat\bsC}{\bsb}.
\end{eqnarray*}
Note that $\det(\rho^\Gamma_{AB})=\frac{\tilde p-\tilde q+\tilde
r+1}{256}$. Thus $\rho_{AB}$ is entangled if and only if
$\det(\rho^\Gamma_{AB})<0$. Therefore we get the desired inequality.
\end{proof}

\begin{exam}[The family of two-qubit Werner states] Two-qubit Wener state of
single parameter is defined by $\rho_w =
w\proj{\psi^-}+(1-w)\frac{\I_4}4$, where
$\ket{\psi^-}=\frac{\ket{01}-\ket{10}}{\sqrt{2}}$ and $w\in[0,1]$,
which can be rewritten as
$$
\rho_w =\frac14\Pa{\I_2\ot\I_2-w\sum^3_{k=1}\sigma_k\ot\sigma_k}.
$$
In such a case, $\bsa=\bsb=\zero$ and $\bsC=-w\I_3$. Then two-qubit
Werner state $\rho_w$ is entangled if and only if Eq.~\eqref{eq:ent}
becomes
\begin{eqnarray}
&&1+2\Inner{\bsC^\t\bsC}{\bsC^\t\bsC}+8\det(\bsC)<\Inner{\bsC}{\bsC}^2+2\Inner{\bsC}{\bsC}\\
&&\Longleftrightarrow1+6w^4-8w^3<9w^4+6w^2\Longleftrightarrow
\frac13<w\leqslant1.\notag
\end{eqnarray}
\end{exam}

\begin{exam}[The family of two-qubit Bell-diagonal states]
Two-qubit Bell-diagonal state of three parameters is defined by
$$
\rho_{\mathrm{Bell}} =\frac14\Pa{\I_2\ot\I_2 +
\sum^3_{k=1}t_k\sigma_k\ot\sigma_k},
$$
where $\bst=(t_1,t_2,t_3)\in D$ (specified later). The set $D$ is a
bounded and closed region: $D\subset[-1,1]^3$. The above mentioned
$D$ is determined by
\begin{eqnarray*}
\begin{cases}
1-t_1-t_2-t_3\geqslant0,\\
1-t_1+t_2+t_3\geqslant0,\\
1+t_1-t_2+t_3\geqslant0,\\
1+t_1+t_2-t_3\geqslant0.
\end{cases}
\end{eqnarray*}
In this case, $\bsa=\bsb=\zero$ and $\bsC=\diag(t_1,t_2,t_3)$. Now
two-qubit Bell-diagonal state $\rho_{\text{Bell}}$ is entangled if
and only if Eq.~\eqref{eq:ent} becomes
\begin{eqnarray*}
&&1+2\Inner{\bsC^\t\bsC}{\bsC^\t\bsC}+8\det(\bsC)<\Inner{\bsC}{\bsC}^2+2\Inner{\bsC}{\bsC}\\
&&\Longleftrightarrow 1+2\sum^3_{j=1}t^4_j +
8t_1t_2t_3<\Pa{\sum^3_{j=1}t^2_j}^2+2\sum^3_{j=1}t^2_j.
\end{eqnarray*}
Note that
\begin{eqnarray*}
&&\Pa{\sum^3_{j=1}t^2_j}^2+2\sum^3_{j=1}t^2_j-2\sum^3_{j=1}t^4_j -
8t_1t_2t_3-1\\
&&=-(t_1-t_2-t_3+1)(t_1+t_2-t_3-1)(t_1-t_2+t_3-1)(t_1+t_2+t_3+1)>0,
\end{eqnarray*}
which is equivalent to $\abs{t_1}+\abs{t_2}+\abs{t_3}>1$.
\end{exam}

\subsection{Proof of Theorem~\ref{th:ent-test}}

\begin{proof}[Proof of Theorem~\ref{th:ent-test}]
Note that we have obtained that a complete set of LU Bargmann
invariants $\set{B_k:k=1,\ldots,18}$ for the description of nonlocal
properties of the two-qubit state. Using the 18 Bargmann generators,
we can test the LU equivalence of two-qubit states by experiment via
measuring Bargmann invariants. Besides, we can use 7 Bargmann
invariants to test entanglement of two-qubit states: By using
Lemma~\ref{lem:LB}, Eq.~\eqref{eq:ent} can be equivalently
transformed into the following form:
\begin{eqnarray*}
6(B_1+B_2-B_1B_2-B_4-B_{10})+12(B_5-B_3)+3B^2_4+4B_6<1.
\end{eqnarray*}
This completes the proof.
\end{proof}



\begin{thebibliography}{99}

\bibitem{Jozsa2003}
R. Jozsa and N. Linden, On the role of entanglement in
quantum-computational speed-up, Proc. R. Soc. Lond. A
\href{https://doi.org/10.1098/rspa.2002.1097}{{\bf459}, 2011(2003).}

\bibitem{Portmann2022}
C. Portmann and R. Renner, Security in quantum cryptography, \rmp
\href{https://doi.org/10.1103/RevModPhys.94.025008}{{\bf94},
025008(2022).}

\bibitem{Bennett1992}
C.H. Bennett and S.J. Wiesner, Communication via one- and
two-particle operators on Einstein-Podolsky-Rosen states, \prl
\href{https://doi.org/10.1103/PhysRevLett.69.2881}{{\bf69},
2881(1992).}

\bibitem{Bennett1993}
C.H. Bennett, G. Brassard, C. Cr\'{e}peau, R. Jozsa, A. Peres, and
W. Wootters, Teleporting an unknown quantum state via dual classical
and EPR channels, \prl
\href{https://doi.org/10.1103/PhysRevLett.70.1895}{{\bf70},
1895(1993).}

\bibitem{Bouwmeester1997}
D. Bouwmeester, J.W. Pan, K. Mattle, M. Eibl, H. Weinfurter, and A.
Zeilinger, Experimental quantum teleportation, Nature
\href{https://doi.org/10.1038/37539}{{\bf390}, 575 (1997).}

\bibitem{Makhlin2002}
Y. Makhlin, Nonlocal properties of two-qubit gates and mixed states,
and the optimization of quantum computations, Quant Inf Process
\href{https://doi.org/10.1023/A:1022144002391}{{\bf1}, 243-252
(2002).}

\bibitem{Kraus2010}
B. Kraus, Local unitary equivalence of multipartite pure states,
\prl \href{https://doi.org/10.1103/PhysRevLett.104.020504}{{\bf104},
020504 (2010).}

\bibitem{Zhou2012}
C. Zhou, T. Zhang, S-M. Fei, N. Jing, and X. Li-Jost, Local unitary
equivalence of arbitrary dimensional bipartite quantum states, \pra
\href{https://doi.org/10.1103/PhysRevA.86.010303}{{\bf86}, 010303(R)
(2012).}

\bibitem{Jing2015}
N. Jing, S-M. Fei, M. Li, X. Li-Jost, and T. Zhang, Local unitary
invariants of generic multiqubit states, \pra
\href{https://doi.org/10.1103/PhysRevA.92.022306}{{\bf92}, 022306
(2015).}

\bibitem{Martins2015}
A.M. Martins, Necessary and Sufficient Conditions for Local Unitary
Equivalence of Multi-qubit States, \pra
\href{https://doi.org/10.1103/PhysRevA.91.042308}{{\bf91}, 042308
(2015)}

\bibitem{Gour2013}
G. Gour, N.R. Wallach, Classification of multipartite entanglement
of all finite dimensionality, \prl
\href{https://doi.org/10.1103/PhysRevLett.111.060502}{{\bf 111},
060502 (2013).}


\bibitem{Gour2011}
G. Gour, N.R. Wallach, Necessary and sufficient conditions for local
manipulation of multipartite pure quantum states, \njp
\href{https://doi.org/10.1088/1367-2630/13/7/073013}{{\bf13}, 073013
(2011).}

\bibitem{Vartiainen2005}
J.J. Vartiainen, Unitary Transformations for Quantum Computing,
\href{http://lib.tkk.fi/Diss/2005/isbn9512276127/isbn9512276127.pdf}{PhD
Thesis (2005).}


\bibitem{Buhrman2001}
H. Buhrman, R. Cleve, J. Watrous, and R. de Wolf, Quantum
fingerprinting, \prl
\href{https://doi.org/10.1103/PhysRevLett.87.167902}{{\bf87},
167902(2001).}

\bibitem{Simon1993}
R. Simon and N. Mukunda, Bargmann invariant and the geometry of the
G\"{u}oy effect, \prl
\href{https://doi.org/10.1103/PhysRevLett.70.880}{{\bf70},
880(1993).}

\bibitem{Mukunda2003}
N. Mukunda, P.K. Aravind, and R. Simon, Wigner rotations, Bargmann
invariants and geometric phases, \jpa: Math. Gen.
\href{https://doi.org/10.1088/0305-4470/36/9/312}{{\bf36}, 2347
(2003).}

\bibitem{Kirkwood1933}
J.G. Kirkwood, Quantum statistics of almost classical assemblies,
\href{https://doi.org/10.1103/PhysRev.44.31}{{\bf44}, 31(1933).}

\bibitem{Dirac1945}
P.A.M. Dirac, On the analogy between classical and quantum
mechanics, \rmp
\href{https://doi.org/10.1103/RevModPhys.17.195}{{\bf17},
195(1945).}

\bibitem{Bamber2014}
C. Bamber and J.S. Lundeen, Observing Dirac's classical phase space
analog to the quantum state, \prl
\href{https://doi.org/10.1103/PhysRevLett.112.070405}{{\bf112},
31(2014).}

\bibitem{Fernandes2024}
C. Fernandes, R. Wagner, L. Novo, and E.F. Galv\~{a}o,
Unitary-Invariant Witnesses of Quantum Imaginarity, \prl
\href{https://doi.org/10.1103/PhysRevLett.133.190201}{{\bf133},
190201 (2024).}

\bibitem{Oszmaniec2024}
M. Oszmaniec, D.J. Brod and E.F. Galv\~{a}o, Measuring relational
information between quantum states, and applications, \njp
\href{https://doi.org/10.1088/1367-2630/ad1a27}{{\bf26}, 013053
(2024).}

\bibitem{Quek2024}
Y. Quek, E. Kaur, and M.M. Wilde, Multivariate trace estimation in
constant quantum depth, Quantum
\href{https://doi.org/10.22331/q-2024-01-10-1220}{{\bf8}, 1220
(2024).}

\bibitem{Procesi1976}
C. Procesi, The invariant theory of $n\times n$ matrices, Adv.Math.
\href{https://doi.org/10.1016/0001-8708(76)90027-X}{{\bf19}, 306
(1976).}

\bibitem{Grassl1998}
M. Grassl, Computing local invariants of quantum-bit systems,
\pra~\href{https://doi.org/10.1103/PhysRevA.58.1833}{{\bf58}, 1833
(1998).}

\bibitem{Vrana2012}
P. Vrana,  Group representations in entanglement theory, PhD Thesis
(2012).

\bibitem{Brauer1937}
R. Brauer, On algebras which are connected with the semisimple
continuous groups, Ann. Math.
\href{https://doi.org/10.2307/1968843}{{\bf38}, 857(1937).}

\bibitem{Horodecki1996}
M. Horodecki, P. Horodecki, and R. Horodecki,  Separability of mixed
states: necessary and sufficient conditions,
\pla~\href{https://doi.org/10.1016/S0375-9601(96)00706-2}{{\bf223},
1 (1996).}

\bibitem{Yu2003}
S. Yu, J-W. Pan, Z-B. Chen, and Y-D. Zhang, Comprehensive test of
entanglement for two-level systems via the indeterminacy
relationship, \prl
\href{https://doi.org/10.1103/PhysRevLett.91.217903}{{\bf91}, 217903
(2003).}

\bibitem{Wyderka2023}
N. Wyderka, A. Ketterer, S. Imai, J.L. B\"{o}nsel, D.E. Jones, B.T.
Kirby, X-D. Yu, and O. G\"{u}hne, Complete characterization of
quantum correlations by randomized measurements, \prl
\href{https://doi.org/10.1103/PhysRevLett.131.090201}{{\bf131},
090201 (2023).}

\bibitem{Zhang2024}
L. Zhang, Matrix integrals over unitary groups: An application of
Schur-Weyl duality,
\href{https://arxiv.org/abs/1408.3782v6}{arXiv:1408.3782v6}

\bibitem{Turner2017}
J. Turner and J. Morton, A  complete set of invariants for
LU-equivalence of density operators, SIGMA
\href{https://doi.org/10.3842/SIGMA.2017.028}{{\bf13}, 028 (2017).}

\bibitem{Horn2013}
R.A. Horn and C.R. Johnson, Matrix Analysis, Cambridge University
Press (2nd Ed), (2013).

\bibitem{Kimura2003}
G. Kimura, The Bloch vector for N-level systems, \pla
\href{https://doi.org/10.1016/S0375-9601(03)00941-1}{{\bf314},
339-349 (2003).}

\end{thebibliography}
\end{document}